%% file: journallmcs.tex
\documentclass{lmcs}

\usepackage{amssymb,amsmath,amsfonts}
\usepackage[utf8]{inputenc}
\DeclareUnicodeCharacter{00A0}{ }
\usepackage{paralist}
\usepackage{xspace}
\usepackage{bbm}
\input{macros}


\begin{document}

\title{Inference from  Visible Information and Background Knowledge}
\date{}

\author[M. Benedikt]{Michael Benedikt}
\address{Department of Computer Science, Oxford University, Parks Rd, Oxford OX1 3QD, UK}
\author[P. Bourhis]{Pierre Bourhis}
\address{CNRS/CRIStAL, Parc scientifique de la Haute Borne 40, avenue Halley. Bat. B, Park Plaza 59650 Villeneuve d'Ascq}
\author[B. ten Cate]{Balder ten Cate}
\address{Google Inc., Mountainview CA}
\author[G. Puppis]{Gabriele Puppis}
\address{CNRS / LaBRI, 351 Cours de la Libération, Talence 33405, France}
\author[M. Vanden Boom]{Michael Vanden Boom}
\address{Department of Computer Science, Oxford University, Pars Rd, Oxford OX1 3QD, UK}

\begin{abstract}
\input{abstract}
\end{abstract}

\maketitle

\input{intro}

\input{organ}

\input{ack}
\input{relatedshort}

\input{def}

\input{pos}

\input{neg}

\input{combo}

\input{conc}

\bibliographystyle{alpha}
\bibliography{priv}

\appendix

\input{gnftoaut/gnftoaut}

\end{document}

%% file: macros.tex
\newcommand{\cA}{\mathcal A}
\newcommand{\powerset}[1]{\mathcal{P}(#1)}
\newcommand{\then}{\rightarrow}
\newcommand{\et}{\wedge}
\newcommand{\vel}{\vee}
\newcommand{\sat}{\models}
\newcommand{\nsat}{\not \models}
\newcommand{\bbN}{\mathbbm{N}}
\newcommand{\bbD}{\mathbbm{D}}
\newcommand{\nin}{\not \in}
\newtheorem{theorem}{Theorem}[section]
\newtheorem{proposition}[theorem]{Proposition}
\newtheorem{lemma}[theorem]{Lemma}

\newtheorem{corollary}[theorem]{Corollary}

\newtheorem{definition}[theorem]{Definition}

\newtheorem{example}{Example}
\newcommand{\canondb}{\mathsf{CanonInst}}
\newcommand{\rankcq}[1]{\operatorname{rank}_{\mathrm{CQ}}%
                        \if\relax\detokenize{#1}\relax\else(#1)\fi}
\newcommand{\sizecq}[1]{\operatorname{size}_{\mathrm{CQ}}%
                        \if\relax\detokenize{#1}\relax\else(#1)\fi}

\newcommand{\psbof}{\mathsf{PQItoGNF}}
\newcommand{\nsbof}{\mathsf{NQItoGNF}}
\newcommand{\visinst}{\mathcal{V}}
\newcommand{\fullinst}{\mathcal{F}}

\newcommand{\size}[1]{\lvert #1 \rvert}
\newcommand{\width}[1]{\operatorname{width}( #1 )}
\newcommand{\gnnf}{\text{\textup{GN}-normal form}\xspace}

\newcommand{\myparagraph}[1]{\noindent{\bf #1. }}
\newcommand{\kw}[1]{\mathsf{#1}\xspace}
\let\oldlhd\lhd
\let\oldrhd\rhd
\renewcommand{\lhd}{{\mathbin{\oldlhd}}}
\renewcommand{\rhd}{{\mathbin{\oldrhd}}}

\newcommand{\bfS}{\mathbf{S}}
\newcommand{\NSB}{\mathsf{NQI}}
\newcommand{\PSB}{\mathsf{PQI}}
\newcommand{\OWQ}{\mathsf{OWQ}}
\newcommand{\true}{\mathsf{true}}
\newcommand{\false}{\mathsf{false}}

\newcommand{\err}{\mathsf{Err}}

\newcommand{\gnfo}{\text{GNFO}}
\newcommand{\gnf}{\gnfo}
\newcommand{\el}{\mathcal{EL}}
\newcommand{\dllite}{\text{DL-LITE}}
\newcommand{\tgd}{\mathsf{TGD}}
\newcommand{\fgtgd}{\text{FGTGD}}
\newcommand{\cf}{\text{NoConst}}

\newcommand{\appoint}{\mathsf{Appointment}}
\newcommand{\patient}{\mathsf{Patient}}

\newcommand{\visible}{\mathsf{Visible}}

\newcommand{\tree}{\ensuremath{\mathcal T}}

\newcommand{\calC}{\Sigma}
\newcommand{\domain}{\mathsf{dom}}
\newcommand{\adomain}{\mathsf{adom}}

\def\ac0{\text{\sc AC$^0$}\xspace}

\def\pspace{\text{\sc PSpace}\xspace}

\def\ptime{\text{\sc PTime}\xspace}
\def\exptime{\text{\sc ExpTime}\xspace}
\def\twoexp{\text{\sc 2ExpTime}\xspace}
\def\expspace{\text{\sc ExpSpace}\xspace}

\def\conp{\text{\sc co}\mbox{-}\text{\sc NP}\xspace}
\def\np{\text{\sc NP}\xspace}

\newcommand{\chase}{\mathrm{chase}_{\mathsf{vis}}}
\newcommand{\Chase}{\mathrm{Chases}_{\mathsf{vis}}}

\newcommand{\myeat}[1]{}

\newcommand{\arity}{\mathrm{ar}}
\newcommand{\confchall}{\mathsf{ConfChallenge}}
\newcommand{\vertchallenge}{\mathsf{VertChallenge}}
\newcommand{\horchallenge}{\mathsf{HorChallenge}}
\newcommand{\hfuncchallenge}{\mathsf{HorFuncChallenge}}

\newcommand{\errhfun}{\mathsf{ErrHorFun}}

\newcommand{\hemptyerror}{\mathsf{HorEmptyError}}
\newcommand{\hemptyhiddenerror}{\mathsf{HorEmptyHiddenError}}

\newcommand{\hlabelerror}{\mathsf{HorLabelError}}
\newcommand{\hlabelhiddenerror}{\mathsf{HorLabelHiddenError}}

\newcommand{\herrorsucc}{\mathsf{HorSuccError}}
\newcommand{\herrorhiddensucc}{\mathsf{HorSuccHiddenError}}

\newcommand{\twolabel}{\mathsf{TwoLabelsError}}
\newcommand{\cerror}{\mathsf{ConstraintError}}
\newcommand{\good}{\mathsf{Good}}
\newcommand{\goal}{\mathsf{Goal}}

%% file: abstract.tex
We provide a wide-ranging study of the scenario where a subset of
the relations in a relational vocabulary are visible to a user --- that is, their complete
contents are known --- while the remaining relations are invisible.
We also have a background theory ---  invariants given by logical
sentences --- which may  relate the visible relations to invisible ones, and 
also may constrain both the visible and invisible relations in isolation. 
We want to determine whether some other information, 
given as a positive existential formula,  can be inferred using only the visible information and
the background theory. This formula whose inference we are concered with is denoted as the \emph{query}.
We consider whether positive information about
the query can be inferred, and also whether negative information  -- the sentence does not hold --
can be inferred. We further consider both the instance-level version of the problem, where both the query
and the visible instance  
are given, and the schema-level version, where we want to know whether truth or falsity of the query can be inferred
in \emph{some} instance of the schema.

%% file: intro.tex
\section{Introduction} \label{sec:intro}
This paper concerns  a setting where there is a  collection
of relations,
% representing information of interest to a set of users, 
but a given user or class of users has access to only a subset of these relations.
This could arise, for example in a database setting, where
 a data owner  restricts access to a subset of the stored relations for privacy reasons.
Another example comes from information integration,
where the integrated schema exposed to users  contains both stored relations
and ``virtual relations'',  whose  contents are not directly accessible, but which
 have a meaning defined
by logical relationships with stored relations.
Both of these scenarios can be subsumed by considering a schema consisting of a set of relations that must satisfy a 
\emph{background theory}
(invariants specified by sentences in some logic) with
only a subset of the relations visible.
A basic computational problem is to determine what questions can be answered by means of reasoning
with the background theory and access to the visible relations.
Can someone use the content of the visible relations along with
background knowledge to answer a given question about the invisible relations?

We study this scenario, where a set of 
semantically-related relations are hidden
while for another set the complete contents are visible.
We will consider background theories specified in a variety 
of logical languages that are 
rich enough to capture complex relationships between relations, including 
relationships that arise in information integration and restrictions
on a single source that have been studied in the database research community (``integrity constraints'').
The basic analysis problem will be as follows.
We are given
a relational vocabulary
partitioned into visible and invisible
relations,
 a logical sentence $Q$ (the \emph{query}, representing information whose inference we want to check),
and a background theory $\calC$, again consisting of logical sentences.
Our goal is to determine whether we can infer using 
the visible relations and the background theory some properties about the evaluation of $Q$.
We will be considering  variations of the problem in two dimensions:
\begin{itemize}
\item [{\bf Instance-level vs. Schema-level}] Can $Q$ be inferred  from $\Sigma$ and the extensions of the visible relations in a particular instance, where
the extension of visible relations are given as input to the problem?  Can $Q$ be inferred on \emph{some} instance?
\item [{\bf Positive vs. Negative}] Can it be inferred that $Q$ is true? Can it be inferred that $Q$ is false?
\end{itemize}

\begin{example} \label{ex:one}
Just in order to give intuition for the problems we study in the paper,
we give an example from logical analysis of information disclosure, in the spirit of prior works such as
\cite{securityawarevldb14}.

Consider a medical datasource with relation $\appoint(p, a, \ldots)$ containing  patient names
$p$, appointment ids $a$, and  other information about the appointment, such as the name $d$ of the doctor. 
A dataowner makes available one projection of $\appoint$ by creating a relation $\patient(p)$, defined by 
the following
logical sentences $\Sigma$:
\begin{align*}
  \forall ~p ~ \patient(p) ~\rightarrow~ \exists~ a ~ d ~ \appoint(p, a, d) 
  \\
  \forall ~p ~ a ~ d ~ \appoint(p, a, d) ~\rightarrow~ \patient(p) \ . \ 
\end{align*}
The query $Q= \exists ~a ~ \appoint(\mbox{``Smith''}, a, \mbox{``Jones''})$ 
asking whether patient Smith made an appointment with Dr. Jones can not be inferred under this 
schema in one sense:
an external user with access to $\patient$ will never be sure that the query is true, in
any instance. We say that there can be no Positive Query Implication on any instance
for this query, schema, and background theory. But suppose we consider whether a user can infer
the query to be false?
On many instances, such an inference is not possible.
But on instances where the visible relation $\patient$ does not contain the patient name Smith,  
an external user will know that the query is false.
We say that there is a Negative Query Implication on the visible instances
where $\patient$ does not contain Smith.
\end{example}

\myparagraph{Our results}
As mentioned above, we will consider the instance-based problems: given a query and instance, 
can a user determine that the query is true (Positive Query Implication) 
or that the query is false (Negative Query Implication)? 
We also look at the corresponding schema-level 
problem: given a query and a schema, is there some instance where 
a query implication of one of the above types occurs?

We start by observing that the instance-level problems, both positive and negative, are decidable for a very
rich logical language for background
theories, using the same technique: a reduction to the guarded negation fragment of first-order logic (see below).
However, when we analyze the complexity of the decision problem as the size of the 
instance increases,
we see surprisingly different behavior between the positive and negative case.
For very simple background theories %, such as inclusion dependencies, 
the negative query implication
problems are %very 
well-behaved as the 
instance changes, namely, in polynomial time and definable within a well-behaved logic.
For the same class of background theories, the corresponding positive query implication questions are 
hard even when the schema and query are fixed.  Our most significant hardness result is
that even for simple background theories
the positive query inference problem is $\exptime$-hard in \emph{data complexity}:
that is, when everything except the visible instance is fixed. This is a big jump
in complexity for the complexity for special cases of the  problem studied in the description
logic \cite{lutzhybrid2} and database community \cite{abiteboulduschka} in the past.

When we turn to the schema-level problems, even decidability is not obvious. 
We prove a set of ``critical instance''
results, showing that whenever there is an instance where information about the query can be implied,
the ``obvious instance'' witnesses this.
Thus, schema-level problems reduce to special cases of the instance-level problems.  
Although we  use this technique to obtain decidability and complexity
results both for positive and for negative query implication, the classes of background knowledge to which
they apply are different. 
We give undecidability results that    show that  when the  classes are even slightly enlarged,
decidability of the existence of an 
instance %schema 
with a query implication is lost.

\smallskip
\myparagraph{Our techniques}
%In the process, 
We make use of a number of tools for 
reasoning on  %use in querying of 
mixtures of complete and incomplete information.
\begin{itemize}
\item {\bf Connection to Guarded negation.} 
      Our first technique involves showing that a large class of instance-level problems
      can be solved by translating them into satisfiability problems for a rich fragment 
      of first-order logic, the guarded negation fragment ($\gnf$).
In fact, we show that there is a natural
connection between these inference problems and  $\gnf$, in that  the ``visibility restriction''
can be expressed in $\gnf$.
      This allows us to exploit powerful prior decidability results for $\gnf$ ``off-the-shelf''. 
      However, to get tight complexity bounds, 
      we  need a new analysis of the complexity of satisfiability for $\gnf$. This analysis
is of interest outside of these inference problems, in that we give a self-contained
reduction from $\gnf$ satisfiability to tree automata,  a reduction which allows us to give a finer-grained
analysis of the sources
of complexity in $\gnf$ satisfiability.
\item {\bf Decidability via canonical counterexamples.}   
      The schema-level analysis asks if there is some instance on which 
      information about the query can be derived. As mentioned above, we show that whenever there is 
      some instance, this can be taken to be the ``simplest possible instance''. While this 
      idea has been used before to simplify  analysis of undecidability 
      (e.g.~\cite{criticalinst1}), and for decidability of Datalog satisfiability  \cite{schmueliundecid},
we provide a significant extension of the technique, and provide new applications of it
       for decidability. 
\item {\bf Tractability via greatest fixed-point.} 
      We show that some of our 
      instance-level implication problems  
      can be reduced to evaluating a certain query of \emph{greatest fixedpoint Datalog} (GFP-Datalog) 
      on the given visible instance. Since
      GFP-Datalog queries can be evaluated in  polynomial time, this shows tractability 
      in the instance size. 
      The reduction to GFP-Datalog requires a new analysis of  when these inference 
      problems are ``active-domain controllable'' (it suffices to see that the query 
      value is invariant over all  hidden instances that lie within the active domain 
      of the visible instance).
\item {\bf Relationships between problems.} 
      In the paper we explain how the 4 inference problems we consider
(combinations of positive/negative and instance-level/schema-level)
differ from previously-studied problems, such as the ``open
world query answering problem''.
However, we also provide reductions  between open world querying
and some of our schema-level problems.
      In addition to clarifying the relationship of the problems,
we can use these reductions to derive 
      complexity bounds.
%\item {\bf Coding techniques for lower bounds.}
%We introduce methods for coding computation in query implication problems. Among the most
%both undecidability results and complexity lower bounds.
\end{itemize}
%In addition to the techniques above, our lower bound results 
%involve a number of techniques for coding  computation in query inference problems.

%\ifdefined\longversion
%\else
%Full proofs are available in the technical report \url{www.cs.ox.ac.uk/michael.benedikt/papers/implication.pdf}.
%\fi

%% file: organ.tex
\myparagraph{Organization} After a review of related work in Section \ref{sec:related},
we formally define the problem in Section \ref{sec:def}.
Section \ref{sec:positive} presents our results on whether we can infer
the truth of a CQ or UCQ: the problems $\PSB$ and $\exists \PSB$.
Section \ref{sec:negative} turns to the problems $\NSB$ and $\exists \NSB$, concerning
inferring the negation of a CQ or UCQ. Section \ref{sec:extspec} deals with some
small extensions of the framework, and some special cases of the problems of particular interest.
We close in Section \ref{sec:conc} with conclusions.

%% file: ack.tex
\myparagraph{Acknowledgements}
This is a long version of the extended abstract that appeared in \cite{lics16}.
We are quite grateful to the referees of LICS for their helpful comments.

Benedikt's work was sponsored by the Engineering and Physical Sciences
Research Council of the United Kingdom, grants EP/M005852/1
and EP/L012138/1. 
Bourhis was supported by CPER Nord-Pas de Calais/FEDER DATA Advanced data science 
and technologies 2015-2020 and ANR  Aggreg project  ANR-14-CE25-0017.

%% file: relatedshort.tex
\section{Related Work} \label{sec:related}
Two different communities have studied 
the  problem  of determining which information can be inferred from
complete access to data in a subset of the relations, %in a relational schema   
using background knowledge in the form of logical sentences relating the subset to the full vocabulary.

%\michael{There are a large number of DB papers mentioning information leakage -- e.g. \cite{authorizationviews}}
%\michael{Also should be some relationship between NSB and consistent query answering}

In the database community, the focus has been  on 
%``information leakage'' from 
views. The schema is divided into the ``base tables'' and ``view tables'', with the
latter being defined by queries (typically conjunctive queries) in terms of the former.
Given a query over the schema, the basic computational problem is determining which answers
can be inferred using only the values of the views. 
  Abiteboul and Duschka \cite{abiteboulduschka} isolate the complexity of this problem in the
case where views are defined by conjunctive queries; in their terminology, it is ``querying under the Closed World Assumption'',
emphasizing the fact that the possible worlds revealed by the views are those where the
view tables have exactly their visible content.
In our terminology, this corresponds
exactly to the  ``Positive Query Implication'' ($\PSB$) problem
in the case where the  background theory  consists entirely of  conjunctive query view definitions.
Chirkova and Yu \cite{chirkovayu} extend to the case where conjunctive query views are supplemented
by weakly acyclic dependencies. 
Another subcase of $\PSB$ that has received considerable attention is the case where the 
background theory consists  \emph{only} of ``completeness assertions''
between the invisible and visible portions of the schema. 
A series of papers by Fan and Geerts 
\cite{floriswenfeiconf,floriswenfeijournal1}
isolate the complexity for several variations of the problem, with particular attention to the case where the completeness assertions are via inclusion dependencies from
 the invisible to the visible part.

The $\PSB$ problem we study in the first part of this work is also related to research on
instance-based determinacy (see in particular the results of Howe et al. in \cite{pricing}) while
the  ``Negative Query Implication'' ($\NSB$) problem in the second half of the paper is examined in the view context 
by Mendelzon and Zhang \cite{authorizationviews}, under the name of ``conditional emptiness''.  
As in the other work mentioned above, the emphasis  has been on view definitions rather than
more general background knowledge which may restrict both the visible and invisible instance.
In contrast, in our work
we deal with  logical languages for the background theory that can restrict the visible and invisible data in ways
incomparable to view definitions  (see also the comparison in Section \ref{sec:extspec}).

%The database literature also includes numerous works on policies for \emph{enforcing} privacy
%(e.g. \cite{securityawarevldb14}). 

%maybe ADD REFERENCE TO MENDELZON on AUTHORIZATION VIEWS

In the description logic community, the emphasis has not been on views, but on querying incomplete
information in the presence of a logical theory.
Our positive query implication problems 
relate to work
in the description logic community on \emph{hybrid closed and open world query answering} or \emph{DBoxes}, in which
the schema is divided into closed-world and open-world relations.
Given a Boolean CQ, we want to find out if it 
holds
in all  instances that can add facts to the open-world relations but do not change the closed-world
relations. In the non-Boolean case, the generalization is to consider which tuples from the initial instance
are in the query answer on all such instances.
Thus closed-world and open-world relations match our notion of visible and invisible, and the hybrid closed and open
world query answering  problem matches our notion of positive query implication, except that we restrict to the
case where the open-world/visible relations of the instance are empty. It is easy to see that
this restriction is actually without loss of generality:
one can reduce the general case to the case we study with a simple linear time reduction, making a closed-world copy $R'$ of
 each open-world relation $R$, and adding an inclusion dependency from $R'$ to $R$. 
As with the database community, the main distinction between our study of the  Positive Query Implication problem  and the prior work in the DL community concerns
the classes of background theories considered.
Lutz et al. \cite{lutzetalmixing,carstenfrank2,lutzhybrid2}  study the complexity of this problem for background knowledge
for several description logics.
For example, for the description logics  $\el$ and $\dllite$ they provide a dichotomy between $\conp$-hard and first-order 
rewritable theories. 
They also show that in all the tractable cases, the problem coincides with the classical 
open-world query answering problem.
Franconi et al. \cite{FranconiIS11} show $\conp$-completeness for a disjunction-free description logic.
Our results on the data complexity of $\PSB$ consider  the same problem, but for background theories
that are more expressive and, in particular, can handle relations of arbitrary arity, 
rather than arity at most $2$ as in \cite{lutzhybrid2,carstenfrank2,FranconiIS11}.

In summary, both the database and DL communities considered the 
$\PSB$ questions addressed in this paper,
but for background theories that are different from those we consider.
The Negative Query Implication problems are not well-studied in the prior literature,
and we know of no work at all dealing with the  schema-level questions (asking for the existence  
of an instance with a query implication) in prior work.  
However, in this paper we show (see Subsection \ref{subsec:negexists})
that there is a close relation between these schema-level questions 
and the works of Lutz et al.~that concern conservativity and 
modularity of ontologies \cite{conservextdl,conserv2}. 

Note that our schema-level analysis considers the existence of \emph{some} instance 
where the query result can be inferred. The converse problem
is to determine  whether the query result can be inferred
on \emph{all} instances. This is exactly the problem of
\emph{determinacy} \cite{NSV}, which is closely related
to the notion of  \emph{implicit definability} in classical
logic \cite{beth}. Determinacy has been extensively
studied for both views \cite{NSV,redspider} and for background
theories and visible relations \cite{ustods,thebook}.

Another contrast is to the work of Miklau and 
Suciu \cite{miklausuciu} considers whether such an inference is valid probabilistically,
looking asymptotically at the uniform distribution over models of increasing size. 

Recently \cite{aaai17} analyzed the complexity
of query implication in the presence of information disclosure methods based on query answering interfaces --- where
an external user can query under the certain answer semantics ---  rather than the model of disclosure based
on exporting a subset of the data, as in our setting.  The analysis in \cite{aaai17} builds on the techniques presented in
this paper.
%We do not deal with probabilistic modelling in this work.

%\balder{``ideas'' sounds vague -- do we mean view determinacy?}
%michael: eliminated mention of db lit

%% file: def.tex
\section{Definitions} \label{sec:def}

We consider partitioned schemas (or simply, schemas)
$\bfS=\bfS_h \cup \bfS_v$, where the partition elements $\bfS_h$ and $\bfS_v$ are finite
sets of relation names (or simply, relations), each with an associated arity.
%michael: tried to make clear that ``schema'' includes the partition into visible and invisible
These are the \emph{hidden} and \emph{visible} relations, respectively.
An \emph{instance} of a schema maps each relation to a set of tuples of the associated arity.
Instances will be used as inputs to the computational problems that are the focus of this work 
-- in this case the instances must be finite. Our computational problems also quantify over 
instances, and  they are also  well-defined when the quantification is over all (finite or
infinite) instances. 
For simplicity, by default \emph{instances are always finite}. However, as we will show, taking any of the quantification over all instances will never impact our results, and this will
allow us to make use of infinite instances freely in our proofs.
The \emph{active domain} of an instance is the set of values occurring 
within the interpretation of some relation in the instance.

As a suggestive notation, we write $\visinst$ (Visible) for instances over $\bfS_v$ and $\fullinst$ 
(Full) for instances over $\bfS$.
Given an instance $\fullinst$ for $\bfS$, its restriction to the $\bfS_v$ relations will be referred to
as its \emph{visible part}, denoted $\visible(\fullinst)$.
%We will familiarity with basic notions of first-order logic with the active-domain semantics \cite{AHV}.
%In this work,  we will only be interested in logical formulas which are \emph{domain independent} -- those
%whose truth value on a relational structure (i.e. an instance and domain extending the active domain)
%is independent of the domain. With this restriction, a formula over a schema
%$\bfS$ has a well-defined set of satisfying tuples in any instance of $\bfS$. \looseness=-1

We will look at background theories defined in a number of logics.
One class of logical sentences that we will focus on
are  Tuple-generating Dependencies (TGDs)m
which are first-order logic sentences of the form
$$
  \forall \bar x ~ \phi(\bar x) ~\rightarrow~ \exists \bar y ~ \rho(\bar x, \bar y)
$$
where $\phi$ and $\rho$ are conjunctions of atoms, which may contain variables and/or constants,
and where all the universally quantified variables $\bar x$ appear in $\phi(\bar x)$.
For all the problems considered in this work, one can take w.l.o.g. the right-hand side $\rho$ to consist of a single atom, and we will assume this henceforth. 
We will often omit the universal quantifiers, writing just $\phi(\bar x) ~\rightarrow~ \exists \bar y ~ \rho(\bar x, \bar y)$.
The main feature of TGDs we will exploit is the lack of disjunction, which will allow for cleaner
characterizations of our query inference problems.

For TGDs we will be able to obtain clean semantic characterizations
for our inference problems.  But most  inference problems involving TGDs are undecidable \cite{AHV}, including
all those we study here.
Thus for our decidability and complexity results we will look at  classes of TGDs that are computationally better behaved:
\begin{compactitem}
\item \emph{Linear  TGDs}: those where $\phi$ consists of a single atom.
\item \emph{Inclusion Dependencies} (IDs), linear TGDs where each of $\phi$ 
      and $\rho$ have no constants and no repeated variables. These correspond to traditional referential integrity constraints in databases.
\item Many of our results on inclusion dependencies will hold for two more general classes.
\emph{Frontier-guarded TGDs} ($\fgtgd$s) \cite{frontier} are TGDs where one of the conjuncts of $\phi$ 
    is an atom that includes every universally quantified variable $x_i$ occurring in $\rho$.
\emph{Connected TGDs} require only that the \emph{co-occurrence graph}  of $\phi$ is connected. 
      The nodes of this graph  are the variables
      $\bar x$, and  variables are connected by an edge if they co-occur in an atom of $\phi$.
\end{compactitem}
Note that  every ID  is a linear TGD, and every linear TGD is frontier-guarded.
We will also consider two logical languages that are generalizations of $\fgtgd$s.
\begin{compactitem}
\item We allow disjunction, by considering \emph{Disjunctive Frontier-guarded TGDs}, which are of the form
      $$
        \forall \bar x ~ \phi(\bar x) ~\rightarrow~ 
        \exists \bar y ~ \bigvee\nolimits_{\!\!i} ~ \rho_i(\bar x, \bar y)
      $$
%      where each $\rho_i$ is a conjunction of atoms and there is one atom  conjoined in $\phi$ that includes every variable $x_i$ included in some $\rho_i$.
      where, for each $i$, $\rho_i$ is a conjunction of atoms and there is an atom in $\phi$ 
      that includes all the variables $x_j$ occurring in $\rho_i$.

\item A key role will be played by an even richer logic, one containing Disjunctive $\fgtgd$s ,
      the \emph{Guarded Negation Fragment}.
      $\gnfo$ is built up inductively according to the grammar:
      \begin{align*}
        \phi ~::=~ & R(\bar t) ~~|~~ t_1=t_2 ~~|~~ \exists x ~ \phi ~~|~~
                     \phi \vee \phi ~~|~~ \phi \wedge \phi ~~|~~  \\
                   & R(\bar t, \bar y) \wedge \neg \phi(\bar y)
      \end{align*}
      where $R$ is either a relation symbol or the equality relation $x=y$,
      and the $t_i$ represent either variables or  constants.
      Notice that any use of negation must occur conjoined with an atomic 
      relation that contains all the free variables of the negated formula 
      -- such an atomic relation is a \emph{guard} of the formula.
In database terms, $\gnfo$ is  equivalent to relational algebra where \emph{the difference operator can only be used to subtract
query results from a relation}. The VLDB paper \cite{bbo} gives both  Relational algebra and SQL-based syntax for $\gnfo$, and
argues that it covers useful queries and database integrity constraints in practice.
\end{compactitem}
%  The purpose of allowing equalities as guards is to ensure that every formula
%  with at most one free variable can be considered guarded.
For simplicity (so that all of our sentences are well-defined on instances) %, we  maintain our focus on instances,  
we will always assume that  our $\gnfo$ formulas are domain-independent;
to enforce this we can use the relational algebra  syntax  for capturing these queries, mentioned above.

For many of the results in the paper, the reader  only needs to know a few facts about   $\gnfo$.
The first is that it is quite expressive, so in proving things about $\gnfo$ 
sentences we immediately
get the results for many classes of theories that we have mentioned above.
$\gnfo$ contains every positive existential formula, is closed under
Boolean combinations of sentences, and it
subsumes disjunctive frontier-guarded TGDs up to equivalence.
That is, by simply writing out a disjunctive frontier-guarded TGD 
using $\exists, \neg, \wedge$, one sees that these are expressible in $\gnfo$.

Secondly, we will use that $\gnfo$ is ``tame'', encapsulated in the following result
from \cite{gnficalp}:

\begin{theorem}[\cite{gnficalp}]\label{thm:gnfsat}
Satisfiability for $\gnfo$ sentences can be tested  effectively, and is %in fact
$\twoexp$-complete. Furthermore, every satisfiable sentence has a finite satisfying model.
\end{theorem}

Note that $\gnf$ does \emph{not} subsume the theories corresponding to CQ view definitions
(e.g. $A(x,y) \wedge B(y,z) \leftrightarrow V(x,z)$ cannot be expressed in GNFO).
However we will cover this special class of theories in Section \ref{sec:extspec}.

Finally, we will consider \emph{Equality-generating Dependencies} (EGDs), 
of the form
$$
  \forall \bar x ~ \phi(\bar x) ~\rightarrow~ x_i=x_j
$$
where $\phi$ is a conjunction of atoms and $x_i, x_j$ are variables.
As with TGDs, EGDs generalize some well-known relational database integrity constraints, such as functional 
dependencies and key constraints. \emph{EGDs with constants} further allow 
equalities between variables and constants, e.g. $x_i = a$, in the right-hand side.

Our problems take as input a background theory and
also a logical sentence whose inference  we want to study,
the query. In this work we will  consider queries specified as 
\emph{conjunctive queries} (CQs), first-order 
formulas built up from relational atoms via conjunction and existential quantification (equivalently,
relational algebra queries built via selection, projection, join, and rename operations), 
and also \emph{unions of CQs} (UCQs), which are disjunctions (relational algebra 
unions) %UNIONs) 
of CQs.
\emph{Boolean UCQs} are simply UCQs with no free variables.
Every CQ $Q$ is associated with a \emph{canonical instance} $\canondb(Q)$, 
where the domain consists of variables and constants of $Q$ and the facts 
are the atoms of $Q$.
%A  Boolean CQ is satisfied in an instance $I$ iff there is a homomorphism 
%from $\canondb(Q)$ to $I$.

We will always assume that we have associated with each value a corresponding constant,
and we will identify the constant with its value. 
Thus distinct constants will always be forced to denote distinct domain elements 
-- this is often called the ``unique name assumption'' (UNA) \cite{AHV}.
While the presence or absence of constants will often make no difference in our results, 
there are several problems where their presence adds significant complications. 
In contrast, it is easy to show that the presence of constants without the UNA will 
never make any difference in any of our results.
\emph{Note that in our background theories and query languages above, with the exception 
of IDs, constants are allowed by default.}
When we want to restrict to formulas without constants, we add the prefix $\cf$;
for example, $\cf$-$\fgtgd$ denotes the frontier-guarded TGDs that do not contain constants.

The crucial definition for our work is the following:

\begin{definition}
Let $Q$ be a Boolean UCQ over schema $\bfS$, $\calC$ a logical theory over $\bfS$, 
and $\visinst$ an instance over a visible schema $\bfS_v\subseteq\bfS$.
\begin{compactitem}
\item $\PSB(Q,\calC, \bfS, \visinst) = \true$ if for every finite instance 
      $\fullinst$ satisfying $\calC$, if $\visinst=\visible(\fullinst)$
      then $Q(\fullinst)=\true$.
\item $\NSB(Q,\calC, \bfS, \visinst) = \true$  if for every finite instance 
      $\fullinst$ satisfying $\calC$, if $\visinst=\visible(\fullinst)$ 
      then $Q(\fullinst)=\false$.
\end{compactitem}
\end{definition}

% We extend the definition of $\PSB$ and $\NSB$ to arbitrary instances  $\visinst$ by allowing
% the output to be arbitrary when $\visinst$ is not realizable.

We call an $\bfS_v$-instance $\visinst$ \emph{realizable} w.r.t. $\calC$ if 
there is an $\bfS$-instance $\fullinst$ satisfying $\calC$ such that 
$\visinst=\visible(\fullinst)$.
If an instance $\visinst$ is not realizable w.r.t. $\calC$, then, trivially,
$\PSB(Q,\calC,\bfS, \visinst)=\NSB(Q, \calC, \bfS, \visinst)=\true$.
%
% In giving our complexity results for $\PSB(Q,\calC,\bfS, \visinst)$ and $\NSB(Q, \calC, \bfS, \visinst)$, we will
% assume that the input $\visinst$ is a realizable $\bfS_v$-instance. 
%  That is, when we say that  $\PSB(Q, \calC, \bfS, \visinst)$ is in complexity class $X$ as
% $\calC$ ranges over  collection
% of constraints $\calC$, we will mean that there is a function in class $X$ that agrees
% with $\PSB(Q, \calC, \bfS, \visinst)$ on realizable $Q$, and similarly for $\NSB$.
In practice, realizable instances
are  the only $\bfS_v$-instances we should ever
encounter. For simplicity we state our instance-level results for the $\PSB$ and $\NSB$ problems
that take as input an arbitrary instance of $\bfS_v$.
But since our lower bound arguments  will only involve realizable instances, an alternative definition that assumes realizable inputs
 yields the same complexity bounds.

%\myparagraph{Finite and unrestricted instances}
%\gabriele{Here we should make a bit more clear that the standards semantics is 
%          with infinite models, but that the two options are in fact irrelevant 
%          (I would prefer avoiding talking of weak and strong versions, since
%          the problems are a priori not reducible one to the other)}
$\PSB(Q, \calC, \bfS,\visinst)$ states something about every finite instance, 
in line with our default assumption that instances are finite.
We can also talk about an ``unrestricted version'' where the quantification 
is over every (finite or infinite) instance.
\emph{For the logical sentences we deal with for our background theories, there will be no difference 
between these notions.}
That is, we will show that the  finite and unrestricted versions of $\PSB$
coincide for a given class of arguments $Q, \calC, \bfS, \visinst$. We express this by saying
that ``$\PSB(Q,\calC, \bfS, \visinst)$ is \emph{finitely controllable}'', and similarly for $\NSB$.

We need a definition of the size of the input.
In our case, an input consists of 
a query $Q$, a set of sentences $\calC$, a relational schema $\bfS$, and 
an instance $\visinst$, and the size is defined by taking the length of the 
binary encoding of such objects. Other intuitive notions of size (e.g. number of symbols)
would also suffice for our results, since they differ from the bit-encoding notion only up to a polynomial
factor.

Often we will be interested in studying the behavior of the $\PSB$ and $\NSB$ 
problems when $Q$, $\calC$, and $\bfS$ are fixed, e.g.~looking at how the computation 
time varies in the size of $\visinst$ only.  
We refer to this as the \emph{data complexity} of the $\PSB$ (resp.~$\NSB$) problem.

The $\PSB$ problem contrasts with the usual \emph{Open-World Query Answering} 
or \emph{Certain Answer} problem, denoted here $\OWQ(Q, \calC, \fullinst)$,  
which is studied extensively in databases and description logics. 
The latter problem takes as input a Boolean query $Q$, an instance $I$, and 
a set of sentences $\calC$, and returns $\true$ iff the query holds in any 
finite instance $I'$ containing all facts of $I$. In $\PSB$ (and $\NSB$) we 
further constrain the instance to be fixed on the visible part while requiring 
the invisible part of the input instance to be empty. This is the mix of ``Closed World'' and ``Open World'',
and we will see that this 
Closed World restriction can make the complexity significantly higher.

\begin{example} \label{ex:psbvscertain}
Consider a scenario where the background theory consists of inclusion dependencies
$F_1(x) \rightarrow \exists y ~ U(x,y)$ and $U(x,y) \rightarrow F_2(y)$. In the schema, the relations
$F_1$ and $F_2$ are visible but $U$ is not.
Consider the query
$Q= \exists x ~ U(x,x)$ and instance consisting only of facts  $F_1(a), F_2(a)$.

There is a $\PSB$ on this instance, since $F_1(a)$ implies that $U(a,c)$ holds for some
$c$, but the other constraint and the fact that $F_2$ must hold only on $a$ 
means that $c=a$, and hence $Q$ holds.

In contrast, one can easily see that $Q$ is not certain  in the usual sense, where $F_1$ and $F_2$ can
be freely extended with additional facts.
\end{example}

%\michael{Maybe move definition of Open World Query Answering here, and give a one-sentence distinction, saying that it is not
%the same as PSB,  but that there will still turn out to be  a close relationship}

Our schema-level problems concern determining if there is a realizable instance 
that admits a query implication:

\begin{definition}
For $Q$  a Boolean conjunctive query over schema $\bfS$, 
and $\calC$ a set of sentences over $\bfS$, we let:
\begin{compactitem}
\item $\exists\PSB(Q, \calC, \bfS)=\true$  if there is a realizable 
      $\bfS_v$-instance $\visinst$ such that $\PSB(Q,\calC,\bfS, \visinst)=\true$;
\item $\exists\NSB(Q, \calC, \bfS)=\true$   if there is a realizable 
      $\bfS_v$-instance $\visinst$ such that $\NSB(Q,\calC,\bfS, \visinst)=\true$.
\end{compactitem}
\end{definition}

Note that these problems now quantify over instances twice,
and hence there are alternatives  depending on whether the instance
$\visinst$ is restricted to be finite, and whether the hidden instances 
$\fullinst$ are restricted to be finite. For a class of input $Q, \calC, \bfS$,
we say that ``$\exists\PSB(Q, \calC, \bfS)$ is \emph{finitely controllable}'' if 
in both quantifications, quantification over finite instances can be freely 
replaced with quantification over arbitrary instances without changing
the truth value of the statement.

%% file: pos.tex
\section{Positive Query Implication} \label{sec:positive}

\input{posinstance}

\input{posexists}

\input{possummary}

%% file: posinstance.tex
\subsection{Instance-level problems} \label{subsec:posinst}

Here we study the problem 
$\PSB(Q,\calC, \bfS, \visinst)$. Recall that this  asks whether $Q(\fullinst)=\true$ 
for every full instance $\fullinst$  satisfying $\calC$ which agrees with 
$\visinst$ in the visible part.
The section is organized in three parts: in the first part we prove upper bounds
for the $\PSB$ problem, establishing a connection to the Guarded Negation  Fragment.
In the second part we present  a technique tailored to background theories of Horn logic,
showing that instances
witnessing the failure of  $\PSB$ can be taken to be tree-like.
In the third part we use this technique
to prove tight lower bounds for the instance-level $\PSB$ problem.

\medskip
\myparagraph{Upper bounds and the connection to Guarded Negation}
We begin by showing that $\PSB$ is decidable when background
theories are in the  logic $\gnfo$,  
the guarded negation fragment. This is interesting first of all since
$\gnfo$ is a very expressive logic. It  subsumes the other decidable logics
that we consider here, such as  guarded TGDs, disjunctive guarded TGDs, 
and Boolean combinations of Boolean CQs. Further, it highlights 
the fact
that  $\gnfo$ suffices to capture the fact that an instance has a particular restriction to the
visible relations. This is
exploited  in
in the following reduction 
to the satisfiability problem for $\gnfo$:

\begin{theorem}\label{thm:gnfoposinstancedecid}
The problem $\PSB(Q,\calC,\bfS,\visinst)$, as $Q$ ranges over  Boolean UCQs and 
$\calC$ over $\gnfo$ sentences, is in $\twoexp$.
\par\noindent
Furthermore, for such sentences the problem is finitely controllable,
that is, $\PSB(Q,\calC,\bfS,\visinst)=\true$ iff for every 
instance $\fullinst$ (of any size) satisfying $\calC$, if 
$\visinst=\visible(\fullinst)$, then $Q(\fullinst)=\true$. 
\end{theorem}

\begin{proof}
One easily sees that $\PSB(Q,\calC,\bfS,\visinst)$ translates to unsatisfiability of the following formula:
%$\nsbof(Q, \calC, \bfS, \visinst)$ defined as:
$$
\begin{aligned}
  \phi^\psbof_{Q,\calC,\bfS,\visinst} ~=~
  & \neg Q ~\wedge~ \calC ~\wedge~ \\[-1ex]
  & \!\!\!\!\!
    \bigwedge_{R\in \bfS_v} \!\! \Big( \!\!\bigwedge_{R(\bar{a})\in\visinst}\!\!\!\!\! R(\bar{a}) 
                                  ~~\wedge~~ 
                                  \forall \bar{x} ~ \big( R(\bar{x}) \rightarrow \!\!\!
                                                          \bigvee_{R(\bar{a})\in\visinst} \!\!\!\!
                                                          \bar{x}=\bar{a} \big) \Big)
\end{aligned}
$$
Intuitively, the formulas requires that the instance on which it is evaluated (which includes
visible and hidden relations) satisfies the background theory, but not the query, and in addition 
the visible part of the instance agrees with $\visinst$. Note that the formula has size
linear in the inputs to $\PSB$, and thus this gives a polynomial time reduction.

If the sentences in the background theory are in $\gnfo$, then the formula above is also in $\gnfo$.
Indeed, the only places where negation is used, either explicitly or implicitly, 
are $\neg Q$, which is guarded since $Q$ has no free variables, and the universal 
quantification $\forall \bar x ~ (R(\bar x) \rightarrow \dots)$, which translates to 
$\neg \exists \bar x ~ (R(\bar x) \wedge \neg \dots)$, 
with the inner negation guarded by $R(\bar x)$ and the outer negation 
involving no free variables.

The finite controllability of $\PSB(Q,\calC,\bfS,\visinst)$ 
comes from the finite controllability of $\gnfo$ formulas (Theorem~\ref{thm:gnfsat}).
\end{proof}

Above we are using results on satisfiability of $\gnfo$ as a ``black-box''. 
Satisfiability tests for $\gnfo$ work by translating a satisfiability problem 
for a formula into a tree automaton which must be tested for non-emptiness.
By a finer analysis of this translation of $\gnfo$ formulas to automata, 
we can see that the \emph{data complexity} 
of the problem is only singly-exponential.

\begin{theorem} \label{thm:expdatacomplexitypsb}
If $Q$ is a Boolean UCQ and $\calC$ is a conjunction of $\gnfo$ sentences
over a schema $\bfS$, then the data complexity of $\PSB(Q,\calC,\bfS,\visinst)$ 
(that is, as $\visinst$ varies over instances) is in $\exptime$.
\end{theorem}

\begin{proof}[Sketch]
In the body of the paper, we provide a proof outline for this. What we  omit is a fine-grained analysis
of the translation of  $\gnf$ to automata, extending the translation to automata found in 
\cite{uslics15bounded}.
The reader interested in this conversion can find the details in the appendix.

We start by stating a satisfiability result for $\gnf$ formulas $\phi$ in a normal form,
 \emph{\gnnf}, similar to one
introduced in \cite{gnfj}.
%We assume for now that the formulas do not use equality or constants,

Formulas  in  \gnnf
can be generated using the following grammar:
\begin{align*}
\phi ::= \
&{\textstyle \bigvee_i \exists \vec{x} ~ \bigwedge_j \psi_{ij}} \\
\psi ::= \,
&\alpha
\, \mid \,
\alpha \wedge \phi
, \mid \,
\alpha \wedge \neg \phi
\end{align*}
where
$\alpha$ is an atomic formula
and free variables of $\phi$ are contained in free
variables of $\alpha$.
As with $\gnf$, in the second production rule we also allow  $\alpha$ to be omitted
if $\phi$
 has at most one free variable (thus allowing free negation of such formulas).
The $\phi$ are referred to as \emph{UCQ-shaped formulas}, with
each of the disjuncts being a \emph{CQ-shaped formula}.   UCQ-shaped formulas
are only used to define the normal form and the related notion of CQ-rank below.
They are clearly as expressive as general $\gnfo$ formulas.

Note that if $\phi_i$ for $i=1 \ldots n$ are \emph{sentences} in normal form then
their conjunction $\bigwedge_i \phi_i$ is also in normal form.

The \emph{width} of $\phi$, denoted $\width{\phi}$,
is the maximum number of free variables of any subformula of $\phi$.

The \emph{CQ-rank} of a formula $\phi$ in $\gnnf$,
denoted $\rankcq{\phi}$, is the maximum number of conjuncts $\psi_i$
in any CQ-shaped subformula
$\exists \vec{x} ~ \bigwedge_i \psi_i$ of $\phi$
for non-empty $\vec{x}$.
For the purposes of CQ-rank,
$\alpha(\vec{x}) \wedge \neg \phi(\vec{x})$
and $\alpha \wedge \phi$
are treated as  subformulas with $1$ conjunct.

\begin{theorem}\label{thm:gnfsatrefined}
For every fixed numbers $r$, $m$, and $w$, there is an $\exptime$
algorithm that determines whether
 a $\gnf$ formula $\phi$  in \gnnf over a schema with relations
of arity at most $m$, with  $\rankcq{\phi} \leq r$ and $\width{\phi} \leq w$ is satisfiable.
\end{theorem}

Theorem \ref{thm:gnfsatrefined} is proven by creating an 
alternating two-way parity automaton whose state set consists of a collection
of formulas derived from $\phi$.  The automaton runs on a tree whose nodes represent
collections of elements in a tree-like model. 
If the formula $\phi$ were in the guarded fragment, rather than $\gnf$,
it would suffice to use the subformulas of $\phi$
as states, where the subformulas would have additional annotations associating
variables with elements of a guarded set. The bound on the arity would suffice to keep the number
of annotations low. In the presence of CQ-shaped formulas, the vertices will not be associated
with a guarded set, but with a set whose size is controlled by $\width{\phi}$. Thus
by bounding $\width{\phi}$ we keep the number of annotations low. A further  problem
is that for CQ-shaped subformulas, one will
have to throw in all subformulas, representing guesses as to which of the conjuncts were true of
the elements associated to a given node of a tree-like structure.  The bound on $\rankcq{\phi}$ guarantees that this need to throw in subformulas does
not blow up the number of states. 

It is important for our application that the result
applies to $\gnf$ formulas that have equality and constants, which are treated 
by adding  additional cases for equality atoms in the automata, and conjoining with an additional
automata that enforces that the facts involving constants are consistent across the tree.
The details of this, as well as other subtleties in the proof of Theorem \ref{thm:gnfsatrefined}, are given in the appendix.

Now fix a Boolean UCQ $Q$ and a conjunction $\calC$ of $\gnf$ sentences over a schema $\bfS$. 
Without loss of generality, we can assume that the sentences in $\calC$ are already in \gnnf.
Consider the formula $\phi^{\kw{PSBtoGNF}}_{Q, \calC, \bfS, \visinst}$ in the proof of Theorem 
\ref{thm:gnfoposinstancedecid}:
$$
\begin{aligned}
    \neg Q ~\wedge~ \calC ~\wedge~ %\\[1ex] 
    \!\!\!
    \bigwedge_{R\in \bfS_v} \!\! \Big( \!\!\bigwedge_{R(\bar{a})\in\visinst}\!\!\!\!\! R(\bar{a}) 
                                  \:\wedge\: 
                                  \forall \bar{x} ~ \big( R(\bar{x}) \rightarrow \!\!\!
                                                          \bigvee_{R(\bar{a})\in\visinst} \!\!\!\!
                                                          \bar{x}=\bar{a} \big) \Big) \ .
\end{aligned}
$$
This formula can be rewritten to eliminate the universally-quantified implication, replacing 
this subformula with the negation
of the sentence
\begin{align*}
 \exists \bar{x} ~  R(\bar{x}) \wedge  \bigwedge_{R(\bar{a})\in\visinst} \bigvee_i x_i \neq a_i
\end{align*}

We can add equality guards on the formulas $x_i \neq a_i$, and 
guards of the form $R(\vec x)$ on the disjunctions $\bigvee_i x_i \neq a_i$.
With these changes, which do not impact the size of the formula,
  the conditions of  \gnnf are satisfied.

Thus the formula $\phi^{\kw{PSBtoGNF}}_{Q, \calC, \bfS, \visinst}$ 
can be normalized in polynomial time, and the schema arity and $\rankcq{}$ of 
$\phi^{\kw{PSBtoGNF}}_{Q, \calC, \bfS, \visinst}$ %$\psbof(Q, \calC, \bfS, \visinst)$
are fixed when $Q$, $\calC$, and $\bfS$ are fixed.
Applying Theorem \ref{thm:gnfsatrefined}
%auttrans}, we get a polynomial-sized two-way alternating automaton. 
%Since emptiness of such automata can be checked in $\exptime$ 
%(see M. Vardi \emph{Reasoning about The Past with Two-Way Automata} in ICALP 98)
%\cite{Vardi98},
the bound claimed in Theorem \ref{thm:expdatacomplexitypsb} now follows.
\end{proof}

\medskip
\myparagraph{A characterization of  $\PSB$ for Horn logics}
We have shown above that $\PSB$ can be reduced to satisfiability of a $\gnf$ formula,
and it is known that a satisfiable $\gnf$ formula can always be taken
to be ``tree-like'' --- indeed, this is what allows automata-theoretic techniques to be applied.
We can give a more concrete algorithm in the special case of  background theories in a certain
family related to the Horn fragment of first order logic; specifically
for EGDs and TGDs.
This will not get us  better worst-case upper bounds for the cases we consider in this work: indeed, for general TGDs and EGDs it is not even effective. But it will prove useful
for  showing stronger lower bounds on the combined and data complexity 
of $\PSB$, since it will allow us to show them when the background theories are in
these restricted  classes. It will also be essential for the schema-level problems considered
in Section \ref{sec:existsPSB}.

We show that for TGDs and EGDs, the $\PSB$ problem can be characterized using a variant of the chase procedure
\cite{onet,fagindataex}. 
Our procedure
receives as input a relational schema $\bfS$, a background theory $\calC$
consisting of TGDs and EGDs, and an initial instance $\fullinst_0$ for the schema $\bfS$, 
which does not need to satisfy  the background theory $\calC$.
The goal of the procedure is to produce a collection of instances (not necessarily finite)
that satisfy  $\calC$, extend the initial instance $\fullinst_0$, and 
agree with this instance on the visible part.
The goal is achieved by repeatedly adding new facts to the initial instance $\fullinst_0$
so as to satisfy the sentences in $\calC$, in a way similar to the classical chase procedure for TGDs.
However, non-deterministic choices are sometimes needed to map the newly generated 
tuples in a visible relation to some existing facts in $\fullinst_0$. Our technique
is actually a variant of the ``disjunctive chase'' of \cite{chaserevisited}, which produces
multiple instances. 

We now describe in detail how the variant for the chase procedure works, since we will need it
in the remainder of the paper. We start with an explanation in the case
where $\calC$ contains only  TGDs, and later extend it to handle EGDs.

Recall that w.l.o.g.~TGDs are assumed to have exactly one atom in the right-hand side.
Due to the unique name assumption (UNA), functions between domain elements are 
tacitly assumed to preserve all the constants that appear in the sentences of the 
background theory
and in the query (i.e.~$h(a)=a$ whenever $a$ appears as a constant in $\calC$ or $Q$).
As usual, such functions are homomorphically extended to relational instances 
(i.e.~$h(R(x_1,\dots,x_n))=R(h(x_1),\dots,h(x_n))$ for all relations $R$).
The procedure builds a \emph{chase tree} of instances, starting with 
the singleton tree consisting of the input $\bfS$-instance $\fullinst_0$ and 
extending the tree by repeatedly applying the following steps.
It chooses an instance $K$ at some leaf of the current tree, a 
TGD $R_1(\bar x_1) \wedge \ldots \wedge R_m(\bar x_m) \rightarrow \exists \bar y ~ S(\bar z)$ in 
$\calC$, where $\bar z$ is a sequence of (possibly repeated) variables from 
$\bar x_1,\ldots,\bar x_m,\bar y$, and a homomorphism $f$ that maps 
$R_1(\bar x_1)$, \dots, $R_m(\bar x_m)$ to some facts in $K$. 
Then, the procedure constructs a new instance from $K$ by adding the 
fact $S(f'(\bar z))$, where $f'$ is an extension of $f$ that maps, in an injective
way, the existentially quantified variables in $\bar y$ to some  values that are not in
$K$.
In the usual terminology of the chase, such an added  value is called
a ``null'', and adding this fact is called ``performing a  chase step''.
Immediately after this step, and only when the relation $S$ is visible, the procedure 
replaces the instance $K'=K\cup\big\{S(f'(\bar z))\big\}$ with copies of it
of the form $g(K')$ such that $\visible(g(K')) = \visible(\fullinst_0)$,
for all possible homomorphisms $g$ that map the variables $\bar z$ 
to some values in the active domain $\{a_1,\ldots,a_n\}$ of 
the visible instance $\visible(\fullinst_0)$. Note that the active domain
of $\{a_1,\ldots,a_n\}$ of $\visible(\fullinst_0)$ does not contain null values.
In the language of prior papers on the chase \cite{chaserevisited},
this step would be a sequence of ``disjunctive chase step'',
for disjunctive EGDs of the form 
$S(\bar z) \rightarrow  z_i = a_1 \vee \ldots \vee z_i = a_n$).
The resulting instances $g(K')$ are then appended as new children of 
$K$ in the tree-shaped collection. 
In the special case where there are no homomorphisms $g$ such that
$\visible(g(K')) = \visible(\fullinst_0)$, we append a ``dummy instance''
$\bot$ as a child of $K$: this is used to represent the fact that 
the chase step from $K$ led to an inconsistency (the dummy node will 
never be extended during the subsequent chase steps).
If $S$ is not visible, then the instance $K'$ is simply appended as 
a new child of $K$.

This process continues iteratively using a strategy that is ``fair'',
namely, that guarantees that whenever a dependency is applicable in 
a node on a maximal path of the chase tree, then it will be fired
at some node (possibly later) on that same maximal path (unless the 
path ends with $\bot$).
In the limit, the process generates a possibly infinite tree-shaped collection
of instances. It remains to complete the collection with ``limits'' in order to 
guarantee that the sentences in the background theory are satisfied. Consider any infinite 
path $K_0,K_1,\ldots$ in the tree (if there are any). It follows from 
the construction of the chase tree that the instances on the path form a 
chain of homomorphic embeddings $K_0 \xrightarrow{h_0} K_1 \xrightarrow{h_1} \ldots$.
Such chains of homomorphic embeddings admit a natural notion of limit, 
which we denote by $\lim_{n\in\bbN}K_n$. We omit the details of this construction
here, which can be found, for instance, in \cite{ChangKeisler}.
%  that contains all the facts of the form $R(t)$ 
% such that $R(t)\in K_n$ for all but finitely many $n\in\bbN$. 
The limit instance $\lim_{n\in\bbN}K_n$ satisfies the background theory $\calC$.
%\balder{I've changed this -- the limit construction is more complicated as it 
%involves a quotient so we don't want to explain it here}
We denote by 
$\Chase(\calC,\bfS,\fullinst_0)$ 
the collection of all non-dummy instances that occur at the leaves of the chase tree, 
plus all limit instances of the form $\lim_{n\in\bbN}K_n$, where $K_0,K_1,\ldots$ is 
an infinite path in the chase tree.  
This is well-defined only once the ordering of steps is chosen, but for the results 
below, which order is chosen will not matter, so we abuse notation by referring to 
$\Chase(\calC,\bfS,\fullinst_0)$ 
as a single object.

We now indicate how we can modify the  procedure
for $\Chase(\calC,\bfS,\visinst)$ so as to take into account also the EGDs
in $\calC$ that can be triggered on the instances that emerge in the chase tree.
Formally, chasing an EGD of the form
$R_1(\bar x_1)\wedge\ldots\wedge R_m(\bar x_m) \rightarrow x=x'$,
where $x,x'$ are two variables from $\bar x_1,\ldots,\bar x_m$,
amounts at applying a suitable homomorphism that identifies the
two values $h(x)$ and $h(x')$ whenever the facts
$R_1(h(\bar x_1)), \ldots, R_m(h(\bar x_m))$
belong to the instance under consideration.
Note that this operation leads to a failure (i.e.~a dummy instance)
when $h(x)$ and $h(x')$ are distinct values from the active
domain of the visible part $\visinst$.

It is clear that every instance in $\Chase(\calC,\bfS,\fullinst_0)$, except the special
``failure instance'', 
satisfies the sentences in $\calC$ and, in addition, agrees with 
$\fullinst_0$ on the visible part of the schema.
Below, we prove that $\Chase(\calC,\bfS,\fullinst_0)$ 
satisfies the following %universal 
property:

\begin{lemma}\label{lem:disjunctive-chase-universal}
Let $\calC$ consist of EGDs and TGDs without constants.
Let $\fullinst_0$ be an instance of a schema $\bfS$ and let $\fullinst$ 
be another instance over the same schema that contains $\fullinst_0$, 
agrees with $\fullinst_0$ on the visible part 
(i.e. $\visible(\fullinst) = \visible(\fullinst_0)$),
and satisfies all sentences of $\calC$.
Then, there exist an instance $K \in \Chase(\calC,\bfS,\fullinst_0)$ 
and a homomorphism from $K$ to $\fullinst$.
\end{lemma}

\begin{proof}
For brevity we prove the result for TGDs only.
We consider the chase tree for $\Chase(\calC,\bfS,\fullinst_0)$ and, 
based on the full instance $\fullinst$, we identify inside this chase tree 
a suitable path $K_0,K_1,\ldots$ and a corresponding sequence of homomorphisms $h_0,h_1,\ldots$ 
such that, for all $n\in\bbN$, $h_n$ maps $K_n$ to $\fullinst$. % and $h_{n+1}$ extends $h_n$. 
Once these sequences are defined, the lemma will follow easily by 
letting $K=\lim_{n\in\bbN}K_n$ and $h=\lim_{n\in\bbN}h_n$, that is, 
$h(\bar a) = \bar b$ if $h_n(\bar a)=\bar b$ for all but finitely
many $n\in\bbN$.

The base step is easy, as we simply let $K_0$ be the initial instance $\fullinst_0$,
which appears at the root of the chase tree, and let $h_0$ be the identity.
As for the inductive step, suppose that $K_n$ and $h_n$ are defined for some 
step $n$, and suppose that 
$R_1(\bar x_1)\wedge\ldots\wedge  R_m(\bar x_m) \rightarrow \exists \bar y ~ S(\bar z)$ 
is the dependency that is applied at node $K_n$, where $\bar z$ is a sequence 
of variables from $\bar x_1,\ldots,\bar x_m,\bar y$.
Let $R_1\big(f(\bar x_1)\big)$, \dots, $R_m\big(f(\bar x_m)\big)$ be the facts 
in the instance $K_n$ that have triggered the chase step, where $f$ is a 
homomorphism from the variables in $\bar x_1,\ldots,\bar x_m$ to the domain 
of $K_n$.
Since $\fullinst$ satisfies the same dependency and contains the facts 
$R_1\big(h_n(f(\bar x_1))\big)$, \dots, $R_m\big(h_n(f(\bar x_m))\big)$, 
it must also contain a fact of the form
$S\big(h'(f'(\bar z))\big)$, where $f'$ is the extension of $f$ that is
the identity on the existentially quantified variables $\bar y$ and 
$h'$ is some extension of $h_n$ that maps the variables $\bar y$ to some 
values in the domain of $\fullinst$.

Now, to choose the next instance $K_{n+1}$, we distinguish two cases, depending
on whether $S$ is visible or not. If $S$ is not visible, then we know that the
chase step appends a single instance $K'=K_n\cup\big\{S(f'(\bar z))\big\}$
as a child of $K_n$; accordingly, we let $K_{n+1}=K'$ and $h_{n+1}=h'\circ f'$.
Otherwise, if $S$ is visible, then we observe that $h'$ is a homomorphism from
$K'=K_n\cup\big\{S(f'(\bar z))\big\}$ to $\fullinst$. In particular, $h'$ maps 
the variables $\bar z$ to some values in the active domain of the visible part 
$\visible(\fullinst_0)$ and hence $h'(K')$ agrees with $\fullinst_0$ on the 
visible part of the schema. 
This implies that the chase step adds at least the instance $h'(K')$ as a child 
of $K_n$. Accordingly, we can define $K_{n+1}=h'(K')$ and $h_{n+1}=f'$.
Given the above constructions, it is easy to see that the homomorphism $h_{n+1}$
maps $K_{n+1}$ to $\fullinst$.

Proceeding in this way, we either arrive at a leaf, and in this case we are done,
 or we obtain an infinite 
path of the chase tree $K_0 \xrightarrow{h_0} K_1 \xrightarrow{h_1} \ldots$, with homomorphisms
$h'_i:K_i\to\fullinst$, such that $h_i\circ h'_{i+1}$ extends $h'_i$, for all $i\in\bbN$. 
In the latter case it can be shown that the limit $\lim_{n\in\bbN} K_n$ also homomorphically 
maps to $\fullinst$. 
\end{proof}

The following proposition characterizes the instances of the $\PSB$ problem 
when the sentences in the background theory are TGDs without constants:

\begin{proposition}\label{prop:PSB-characterizion}
If $Q$ is a Boolean UCQ, $\calC$ is a set of TGDs or EGDs without constants
over a schema $\bfS$, and $\visinst$ is a visible instance, then
$\PSB(Q,\calC,\bfS,\visinst)=\true$ iff every instance $K$ in
$\Chase(\calC,\bfS,\visinst)$ satisfies $Q$.
\end{proposition}

\begin{proof}
Suppose that $\PSB(Q,\calC,\bfS,\visinst)=\true$ and recall that every instance in
$\Chase(\calC,\bfS,\visinst)$
satisfies the sentences in $\calC$ and agrees with $\visinst$ on the visible part.
In particular, this means that every instance in
$\Chase(\calC,\bfS,\visinst)$
satisfies the query $Q$.

Conversely, suppose that $\PSB(Q,\calC,\bfS,\visinst)=\false$.
This means that there is an $\bfS$-instance $\fullinst$ that has
$\visinst$ as visible part, satisfies the sentences in $\calC$,
but not the query $Q$.
By Lemma \ref{lem:disjunctive-chase-universal}, letting $\fullinst_0=\visinst$,
we get an instance
%$K \in \Chase_\visinst(\calC,\bfS,\visinst)$
$K \in \Chase(\calC,\bfS,\visinst)$
and a homomorphism from $K$ to $\fullinst$. Since $Q$ is preserved
under homomorphisms, $K$ does not satisfy $Q$.
\end{proof}

\medskip
\myparagraph{Lower bounds}
Below we show that the data complexity bound in Theorem \ref{thm:expdatacomplexitypsb} 
is tight even for inclusion dependencies (IDs). The proof proceeds by showing that a ``universal machine''
for alternating $\pspace$ can be constructed by fixing appropriate $Q, \calC, \bfS$ in a $\PSB$ problem.

\begin{theorem}\label{th:psb-instance-hard}
There are a Boolean CQ $Q$ and a set $\calC$ of IDs over a schema $\bfS$
for which the problem $\PSB(Q,\calC,\bfS,\visinst)$ is $\exptime$-hard 
in data complexity.
\end{theorem}

\begin{proof}
%\gabriele{Super-simplified proof version, no tree-like}
We first prove the hardness result using a UCQ $Q$; later, we show how to 
generalize this to a CQ.
We reduce the acceptance  problem for an alternating $\pspace$ Turing machine $M$ 
to the negation of $\PSB(Q,\calC,\bfS,\visinst)$.

A configuration of $M$ is defined, as usual, by a control state, a position of the 
head on the tape, and a finite string representing the content on the tape.
The input of the machine is assumed to be a string of blanks $\sqcup\dots\sqcup$ 
(thus only its length matters). Moreover, special symbols $\vdash,\dashv$ are added 
at the extremities of the input to mark the endpoints of the working tape.
%In fact, for technical reasons (namely, to ease detecting badly formed encodings of computations), 
%it is convenient to repeat twice each marker, and assume that the head of $M$ never visits 
%the most external markers, and never removes a marker. 
%Accordingly, the initial configuration of $M$ has tape content 
%of the form $\vdash\vdash\sqcup\dots\sqcup\dashv\dashv$ and the
%head on the second occurrence of $\vdash$.
Accordingly, the initial configuration of $M$ has tape content 
of the form $\vdash\sqcup\dots\sqcup\dashv$ and the
head on the first position.

The transition function of $M$ describes a set of target configurations on 
the basis of the current configuration.
We distinguish between existential and universal control states of $M$, 
and we assume that there is a strict alternation between existential and 
universal states along every sequence of transitions.
Without loss of generality, we also assume that there are exactly $2$ target 
configurations for each transition that departs from a universal state.
A computation of $M$ is thus represented by a tree of configurations, where 
the root represents the initial configuration and every node with an 
existential (resp., universal) control state has exactly one (resp., two)
successor configuration(s).
Furthermore, to make the coding simpler, we adopt a non-standard acceptance condition.
Specifically, we assume that the Turing machine $M$ never halts, namely,
its transition function is defined on every configuration, and we distinguish
two special control states, $q_{\kw{acc}}$ and $q_{\kw{rej}}$. 
We further assume that every infinite path in a computation tree of $M$ 
eventually reaches a configuration with either $q_{\kw{acc}}$ or $q_{\kw{rej}}$ 
as control state, and from there onwards there is no change of configuration. 
Accordingly, we say that $M$ accepts (its input) if it admits a computation
tree where the state $q_{\kw{acc}}$ appears on all paths; symmetrically, 
we say that $M$ rejects if every computation tree has a path leading to
$q_{\kw{rej}}$.

The general idea of the reduction is to create schema, background theory, and query that together
represent a ``universal machine'' for alternating $\pspace$. Then, given an alternating 
$\pspace$ machine $M$ encoded in the visible instance, an accepting computation tree of $M$
will be encoded by an arbitrary full instance that satisfies the background theory and violates 
the query --- that is, a witness of the failure of $\PSB$.
We first devise the schema with hidden relations that will store the computation tree of 
a generic alternating $\pspace$ machine.
The background theory and (the negation of) the query will be used to restrict the hidden relations 
so as to guarantee that the encoding of the computation tree is correct.
By ``generic'' we mean that the hidden relations and corresponding
background theory will be independent of the tape size, number of 
control states, and transition function of the machine.
The visible instance will store the ``representation'' of an alternating $\pspace$ machine $M$
--- that is, an encoding of $M$ that can be calculated  efficiently once  
$M$ is known. This will include the  tape size and an encoding of the transition function. 
We will then give the reduction that takes an alternating polynomial space machine $M$
and instantiates all the  visible relations with the encoding.
The space bound on $M$  will allow us to create
the tape components in the visible instance efficiently. In contrast, the hidden relations
will store aspects of a computation that can \emph{not} be computed easily from $M$.
In summary, below we will be describe each part of the schema $\bfS$ for computation
trees of a machine, along with the polynomial mapping that transforms a machine $M$ 
into data filling up the visible parts of the schema.

\smallskip
To begin with, we explain how to encode the tape (devoid of its content) into a binary relation $T$.
The relation $T$ will be visible, and can be filled efficiently once the length of the tape of $M$ is known.
Given $M$, it will be filled  in the following natural way: 
it contains all the facts $T(y,y')$, where $y$ is the identifier of a cell 
and $y'$ is the identifier of the successor of this cell in the tape.
Recall that the Turing machine $M$ works on a tape of polynomial length,
and hence the visible instance for the relation $T$ has also size polynomial 
in $M$.
We also add unary visible relations $\kw{First}$ and $\kw{Last}$, 
that are intended to distinguish the first and last cells of the tape.
Given $M$, we will instantiate $\kw{First}$ (resp., $\kw{Last}$) 
with the singleton consisting of the identifier of the first (resp., last) cell.
Moreover, despite the fact that the tape length is finite, it is convenient to 
assume that every cell has a successor --- this assumption will be exploited 
later to ease the instantiation of new tape contents for each configuration. 
We will thus add to the visible relation $T$ also the ``dummy'' pair $(y,y)$, 
where $y$ is the identifier of the rightmost cell of the tape. 

As for the configurations of the machine, these are described
by specifying, for each configuration and each tape cell,
a suitable value that represents the content of that cell,
together with the information on whether the Turing machine
has its head on the cell, to the right, or to the left,
and what is the corresponding control state. 
%For a technical reason (specifically, to allow detecting 
%violations of the transition rules between pairs of subsequent 
%configurations), we adjoin to the labelling of a cell also that 
%of the adjacent cells whenever the head is within the neighbourhood.
Formally, the configurations of the machine are encoded by a hidden 
ternary relation $C$, where each fact $C(x,y,z)$ indicates 
that, in the configuration identified by $x$, the cell $y$ 
has value $z$.  
We will enforce that the cell values range over an appropriate 
domain, defined by a visible unary relation $V$. 
%As with all of our visible relations, we can fill $V$
%easily once we have a specific machine $M$.
%In our reduction from $M$, we will fill this relation $V$ with
%$(\Sigma_\lhd \times \Sigma_Q \times \Sigma_\rhd) \uplus
% (\Sigma_Q \times \Sigma_\rhd) \uplus
% (\Sigma_\lhd \times \Sigma_Q) \uplus
% \Sigma_\lhd \uplus \Sigma_\rhd$,
%where $\Sigma$ is the tape alphabet of $M$,
%$Q$ is the set of its control states, 
%$\lhd,\rhd$ are fresh symbols, and 
%$\Sigma_\lhd = \Sigma\times\{\lhd\}$, 
%$\Sigma_\rhd = \Sigma\times\{\rhd\}$, 
%$\Sigma_Q = \Sigma\times Q$.
%When a cell has value $\big((a,\lhd),(b,q),(c,\rhd)\big)$, this 
%means that its content is $b$, the Turing machine stores the 
%control state $q$, the head is precisely on this cell, and 
%the neighbouring cells to the left and to the right have labels 
%$a$ and $c$, respectively.
%Similarly, when a cell has value $\big((a,q),(b,\rhd)\big)$ 
%(resp., $\big((b,\lhd),(c,q)\big)$), 
%this means that its content is $b$, the Turing machine stores 
%the control state $q$, the head is on the predecessor
%(resp., successor) cell, and this latter cell carries the letter $a$ (resp., $c$).
%In all other cases, we simply store the content $b$ of the cell
%together with the information of whether the head is far to the 
%right or far to the left.
In our reduction from $M$, we will fill this relation $V$ with
$\Sigma_Q \uplus \Sigma_\lhd \uplus \Sigma_\rhd$,
where $\Sigma$ is the tape alphabet of $M$ (which includes
the markers $\vdash$ and $\dashv$),
$\Sigma_Q = \Sigma\times Q$,
$\Sigma_\lhd = \Sigma\times\{\lhd\}$, 
$\Sigma_\rhd = \Sigma\times\{\rhd\}$, 
$Q$ is the set of its control states, 
and $\lhd,\rhd$ are fresh symbols.
When a cell has value $(a,q)$, this means that its content is $a$,
the Turing machine stores the control state $q$, and the head is 
precisely on this cell. 
Similarly, when a cell has value $(a,\lhd)$ (resp., $(a,\rhd)$), 
this means that its content is $a$ and the cell is to the immediate left 
(resp., immediate right) with respect to the position of the head of the 
Turing machine.

Because we need to associate the same tape structure with several different configurations,
the content of the relations $T$ and $\kw{First}$ will end up being replicated within 
new hidden relations $T^C$ and $\kw{First}^C$, where it will be paired with the 
identifier of a configuration.
For example, a fact $T^C(x,y,y')$ will indicate that, in the configuration 
identified by $x$, the cell $y$ precedes the cell $y'$. 
Similarly, a fact $\kw{First}^C(x,y)$ will indicate that $y$ 
is the first cell of the tape of configuration $x$.
Of course, we will enforce the condition that the relations 
$T^C$ and $\kw{First}^C$, devoid of the first attribute, 
are contained in $T$ and $\kw{First}$, respectively.

We now turn to the encoding of the computation tree.
For this, we introduce a visible unary relation $I$ that contains the 
identifier of the initial configuration. We also introduce the hidden binary 
relations $S^\exists$, $S^\forall_1$, and $S^\forall_2$. 
We recall that every configuration $x$ with an existential control state has exactly one
successor $x'$ in the computation tree, so we represent this with the fact $S^\exists(x,x')$.
Symmetrically, every configuration $x$ with a universal control state has exactly two
successors $x_1$ and $x_2$ in the computation tree, and we represent this with
the facts $S^\forall_1(x,x_1)$ and $S^\forall_2(x,x_2)$.

So far, we have introduced the visible relations $T$, $\kw{First}$, $\kw{Last}$,
$V$, $I$, and the hidden relations $C$, $T^C$, $\kw{First}^C$, 
$S^\exists$, $S^\forall_1$, $S^\forall_2$.
These are sufficient to store an encoding of the computation tree of the machine.
However, the background theories are only allowed to contain inclusion dependencies, 
which are not powerful enough to guarantee that these relations indeed represent 
a correct encoding. To overcome this problem, we will later introduce a few additional 
relations and exploit a union of CQs to detect those violations of the background
theory
that are not captured by inclusion dependencies.

\smallskip
We now list some inclusion dependencies in $\calC$ that enforce basic 
restrictions on the relations.
\begin{itemize}
  \item We begin with some sentences that guarantee that the relations 
        $T$ and $T^C$ induce the same ``successor'' relation on the cells of the tape:
        $$
        \begin{array}{rclrcl}
          T^C(x,y,y')         &\rightarrow& T(y,y')        & \qquad
          \kw{First}^C(x,y)   &\then& \exists ~ y' ~ T^C(x,y,y') \\[0.5ex]
          \kw{First}^C(x,y)   &\then& \kw{First}(y)  & \qquad
          T^C(x,y,y')         &\then& \exists ~ y'' ~ T^C(x,y',y'') \ .
        \end{array}
        $$
        Note that, while we can easily enforce that $T$ contains the projection of 
        $T^C$ onto the last two attributes, and similar for $\kw{First}$ and $\kw{First}^C$.
        It is more difficult, instead, to enforce that $T^C$ contains copies of $T$ 
        annotated with each configuration identifier.
        This is done indirectly by requiring that every tuple $(x,y)$ in $\kw{First}^C$ 
        is the source of an infinite chain of successors inside $T^C$, all annotated 
        with the same configuration identifier. 
        Paired with the previous sentences, this will guarantee that $T^C$ contains 
        the annotated copy $\{x\}\times T$. 
        Further note that, for this to work, it is crucial to have assumed that 
        there is a ``dummy'' successor $T(y,y)$ on the last tape cell $y$.
        The existence of facts of the form $\kw{First}^C(x,y)$ for 
        each configuration $x$ will be enforced later.
  \item We proceed by enforcing the existence of values associated with each cell in each configuration:
        $$
        \begin{array}{rclrcl}
          T^C(x,y,y') &\then& \exists ~ z ~ C(x,y,z)  & \qquad\quad~ 
          C(x,y,z)    &\then& V(z) \ . 
          \qquad\qquad~~\ 
        \end{array}
        $$
        Note that the sentences in the background theory defined so far may allow a cell to be associated with 
        multiple values. We will show later how to detect this case using a suitable query.
  \item We finally enforce a graph structure representing the evolution of the configurations, 
        assuming that the machine starts with the existential configuration 
        contained in the visible relation $I$:
        $$
        \begin{array}{rcl}
          I(x)               &\then& \exists ~ x' ~ S^\exists(x,x') \\[0.5ex]
          S^\exists(x,x')    &\then& \exists ~ x_1 ~ S^\forall_1(x',x_1) \\[0.5ex]
          S^\exists(x,x')    &\then& \exists ~ x_2 ~ S^\forall_2(x',x_2) \\[0.5ex]
          S^\forall_1(x,x_1) &\then& \exists ~ x' ~ S^\exists(x_1,x') \\[0.5ex]
          S^\forall_2(x,x_2) &\then& \exists ~ x' ~ S^\exists(x_2,x')
        \end{array}
        \qquad
        \begin{array}{rcl}
          S^\exists(x,x')    &\then& \exists ~ y ~ \kw{First}^C(x,y) \\[2.5ex]
          S^\forall_1(x,x_1) &\then& \exists ~ y ~ \kw{First}^C(x,y) \\[2.5ex]
          S^\forall_2(x,x_2) &\then& \exists ~ y ~ \kw{First}^C(x,y) \ .
        \end{array}
        $$
        Note that the rules on the right side above trigger the creation of a first tape cell
        for each configuration, which in turn spawns copies of the entire tape.
\end{itemize}
Next, we explain how to detect badly-formed encodings of the computation tree. 
For this, we use additional visible relations 
$\err_C$, 
$\err_{I,\kw{first}}$, $\err_{I,\kw{last}}$, $\err_{I,\kw{adj}}$, 
$\err_{C,\kw{adj}}$, $\err_{S^\exists}$, $\err_{S^\forall_1}$, 
and $\err_{S^\forall_2}$, instantiated as follows.
\begin{itemize}
  \item The relation $\err_C$ is binary and contains all pairs
        of {\sl distinct} cell values from $V\times V$.
        This is used to check that every cell, in every 
        configuration, is associated with at most one value.
        The CQ below holds precisely when this latter property is violated:
        $$
          Q_C ~=~ \exists ~ x ~ y ~ z ~ z' ~~
                  C(x,y,z) ~\et~ C(x,y,z') ~\et~ \err_C(z,z') \ .
        $$
%  \item The relation $\err_{I,\kw{first}}$ is ternary and
%        contains all triples of values that cannot be associated 
%        with the first three cells in the initial configuration
%        (recall that the first two cells carry the symbol $\vdash$,
%         the third cell carries the symbol $\sqcup$, and $M$ starts 
%         with state $q_0$ on the second cell).
%        Formally, $\err_{I,\kw{first}}$ contains all the triples 
%        in $V\times V\times V$ except $(z_1,z_2,z_3)$, where 
%        $z_1=\big((\vdash,\lhd),(\vdash,q_0)\big)$, 
%        $z_2=\big((\vdash,\lhd),(\vdash,q_0),(\sqcup,\rhd)\big)$, 
%        and $z_3=\big((\vdash,q_0),(\sqcup,\rhd)\big)$.
%        Accordingly, we can detect whether the values of the first 
%        three cells in the initial configuration are badly-formed 
%        using the following CQ:
%        $$
%        \begin{aligned}
%          Q_{I,\kw{first}} ~=& ~ \exists ~ x ~ y ~ y' ~ y'' ~ z ~ z' ~ z'' \\
%                             & ~ I(x) ~\et~ C(x,y,z) ~\et~ C(x,y',z') ~\et~ C(x,y'',z'') ~\et~ \err_{I,\kw{first}}(z,z',z'') \\
%                             & ~ T_{\kw{first}}(y) ~\et~ T(y,y') ~\et~ T(y,y',y'') \ .
%        \end{aligned}
%        $$        
  \item The relation $\err_{I,\kw{first}}$ is also binary, and
        contains all pairs of values that cannot be associated 
        with the first two cells in the initial configuration
        (recall that the first two cells carry the symbols $\vdash$
         and $\sqcup$, and $M$ starts with state $q_0$ on the first cell).
        Formally, $\err_{I,\kw{first}}$ contains all the pairs
        in $V\times V$ except $(z_0,z_1)$, where 
        $z_0=(\vdash,q_0)$ and $z_1=(\sqcup,\rhd)$.
        Accordingly, we can detect whether the values of the first 
        two cells in the initial configuration are badly-formed 
        using the following CQ:
        $$
        \quad
        \begin{aligned}
          Q_{I,\kw{first}} ~= 
          & ~ \exists ~ x ~ y ~ y' ~ z ~ z' ~~ \\
          & ~ I(x) ~\et~ \kw{First}(y) ~\et~ T(y,y') ~\et~ C(x,y,z) ~\et~ C(x,y',z') ~\et~ \err_{I,\kw{first}}(z,z') \ .
        \end{aligned}
        $$        
%  \item The relation $\err_{I,\kw{last}}$ is also ternary, and
%        contains triples of values that cannot be associated 
%        with the last three cells in the initial configuration,
%        i.e., $\err_{I,\kw{last}} = (V\times V\times V) \setminus (z_{-2},z_{-1},z_{-1})$, 
%        where $z_{-2}=(\sqcup,\rhd)$ and $z_{-1}=(\dashv,\rhd)$.
%        We can detect whether the last three values in the initial configuration
%        are inconsistent with the CQ
%        $$
%        \begin{aligned}
%          Q_{I,\kw{last}} ~=& ~ \exists ~ x ~ y ~ y' ~ y'' ~ z ~ z' ~ z'' \\
%                            & ~ I(x) ~\et~ C(x,y,z) ~\et~ C(x,y',z') ~\et~ C(x,y'',z'') ~\et~ \err_{I,\kw{last}}(z,z',z'') \\
%                            & ~ T(y,y') ~\et~ T(y',y'') ~\et~ T_{\kw{last}}(y'') \ .
%        \end{aligned}
%        $$        
  \item Similarly, the relation $\err_{I,\kw{last}}$ contains pairs of values 
        that cannot be associated with the last two cells in the initial configuration,
        i.e., $\err_{I,\kw{last}} = (V\times V) \setminus (z_1,z_{-1})$, 
        where $z_1=(\sqcup,\rhd)$ is defined as before and $z_{-1}=(\dashv,\rhd)$.
        We can detect whether the last two values in the initial configuration
        are inconsistent using the CQ
        $$
        \quad
        \begin{aligned}
          Q_{I,\kw{last}} ~=
          & ~ \exists ~ x ~ y ~ y' ~ z ~ z' \\
          & ~ I(x) ~\et~ T(y,y') ~\et~ \kw{Last}(y') ~\et~ C(x,y,z) ~\et~ C(x,y',z') ~\et~ \err_{I,\kw{last}}(z,z') \ .
        \end{aligned}
        $$        
  \item The relation $\err_{I,\kw{adj}}$ contains pairs of 
        values that cannot appear on any two consecutive cells
        of the initial configuration, namely, 
        $\err_{I,\kw{adj}}$ contains all the pairs in $V\times V$, 
        but the following ones: $(z_0,z_1)$, $(z_1,z_1)$, $(z_1,z_{-1})$.
        This type of violation is checked with the CQ
        $$
          \qquad
          Q_{I,\kw{adj}} 
          ~=~ \exists ~ x ~ y ~ y' ~ z ~ z' ~~ 
              I(x) ~\et~ T(y,y') ~\et~ C(x,y,z) ~\et~ C(x,y',z') ~\et~ \err_{I,\kw{adj}}(z,z') \ .
        $$
  \item In a similar way we can check violations of labellings of consecutive cells 
        in every configuration. This is done with the binary visible relation 
        $\err_{C,\kw{adj}}$, instantiated with all pairs from $V\times V$ that 
        cannot be adjacent in an arbitrary configuration (for example, the pair
        $\big((a,\lhd),(b,\rhd)\big)$), and the CQ
        $$
          Q_{C,\kw{adj}} 
          ~=~ \exists ~ x ~ y ~ y' ~ z ~ z' ~~ T(y,y') ~\et~ C(x,y,z) ~\et~ C(x,y',z') ~\et~ \err_{C,\kw{adj}}(z,z') \ .
        $$
%  \item The relation $\err_{S^\exists}$ is used to check consistency along a transition
%        that departs from an existential configuration. 
%        It contains all pairs of cell values from $V\times V$ that cannot appear on the same
%        position of the tape at an existential configuration and its immediate successor
%        (of course, this relation depends on the transition function of the Turing machine).
%        A violation of the corresponding constraint can be exposed by the following CQ:
%        $$
%          Q_{S^\exists} ~=~ \exists ~ x ~ x' ~ y ~ z ~ z' ~~ 
%                            S^\exists(x,x') ~\et~ C(x,y,z) ~\et~ C(x',y,z') ~\et~ 
%                            \err_{S^\exists}(z,z') \ .
%        $$
  \item The relation $\err_{S^\exists}$ is used to check consistency along a transition
        that departs from an existential configuration. 
        It contains a quadruple of cell values $(z,z',z'',z''') \in V\times V\times V\times V$ 
        whenever it is {\sl not} possible to have an existential configuration where the labels
        $z,z',z''$ appear on three consecutive positions $y,y',y''$, together with
        a successor configuration that carries value $z'''$ at position $y'$.
        Of course, the content of this relation depends on the transition function of 
        the Turing machine.
        A violation of the corresponding constraint is exposed by the following CQ:
        $$
        \begin{aligned}
          Q_{S^\exists} ~=
          & ~ \exists ~ x ~ x' ~ y ~ y' ~ y'' ~ z ~ z' ~ z'' ~ z''' \\
          & ~ S^\exists(x,x') ~\et~ T(y,y') ~\et~ T(y',y'') ~\et~ \\
          & ~ C(x,y,z) ~\et~ C(x,y',z') ~\et~ C(x,y'',z'') ~\et~ C(x',y',z''') ~\et~ \\
          & ~ \err_{S^\exists}(z,z',z'',z''') \ .
        \end{aligned}
        $$
%  \item Similarly, the relation $\err_{S^\forall_1}$ (resp., $\err_{S^\forall_2}$)
%        contains those pairs of values that cannot appear on the same 
%        position of the tape of a universal configuration and that of 
%        the first (resp., second) successor. The corresponding CQs are
%        $$
%        \begin{aligned}
%          Q_{S^\forall_1} &~=~ \exists ~ x ~ x_1 ~ y ~ z ~ z' ~~ 
%                               S^\forall_1(x,x_1) \,\et\, C(x,y,z) \,\et\, C(x_1,y,z') \,\et\,
%                               \err_{S^\forall_1}(z,z')  \\[1ex]
%          Q_{S^\forall_2} &~=~ \exists ~ x ~ x_2 ~ y ~ z ~ z' ~~ 
%                               S^\forall_2(x,x_2) \,\et\, C(x,y,z) \,\et\, C(x_2,y,z') \,\et\,
%                               \err_{S^\forall_2}(z,z') \ .
%        \end{aligned}
%        $$
  \item Similarly, the relation $\err_{S^\forall_1}$ (resp., $\err_{S^\forall_2}$)
        contains quadruples of values that cannot appear on positions $y-1,y,y+1$
        of some universal configuration $x$, and at position $y$ of the
        first (resp., second) successor configuration. 
        The corresponding CQs $Q_{S^\forall_1},Q_{S^\forall_2}$ are defined by
        $$
        \begin{aligned}
          Q_{S^\forall_i} ~=
          & ~ \exists ~ x ~ x' ~ y ~ y' ~ y'' ~ z ~ z' ~ z'' ~ z''' \\
          & ~ S^\forall_i(x,x') ~\et~ T(y,y') ~\et~ T(y',y'') ~\et~ \\
          & ~ C(x,y,z) ~\et~ C(x,y',z') ~\et~ C(x,y'',z'') ~\et~ C(x_1,y',z''') ~\et~ \\
          & ~ \err_{S^\forall_i}(z,z',z'',z''') \ .
        \end{aligned}
        $$
\end{itemize}
It remains to check whether the Turing machine $M$ reaches the rejecting state $q_{\kw{rej}}$
along some path of the computation tree. This can be done by introducing a last visible 
relation $V_{\kw{rej}}$ that contains all cell values of the form $(a,q_{\kw{rej}})$,
for some $a\in\Sigma$. The CQ that checks this property is
$$
  Q_{\kw{rej}} ~=~ \exists ~ x ~ y ~ z ~~ C(x,y,z) ~\et~ V_{\kw{rej}}(z) \ .
$$
The final query is thus a disjunction of all the above CQs:
$$
  Q ~=~ Q_C \vel Q_{I,\kw{first}} \vel Q_{I,\kw{last}} \vel Q_{I,\kw{adj}} \vel Q_{C,\kw{adj}} 
            \vel Q_{S^\exists} \vel Q_{S^\forall_1} \vel Q_{S^\forall_2} \vel Q_{\kw{rej}} \ .
$$

\smallskip
We are now ready to give the reduction.
Denote by $\visinst_M$ the instance that captures the intended semantics of the
visible relations $T$, $\kw{First}$, $\kw{Last}$, $V$, $I$, 
$\err_C$, $\err_{I,\kw{first}}$, $\err_{I,\kw{last}}$, $\err_{I,\kw{adj}}$,
$\err_{C,\kw{adj}}$, $\err_{S^\exists}$, $\err_{S^\forall_1}$, and $\err_{S^\forall_2}$. 
We have described these semantics above, and argued why they
can be created in polynomial time.
Below, we prove that the Turing machine $M$ has a successful computation tree
where all paths visit the control state $q_{\kw{acc}}$ if and only if $\PSB(Q,\calC,\bfS,\visinst_M)=\false$.

Suppose that $M$ has a successful computation tree $\rho$. 
On the basis of $\rho$,
and by following the intended semantics of the hidden relations 
$C$, $T^C$, $\kw{First}^C$, $S^\exists$, $S^\forall_1$, $S^\forall_2$, 
we can easily construct a full instance $\fullinst$ that satisfies all the 
sentences in $\calC$, and agrees with $\visinst_M$ on the visible part. 
Furthermore, because we correctly encode a successful computation tree of $M$, 
the instance $\fullinst$ violates every disjunct of $Q$, and hence 
$\PSB(Q,\calC,\bfS,\visinst_M)=\false$.

Conversely, suppose that $\PSB(Q,\calC,\bfS,\visinst_M)=\false$.
Let $\fullinst$ be an $\bfS$-instance that agrees with $\visinst_M$
on the visible part, satisfies the sentences in $\calC$, and violates
every disjunct of the UCQ $Q$. We first construct from $\fullinst$ 
a graph, where every node encodes a configuration and, depending
on whether the configuration is existential or universal, it 
has either one or two outgoing edges that represent some 
transitions of $M$.
We will then argue that the unfolding of this graph from its 
initial node correctly represents an accepting computation 
tree of $M$.
The nodes of the graph are identified by the values $x$ that
appear in facts of $\fullinst$ of the form 
$S^\exists(x,x')$, $S^\forall_1(x,x')$, or $S^\forall_2$. 
The initial node is identified by the unique 
value $x_0$ in the singleton visible relation $I$. 

Thanks to the background theory $\calC$, every configuration identifier $x$ 
also appears in the first column of the hidden relation 
$\kw{First}^C$, and there exist similar occurrences in 
$T^C$ and $C$, one for each cell of the tape.
The content of $C$ can then be used to determine the labeling of
the tape cells, the control state, and the head position for 
each configuration, as indicated by the intended semantics. 
For example, we set the content of a tape cell $y$ in some configuration 
$x$ to be $a$ whenever there is a fact of the form $C(x,y,z)$, 
with $z$ among $(a,q)$, $(a,\lhd)$, or $(a,\rhd)$.
We observe this is well-defined (that is, every tape position 
$y$ at every configuration $x$ has exactly one associated value)
thanks to the sentences
$T^C(x,y,y') \then \exists ~ z ~ C(x,y,z)$ and
$C(x,y,z) \then V(z)$, and thanks to the fact 
that the query $Q_C$ is violated.
Moreover, because the CQs $Q_{I,\kw{first}}$,
$Q_{I,\kw{last}}$, and $Q_{I,\kw{adj}}$ are also violated, the 
configuration at the initial node $x_0$ is correct, that is, 
encodes the tape content $\vdash\sqcup\dots\sqcup\dashv$, 
with control state $q_0$, and head on the first position.

Next, the edges of the graph are constructed using the hidden relations 
$S^\exists$, $S_1^\forall(x,x_1)$, and $S_2^\forall(x,x_2)$ of $\fullinst$.
Formally, for every existential node $x$, the sentences constraining
$S^\exists(x,x')$ imply the existence of at least one node $x'$ 
forming a fact $S^\exists(x,x')$.
We can thus chose any such node $x'$ and declare $(x,x')$ to be
an edge of the graph. A similar argument applies to the universal
nodes, with the only difference that we now introduce two edges 
instead of one, and there is no choice.
Moreover, using the assumption that the CQs $Q_{C,\kw{adj}}$, 
$Q_{S^\exists}$, $Q_{S_1^\forall}$, and $Q_{S_2^\forall}$
are all violated, one can easily verify that the thus defined 
edges represent valid transitions between the encoded configurations.
The above arguments imply that the unfolding of the graph 
from the initial node $x_0$ results in a valid computation tree of $M$. 
Finally, because the CQ $Q_{\kw{rej}}$ is also violated, 
the computation tree must be accepting.

\smallskip
We have just shown the $\exptime$-hardness result for the data complexity 
of the $\PSB$ problem, using a UCQ as query. To finish the proof of Theorem 
\ref{th:psb-instance-hard}, we show that $\PSB$ problems for UCQs can be 
reduced to analogous problems for CQs.

\begin{lemma}\label{lem:ucqtocqpsbinstanceid}
Let $Q = \bigcup Q_i$ be a Boolean UCQ, let $\calC$ be a set of sentences over a 
schema $\bfS$, and let $\visinst$ be an instance for the visible part of $\bfS$. 
There exist a schema $\bfS'$, a CQ $Q'$, a set $\calC'$ of sentences, and an 
$\bfS'_v$-instance $\visinst'$, all having polynomial size with respect to the 
original objects $\bfS$, $Q$, $\calC$, and $\visinst$, such that 
$\PSB(Q,\calC,\bfS,\visinst) = \true$ iff $\PSB(Q',\calC',\bfS',\visinst') = \true$.
\par\noindent
Moreover, the transformation preserves all logical languages considered 
for background theories in our results (e.g., inclusion dependencies).
\end{lemma}

\begin{proof}
The general idea is as follows.
For every visible (resp., hidden) relation $R$ of $\bfS$ of arity $k$, 
we add to $\bfS'$ a corresponding visible (resp., hidden) relation $R'$ 
of arity $k+1$. 
The idea is that the additional attribute of $R'$ represents a truth value, 
e.g. $0$ or $1$, which indicates the presence of a tuple in the original relation $R$. 
For example, the fact $R'(\bar a,1)$ indicates the presence of the tuple 
$\bar a$ in the relation $R$, but $R'(\bar a,0)$ does not. 
The sentences $\calC$ will be rewritten accordingly, so as to propagate these truth values. 
We can then simulate the disjunctions in the query $Q$ by using conjunctions and an 
appropriate look-up table $\kw{Or}$.
This technique has been used in a number of previous works, for example \cite{georgchristos}, and
will also be used later in this paper.
However, due to the nature of the $\PSB$ problem, we also need to add dummy facts
$R'(\bot,\dots,\bot,0)$ in order to correctly transfer the validity from the UCQ 
$Q$ to the CQ $Q'$. We give below the full details.

As mentioned, the new schema $\bfS'$ contains a copy $R'$ of each relation $R$ in $\bfS$, 
where $R'$ is visible iff $R$ is visible, and $R'$ has arity $k+1$ iff $R$ has arity $k$.
In addition, the schema $\bfS'$ contains the visible relations 
$\kw{Or}$, $\kw{Zero}$, $\kw{One}$ of arities $3$, $0$, $0$, respectively,
and some other visible relations $\kw{Bottom}_k$ of arity $k+1$, for all $k$ 
ranging from $0$ to the maximal arity in $\bfS$.

Let us now describe the visible instance $\visinst'$ constructed from $\visinst$.
We choose some fresh values $0$, $1$, and $\bot$ that do not belong to the active 
domain of $\visinst$. 
First, we include in $\visinst'$ the facts $\kw{Or}(1,1,1)$, $\kw{Or}(1,0,1)$, 
$\kw{Or}(0,1,1)$, $\kw{Zero}(0)$, $\kw{One}(1)$, and $\kw{Bottom}_k(\bot,\dots,\bot,0)$
for all arities $k$.
Then, for each visible relation $R$ of $\bfS$, we add to $\visinst'$ the fact 
$R(\bar a,1)$ whenever $R(\bar a)$ is a fact in $\visinst$. 

As for the sentences in the background theory, we proceed as follows. If 
$$
  R(\bar x) ~\then~ \exists \bar y ~ S_1(\bar z_1) \et \ldots \et S_m(\bar z_m)
$$
is a sentence in $\calC$, with $\bar z_1,\ldots,\bar z_m$ sequences of variables or 
constants from $\bar x,\bar y$, then we add to $\calC'$ a corresponding sentence
$$
  R'(\bar x,b) ~\then~ \exists \bar y ~ S'_1(\bar z_1,b) \et \ldots \et S'_m(\bar z_m,b) \ .
$$
Furthermore, for each relation $R$ of arity $k$ in $\bfS$, we introduce the ID
$$
  \kw{Bottom}_{k+1}(x_1,\dots,x_k,y) ~\then~ R'(x_1,\dots,x_k,y) \ .
$$
Recall that $\kw{Bottom}_{k+1}$ is a visible relation of $\bfS'$ that contains the 
single fact $(\bot,\dots,\bot,0)$. Therefore, the effect of the above sentence
is to introduce dummy facts $R'(\bot,\dots,\bot,0)$ for each (visible or hidden)
relation $R'$. 

It now remains to transform the UCQ $Q$ into a CQ $Q'$.
Let $Q_1,\dots,Q_n$ be the disjuncts (CQs) in $Q$. 
We define
$$
  Q' ~=~ \exists \: b_1 \ldots b_n ~ b'_0 ~ b'_1 ~ \ldots ~ b'_n ~\:
         \bigwedge_i Q'_i(b_i) ~\et~
         \kw{Zero}(b'_0) ~\et~ \kw{One}(b'_n) ~\et~
         \bigwedge_i \kw{Or}(b'_{i-1},b_i,b'_i)
$$
where each $Q'_i$ is obtained from the $i$-th disjunct 
$Q_i = \exists \bar y ~ S_1(\bar z_1) ~\et~ \ldots ~\et~ S_m(\bar z_m)$ of $Q$
by letting $Q'_i(b_i) = \exists \bar y ~ S'_1(\bar z_1,b_i) ~\et~ \ldots ~\et~ S'_m(\bar z_m,b_i)$.
Note that the presence of the facts $R'(\bot,\ldots,\bot,0)$ in every instance that extends 
$\visinst'$ and satisfies $\calC'$ guarantees that the rewritten CQs $Q'_i(b_i)$ can 
always be satisfied by letting $b_i=0$. In particular, the sub-query 
$\bigwedge_i Q'_i(b_i)$ holds at least with all the $b_i$'s set to $0$.
The remaining part of the query $Q'$ precisely requires that at least one of those 
$b_i$'s is set to $1$.

We are now ready to prove that 
$\PSB(Q,\calC,\bfS,\visinst) = \true$ iff $\PSB(Q',\calC',\bfS',\visinst') = \true$.
Suppose that $\PSB(Q',\calC',\bfS',\visinst') = \true$ and consider an $\bfS$-instance 
$\fullinst$ that satisfies the sentences in $\calC$ and such that 
$\visible(\fullinst) = \visinst$. 
Without loss of generality, we can assume that the active domain of $\fullinst$ does not
contain the values $0$, $1$, and $\bot$.
We can easily transform $\fullinst$ into an $\bfS'$-instance $\fullinst'$ 
by expanding all facts with the additional attributed value $1$ 
and by adding new facts of the form
$R'(\bot,\ldots,\bot,0)$, for all relations $R'\in\bfS'$, 
together with the visible facts
$\kw{Or}(1,1,1)$, $\kw{Or}(1,0,1)$, $\kw{Or}(0,1,1)$, $\kw{Zero}(0)$, $\kw{One}(1)$,
and $\kw{Bottom}_k(\bot,\dots,\bot,0)$ for all arities $k$.
One easily verifies that $\fullinst'$ satisfies the sentences in $\calC'$
and agrees with $\visinst'$ on the visible part. Since $\PSB(Q',\calC',\bfS',\visinst') = \true$,
we know that $\fullinst'$ also satisfies the query $Q'$ and, in particular, 
it satisfies one of the conjuncts $Q'_i(b_i)$ of $Q'$ with $b_i=1$.
This implies that $\fullinst$ satisfies the corresponding Boolean CQ $Q_i$, and hence $Q$ as well.

Conversely, suppose that $\PSB(Q,\calC,\bfS,\visinst) = \true$ and 
consider an $\bfS'$-instance $\fullinst'$ that satisfies the sentences
in $\calC'$ and such that $\visible(\fullinst') = \visinst'$.
By selecting from $\fullinst'$ only the facts of the form $R'(\bar a,1)$, with $R\in\bfS$,
and by projecting away the last attribute, we obtain an $\bfS$-instance $\fullinst$
that satisfies the sentences in $\calC$ and such that $\visible(\fullinst) = \visinst$.
Since $\PSB(Q,\calC,\bfS,\visinst) = \true$, we know that $\fullinst$ satisfies 
at least one of the disjuncts $Q_i$ of $Q$. 
This immediately implies that $\fullinst'$ satisfies the CQ $Q'_i(b_i)$ with $b_i=1$.
As for the remaining conjuncts of the query $Q'$, we recall that $\fullinst'$ must 
contain facts of the form $R'(\bot,\dots,\bot,0)$ for all relations $R'$. 
Thanks to these facts, the CQs $Q'_j(b_j)$ hold on $\fullinst'$ with $b_j=0$,
for all $j\neq i$, and hence $Q'$ holds on $\fullinst'$ as well.
\end{proof}

Applying the lemma above, we have proven Theorem \ref{th:psb-instance-hard}.
\end{proof}

We note that the above lower bound for data complexity makes use of a schema with arity above $2$,
even for CQs. See, for example, the ternary relation  $C$.
 We do not know whether our lower bound still holds
for the arity $2$ case.
Our results contrasts with results of Franconi et al.~\cite{FranconiIS11}, which  show that
the data complexity lies in $\conp$ (and can be $\conp$-hard) for certain
description logics over arity $2$. 
%We believe that if move up from IDs to linear TGDs (which can not be captured in the formalism
%of \cite{FranconiIS11}, we can adapt the argument for the above 
%result to show $\exptime$-hardness even for arity $2$.

\medskip
We now turn to the combined complexity and show that the $\twoexp$ upper bound 
of Theorem \ref{thm:gnfoposinstancedecid} is tight even for IDs.

\begin{theorem}\label{thm:psbidtwoexphardcomb}
Checking $\PSB(Q,\calC,\bfS,\visinst)$, where $Q$ ranges over CQs and 
$\calC$ over sets of inclusion dependencies, is $\twoexp$-hard for 
combined complexity.
\end{theorem}

\begin{proof}
%\gabriele{Super-simplified proof version, no tree-like}
This proof builds up on ideas from the previous proof for Theorem \ref{th:psb-instance-hard}.
Specifically, we reduce the acceptance problem for an alternating $\expspace$ Turing machine 
$M$ to the negation of $\PSB(Q,\calC,\bfS,\visinst)$, where $Q$ is a Boolean UCQ and
$\calC$ consists of inclusion dependencies. Note that to further reduce the problem to a 
Positive Query Implication problem with a Boolean CQ, one can exploit Lemma \ref{lem:ucqtocqpsbinstanceid}.

The additional technical difficulty here is to encode a tape of exponential size.
Of course, this cannot be done succinctly using an instance with visible relations. 
However, we can represent the exponential tape by a set of tuples of bits.
More precisely, given an alternating $\expspace$ Turing machine $M$ and an input
for $M$ of length $n$, we identify each cell of the tape of $M$ by an $n$-tuple of bits.
Note that, differently from the reduction in Theorem \ref{th:psb-instance-hard}, here 
we can let the schema, the sentences, and the query depend on $M$ and $n$, 
since the goal here is to prove a lower bound for combined complexity.

For the sake of simplicity, we first explain how to create a single tape of exponential
length, without being concerned about the content of the cells and the different configurations
that can be reached by $M$.
For this, we introduce three visible relations $\kw{Zero}$, $\kw{One}$, and $\kw{Bit}$, 
instantiated with $\{0\}$, $\{1\}$, and $\{0,1\}$, respectively. 
We also introduce hidden relations $T_i,T_{i,\kw{zero}},T_{i,\kw{one}}$ of arity $i$, 
for all $i=1,\dots,n$, and an additional hidden relation $T_0$ of arity $0$. 
Intuitively, the intended semantics of each relation $T_i$ is to contain all $i$-tuples of bits,
while $T_{i,\kw{zero}}$ (resp., $T_{i,\kw{one}}$) is the restriction of $T_i$ to the
tuples ending with $0$ (resp., $1$). We enforce this semantics using a simple induction 
on $i=1,\dots,n$ and the following inclusion dependencies:
$$
\begin{array}{l}
\begin{array}{rcl}
  \kw{true} &\then& T_0() \\[0.5ex]
  (\forall j\le i) \qquad T_i(y_1,\dots,y_i) &\then& \kw{Bit}(y_j)
\end{array}
\\[3.5ex]
\begin{array}{rcl}
  T_{i-1}(y_1,\dots,y_{i-1}) &\then& \exists y_i ~ T_{i,\kw{zero}}(y_1,\dots,y_i) \\[0.5ex]
  T_{i-1}(y_1,\dots,y_{i-1}) &\then& \exists y_i ~ T_{i,\kw{one}}(y_1,\dots,y_i)
\end{array}
\end{array}
\quad
\begin{array}{rcl}
  T_{i,\kw{zero}}(y_1,\dots,y_i) &\then& \kw{Zero}(y_i) \\[0.5ex]
  T_{i,\kw{one}}(y_1,\dots,y_i) &\then& \kw{One}(y_i) \\[1.75ex]
  T_{i,\kw{zero}}(y_1,\dots,y_i) &\then& T_i(y_1,\dots,y_i) \\[0.5ex]
  T_{i,\kw{one}}(y_1,\dots,y_i) &\then& T_i(y_1,\dots,y_i) \ .
\end{array}
$$
It is clear that every instance satisfying the above sentences will have $T_n=\kw{Bit}^n$,
so the tuples in $T_n$ can be used to represent the cells of a tape of exponential length.

Cells are naturally ordered in the tape, and so must be the tuples in $T_n$. We use the lexicographic 
order on $n$-tuples of bits, and show how to access this order by means of a formula.
Formally, we need to write a UCQ that checks whether two cells, identified by some $n$-tuples
$\bar y=(y_1,\dots,y_n)$ and $\bar y'=(y'_1,\dots,y'_n)$ in $T_n$, are adjacent according to 
the lexicographic ordering. 
A well-known technique consists in determining the smallest index $1\le i\le n$ such that $y_i\neq y'_i$.
Then, given such $i$, one verifies that $y_i=0$, $y'_i=1$, $y_j=1$, and $y'_j=0$ for all $j>i$.
We give beforehand the formula that checks these conditions.
The formula is the disjunction over all $i=1,\ldots,n$ of the following CQs:
%\gabriele{I guess it is not a problem to use equalities here (eg. $y_j=y'_j$).
 %         One could alternatively use a visible relation $\kw{Eq}=\{(0,0),(1,1)\}$.}
%michael: explained that this is a shorthand
$$
\begin{aligned}
  Q_{\kw{adj},i}(\bar y,\bar y') ~=
  \bigwedge_{1\le j<i} (y_j = y'_j) ~\et~ 
  \kw{Zero}(y_i) ~\et~ \kw{One}(y'_i) ~\et
  \bigwedge_{i<j\le n} \kw{One}(y_j) ~\et
  \bigwedge_{i<j\le n} \kw{Zero}(y'_j) \ .
\end{aligned}
$$
Here for convenience of description we allow equalities in a CQ, but
they can be replaced in favor of an explicit substitution.
It is not difficult to see that the UCQ $\bigvee_{1\le i\le n}Q_{\kw{adj},i}$ defines 
precisely those pairs of tuples that are consecutive in the lexicographic order.
Moreover, we will need to easily identify the first and the last cell of the tape.
For this we introduce two visible relations $\kw{First}$ and $\kw{Last}$, both of arity $n$, 
and instantiate them with the singletons $\{(0,\dots,0)\}$ and $\{(1,\dots,1)\}$, respectively.

\smallskip
Now that we know how to represent exponentially many cells in the tape and check 
their adjacency, we proceed as in the proof of Theorem \ref{th:psb-instance-hard}.
We begin by encoding configurations of $M$. Intuitively, the goal is to create
a copy $C$ of the relation $T_n$, expanded with configuration identifiers and cell values,
in such a way that a fact of the form $C(x,y_1,\dots,y_n,z)$ denotes the existence of 
a configuration identified by $x$, where the tape cell represented by $\bar y=(y_1,\dots,y_n)$ 
carries the value $z$.
As usual (cf.~proof of Theorem \ref{th:psb-instance-hard}), we define cell values as 
elements from a visible unary relation
$V = \Sigma_Q \uplus \Sigma_\lhd \uplus \Sigma_\rhd$,
where $\Sigma$ is the alphabet of the Turing machine, 
$\Sigma_Q = \Sigma\times Q$.
$\Sigma_\lhd = \Sigma\times\{\lhd\}$, 
$\Sigma_\rhd = \Sigma\times\{\rhd\}$, 
$Q$ is the set of its control states, and
$\lhd,\rhd$ are fresh symbols.
To correctly instantiate the relation $C$, we create also copies of the relations 
$T_i,T_{i,\kw{zero}},T_{i,\kw{one}}$, expanded with configuration identifiers, and 
enforce constraints analogous to the ones introduced in the sentences above.
More precisely, we have the following hidden relations:
$C$ of arity $n+2$,
$T^C_i$ of arity $i+1$, for all $i=0,\dots,n$, 
$T^C_{i,\kw{zero}}$ and $T^C_{i,\kw{one}}$ of arity $i+1$, for all $i=1,\dots,n$.
We have the following sentences  for all $i=1,\dots,n$:
$$
\!\!
\begin{array}{l}
  (\forall j\le i) \qquad T^C_i(x,y_1,\dots,y_i) ~\then~ \kw{Bit}(y_j)
\\[2ex]
\begin{array}{rclrcl}
  T^C_n(x,y_1,\dots,y_n) &\!\!\!\!\then&\!\!\!\! \exists z ~ C(x,y_1,\dots,y_n,z)
  &~
  T^C_{i,\kw{zero}}(x,y_1,\dots,y_i) &\!\!\!\!\then&\!\!\!\! \kw{Zero}(y_i) 
  \\[0.5ex]
  C(x,y_1,\dots,y_n,z)   &\!\!\!\!\then&\!\!\!\! V(z)
  &~
  T^C_{i,\kw{one}}(x,y_1,\dots,y_i) &\!\!\!\!\then&\!\!\!\! \kw{One}(y_i) 
  \\[2ex]
  T^C_{i-1}(x,y_1,\dots,y_{i-1}) &\!\!\!\!\then&\!\!\!\! \exists y_i ~ T^C_{i,\kw{zero}}(x,y_1,\dots,y_i)
  &~
  T^C_{i,\kw{zero}}(x,y_1,\dots,y_i) &\!\!\!\!\then&\!\!\!\! T^C_i(x,y_1,\dots,y_i) 
  \\[0.5ex]
  T^C_{i-1}(x,y_1,\dots,y_{i-1}) &\!\!\!\!\then&\!\!\!\! \exists y_i ~ T^C_{i,\kw{one}}(x,y_1,\dots,y_i)
  &~ 
  T^C_{i,\kw{one}}(x,y_1,\dots,y_i) &\!\!\!\!\then&\!\!\!\! T^C_i(x,y_1,\dots,y_i).
\end{array}
\end{array}
$$
Note that the analog of the sentence $\kw{true} \then T_0()$ is missing here.
This will be given later, when we will explain how new configurations are created
to simulate a computation tree of $M$. For the moment it suffices to observe that,
in every instance that satisfies the above sentences, as soon as $T^C_0$ contains
a configuration identifier $x$, then $T^C_n$ contains all tuples of the form
$(x,y_1,\dots,y_n)$, with $(y_1,\dots,y_n)\in\kw{Bit}^n$,
and $C$ specifies at least one value $z$ for each configuration identifier $x$
and each cell $(y_1,\dots,y_n)$.

We now turn towards the encoding of the computation tree of $M$. 
This is almost the same as in the proof of Theorem \ref{th:psb-instance-hard}.
We introduce a visible unary relation $I$, which contains the identifier $x_0$ of 
the initial existential configuration, and three hidden binary relations 
$S^\exists$, $S^\forall_1$, and $S^\forall_2$. A fact of the form $S^\exists(x,x')$ 
(resp., $S^\forall_1(x,x_1)$, $S^\forall_2(x,x_1)$) represents a transition from an 
existential (resp., universal) configuration $x$ to a universal (resp., existential) 
configuration $x'$ (resp., $x_1$, $x_2$). We then include the following
sentences in the background theory:
$$
\begin{array}{rcl}
  I(x)               &\then& \exists ~ x' ~ S^\exists(x,x') \\[0.5ex]
  S^\exists(x,x')    &\then& \exists ~ x_1 ~ S^\forall_1(x',x_1) \\[0.5ex]
  S^\exists(x,x')    &\then& \exists ~ x_2 ~ S^\forall_2(x',x_2) \\[0.5ex]
  S^\forall_1(x,x_1) &\then& \exists ~ x' ~ S^\exists(x_1,x') \\[0.5ex]
  S^\forall_2(x,x_2) &\then& \exists ~ x' ~ S^\exists(x_2,x')
\end{array}
\qquad
\begin{array}{rcl}
  S^\exists(x,x')    &\then& T^C_0(x) \\[2.5ex]
  S^\forall_1(x,x_1) &\then& T^C_0(x) \\[2.5ex]
  S^\forall_2(x,x_2) &\then& T^C_0(x) \ .
\end{array}
$$
Intuitively, the rules on the left enforce the existence of a transition graph
where $x_0\in I$ is the initial node and every node has one or two outgoing edges,
depending on whether it is existential or universal.
The rules on the right trigger the instantiation of the tables $T^C_n$ and $C$, 
with the intended goal of representing the content of the tape associated with
each node/configuration. As usual, the unfolding of the transition graph from 
the initial node yields a tree, which should represent a computation of $M$.

\smallskip
It remains to describe how we detect badly-formed encodings of computations 
of $M$. For this, we introduce new visible relations 
$\err_C$, $\err_{I,\kw{first}}$, $\err_{I,\kw{last}}$, $\err_{I,\kw{adj}}$,
$\err_{C,\kw{adj}}$, $\err_{S^\exists}$, $\err_{S^\forall_1}$, and $\err_{S^\forall_2}$, 
whose instances are defined exactly as in the proof of Theorem \ref{th:psb-instance-hard}.
\begin{itemize}
  \item The relation $\err_C$ is binary and contains 
        all pairs of distinct values from $V\times V$.
        This is used to detect multiple values associated with the same cell:
        $$
          Q_C ~=~ \exists ~ x ~ \bar y ~ z ~ z' ~~
                  C(x,\bar y,z) ~\et~ C(x,\bar y,z') ~\et~ \err_C(z,z') \ .
        $$
  \item The relation $\err_{I,\kw{first}}$ contains all pairs in $V\times V$
        but $(z_0,z_1)$, where $z_0=(\vdash,q_0)$ and $z_1=(\sqcup,\rhd)$.
        This is used to detect wrong values associated with the first two cells 
        of the initial configuration:
        $$
        \begin{aligned}
          Q_{I,\kw{first}} ~= 
          & ~ \exists ~ x ~ \bar y ~ \bar y' ~ z ~ z' ~~ \\
          & ~ I(x) ~\et~ \kw{First}(\bar y) ~\et~ 
              \bigvee\nolimits_{1\le i\le n} Q_{\kw{adj},i}(\bar y,\bar y') ~\et~ \\
          & ~ C(x,\bar y,z) ~\et~ C(x,\bar y',z') ~\et~ \err_{I,\kw{first}}(z,z') \ .
        \end{aligned}
        $$        
        Note that, strictly speaking, the above query is not a UCQ, but can be easily
        normalized into a UCQ of polynomial size. The same remark applies to all 
        remaining queries.
  \item Similar visible relations $\err_{I,\kw{last}}$, $\err_{I,\kw{adj}}$, $\err_{C,\kw{adj}}$ 
        and UCQs $Q_{I,\kw{last}}$, $Q_{I,\kw{adj}}$, $Q_{C,\kw{adj}}$ are used
        to detect wrong values, respectively, for the last two cells of 
        the initial configuration, for any two adjacent cells of the initial configuration, 
        and for any two adjacent cells of an arbitrary configuration.
  \item To detect the violations that involve values associated with the same 
        position of the tape but in two consecutive configurations, 
        we use the following UCQs:
        $$
        ~~\qquad
        \begin{aligned}
          Q_{S^\exists} ~=
          & ~ \exists ~ x ~ x' ~ \bar y ~ \bar y' ~ \bar y'' ~ z ~ z' ~ z'' ~ z''' \\
          & ~ S^\exists(x,x') ~\et~ 
              \bigvee\nolimits_{1\le i\le n} Q_{\kw{adj},i}(\bar y,\bar y') ~\et~ 
              \bigvee\nolimits_{1\le i\le n} Q_{\kw{adj},i}(\bar y',\bar y'') ~\et~ \\
          & ~ C(x,\bar y,z) ~\et~ C(x,\bar y',z') ~\et~ C(x,\bar y'',z'') ~\et~ C(x',\bar y',z''') ~\et~ 
              \err_{S^\exists}(z,z',z'',z''') \\[2ex]
          Q_{S^\forall_1} ~=
          & ~ \exists ~ x ~ x_1 ~ \bar y ~ \bar y' ~ \bar y'' ~ z ~ z' ~ z'' ~ z''' \\
          & ~ S^\forall_1(x,x_1) ~\et~ 
              \bigvee\nolimits_{1\le i\le n} Q_{\kw{adj},i}(\bar y,\bar y') ~\et~ 
              \bigvee\nolimits_{1\le i\le n} Q_{\kw{adj},i}(\bar y',\bar y'') ~\et~ \\
          & ~ C(x,\bar y,z) ~\et~ C(x,\bar y',z') ~\et~ C(x,\bar y'',z'') ~\et~ C(x_1,\bar y',z''') ~\et~ 
              \err_{S^\forall_1}(z,z',z'',z''') \\[2ex]
          Q_{S^\forall_2} ~=
          & ~ \exists ~ x ~ x_2 ~ \bar y ~ \bar y' ~ \bar y'' ~ z ~ z' ~ z'' ~ z''' \\
          & ~ S^\forall_2(x,x_2) ~\et~ 
              \bigvee\nolimits_{1\le i\le n} Q_{\kw{adj},i}(\bar y,\bar y') ~\et~ 
              \bigvee\nolimits_{1\le i\le n} Q_{\kw{adj},i}(\bar y',\bar y'') ~\et~ \\
          & ~ C(x,\bar y,z) ~\et~ C(x,\bar y',z') ~\et~ C(x,\bar y'',z'') ~\et~ C(x_2,\bar y',z''') ~\et~ 
              \err_{S^\forall_2}(z,z',z'',z''')
        \end{aligned}
        $$
        where $\err_{S^\exists}$, $\err_{S^\forall_1}$, and $\err_{S^\forall_2}$ 
        are defined exactly as in the proof of Theorem \ref{th:psb-instance-hard}.
\end{itemize}
In addition, we check whether the Turing machine $M$ reaches the rejecting state 
$q_{\kw{rej}}$ along some path in its computation tree. This is done with the CQ
$$
  Q_{\kw{rej}} ~=~ \exists ~ x ~ \bar y ~ z ~~ C(x,\bar y,z) ~\et~ V_{\kw{rej}}(z)
$$
where $V_{\kw{rej}}$ is the visible relation that contains all cell values 
of the form $(a,q_{\kw{rej}})$, for some $a\in\Sigma$. 

Let $Q$ be the disjunction of all the previous UCQs and let $\visinst$ 
be the instance that captures the intended semantics of the visible relations 
$\kw{Zero}$, $\kw{One}$, $\kw{Bit}$, $V$, $\err_C$, $\err_{I,\kw{first}}$, $\err_{I,\kw{last}}$, 
$\err_{I,\kw{adj}}$, $\err_{C,\kw{adj}}$, $\err_{S^\exists}$, $\err_{S^\forall_1}$, and $\err_{S^\forall_2}$.
We can argue along the same lines of the proof of Theorem \ref{th:psb-instance-hard} 
that $M$ has a successful computation tree iff $\PSB(Q,\calC,\bfS,\visinst)=\false$.
\end{proof}

%% file: posexists.tex
\subsection{Schema-level problem}\label{sec:existsPSB}

In this section we focus on the schema-level problem $\exists \PSB$, 
namely, the problem of deciding the existence of a instance $\visinst$
such that $\PSB(Q,\calC, \bfS,\visinst)=\true$.

Let $a$ be an arbitrary domain element.
Further let $\visinst_{\{a\}}$ be a fixed instance for the visible part of 
a schema $\bfS$ whose domain contains the single value $a$ and whose 
visible relations are singleton relations of the form $\{(a,\ldots,a)\}$.
We will show that, for certain languages for the background theories, if $\exists\PSB(Q, \calC, \bfS)=\true$,
then the witnessing instance can be taken to be $\visinst_{\{a\}}$.
This can be viewed as an extension of the ``critical instance'' method
which has been applied previously to chase termination problems:
Proposition 3.7 of Marnette and Geerts \cite{mgicdt} states a related
result for disjunctive TGDs in isolation; 
Gogacz and Marcincowski \cite{criticalinst1} call such an instance a ``well of positivity''. 
The following result shows that the technique applies to TGDs and EGDs without constants.

\begin{theorem}\label{thm:psbexistcollapselin}
For every Boolean UCQ $Q$ without constants, and every set $\calC$ of TGDs and EGDs without constants,
$\exists \PSB(Q,\calC, \bfS)=\true$ iff $\PSB(Q,\calC, \bfS, \visinst_{\{a\}})=\true$.
\end{theorem}

First, we prove the theorem for background theories consisting only of TGDs without constants.
Then we will show how to generalize the proof in the additional presence of EGDs 
without constants.

We recall that the visible instance $\visinst_{\{a\}}$ is constructed
over a singleton active domain and the sentences in the background theory $\calC$ have no constants. 
This implies that there are no disjunctive choices to perform while 
chasing with the dependencies starting from the initial instance $\visinst_{\{a\}}$. 
Moreover, it is easy to see that this chase always succeeds.
That is, it returns a collection $\Chase(\calC,\bfS,\visinst_{\{a\}})$ 
with exactly one instance --- in particular, $\visinst_{\{a\}}$ 
is a realizable instance. By a slight abuse of notation, we denote by 
$\chase(\calC,\bfS,\visinst_{\{a\}})$
the unique instance in the collection 
$\Chase(\calC,\bfS,\visinst_{\{a\}})$. 
% \balder{The ``always succeeds'' claim is not true when EGDs are added. So, 
% in anticipation, it would be better to avoid relying on this here?}

\begin{lemma}\label{lem:well-of-positivity}
If $\calC$ is a set of TGDs without constants over a schema $\bfS$ and $\visinst$ 
is an instance of the visible part of $\bfS$, then every instance 
$K \in \Chase(\calC,\bfS,\visinst)$ maps homomorphically to 
$\chase(\bfS,\calC,\visinst_{\{a\}})$, 
that is, 
$h(K) \subseteq \chase(\bfS,\calC,\visinst_{\{a\}})$
for some homomorphism $h$.
\end{lemma}

\begin{proof}
Recall that the instances in $\Chase(\calC,\bfS,\visinst)$ 
are either leaves or limits of infinite paths of the chase tree. 
Below, we prove that every instance $K$ in the chase tree 
for $\Chase(\bfS,\calC,\visinst)$ 
maps to 
$\chase(\bfS,\calC,\visinst_{\{a\}})$ via 
some homomorphism $h$. In addition, we ensure that, 
if $K'$ is a descendant of $K$ in the same chase tree, 
then the corresponding homomorphism $h'$ is obtained
by composing some homomorphism with an extension of $h$.
This way of constructing homomorphisms is compatible with 
limits in the following sense: if $h_0,h_1,\ldots$ are
homomorphisms mapping instances $K_0,K_1,\ldots$ 
along an infinite path of the chase tree, then there is
a homomorphism $\lim_{n\in\bbN}h_n$ that maps the limit 
instance $\lim_{n\in\bbN}K_n$ to $\fullinst$. 

For the base case of the induction, we consider the initial
instance $\visinst$ at the root of the chase tree, which
clearly maps homomorphically to $\visinst_{\{a\}}$
(recall that there are no constants in the query or sentences of the background
theory, and homomorphisms
are free to map all domain elements to $a$). 
For the inductive case, we consider an instance $K$ in the 
chase tree 
%for $\Chase(\bfS,\calC,\visinst)$ 
and suppose that it maps to
$\chase(\bfS,\calC,\visinst_{\{a\}})$ 
via a homomorphism $h$.
We also consider an instance $K'$ that is a child of $K$ and 
is obtained by chasing some dependency 
$R_1(\bar x_1) \et\ldots\et R_m(\bar x_m) \then \exists \bar y~S(\bar z)$,
where $\bar z$ is a sequence of variables from $\bar x_1,\ldots,\bar x_m,\bar y$.
This means that there exist two homomorphisms $f$ and $g$ such that
\begin{enumerate}
  \item $f$ maps the variables $\bar x_1,\ldots,\bar x_m$ to some values in $K$
        and maps injectively the variables $\bar y$ to fresh values;
%  \item $g$ is the identity on $f(\bar x_1),\ldots,f(\bar x_m)$ and 
%        maps the fresh values $f(\bar y)$ to values of the instance
%        $\chase(\bfS,\calC,\visinst_{\{a\}})$,
  \item $g$ either maps $f(\bar z)$ to values in the active domain
        of $\visinst$ or is the identity on $f(\bar z)$,
        depending on whether $S$ is visible or not;
  \item $R_j\big(f(\bar x_j)\big)\in K$ for all $1\le j\le m$; 
  \item $K' = g\big( K \cup \big\{ S(f(\bar z)) \big\} \big)$.
\end{enumerate}
Note that $h$ maps each fact $R_j\big(f(\bar x_j)\big)$ 
in $K$ to $R_j\big(h(f(\bar x_j))\big)$ in 
$\chase(\bfS,\calC,\visinst_{\{a\}})$.
Since 
$\chase(\bfS,\calC,\visinst_{\{a\}})$ 
satisfies the chased dependency, it must also contain a fact 
of the form $S\big(h'(f(\bar z))\big)$, where $h'$ is a homomorphism 
that extends $h$ on the fresh values $f(\bar y)$.
Moreover, if $S$ is visible, then $h'$ maps all values $f(\bar z)$ 
to the same value $a$, which is the only element of the active domain of $\visinst_{\{a\}}$.

We can now define a homomorphism that maps the instance
$K'=g\big( K \cup \big\{ S(f(\bar z)) \big\} \big)$ to 
$\chase(\bfS,\calC,\visinst_{\{a\}})$.
%In particular, the desired homomorphism needs to be defined 
%on the values of the form $g(f(z))$, for some variable $z$ 
%among $\bar x_1,\ldots,\bar x_m,\bar y$.
%Recall that $g$ is the identity on $f(\bar x_1),\ldots,f(\bar x_m)$
%and is injective on $f(\bar y)$. This means that $g^{-1}$ is 
%well-defined on the values $g(f(\bar x_1)),\ldots,g(f(\bar x_m)),g(f(\bar y))$.
%We can thus define $h''=h'\circ g^{-1}$ and conclude that 
%$h''(K') = h'\big( K \cup \big\{ S(f(\bar z)) \big\} \big) 
%         \subseteq \chase(\bfS,\calC,\visinst_{\{a\}})$.
If $S$ is not visible, then we recall that $g$ is 
the identity on $f(\bar z)$, and hence $h'$ already maps 
$K'=g\big( K \cup \big\{ S(f(\bar z)) \big\} \big) = K \cup \big\{ S(f(\bar z)) \big\}$
to $\chase(\bfS,\calC,\visinst_{\{a\}})$.
Otherwise, if $S$ is visible, then we recall that
$g$ maps $f(\bar z)$ to values in the active domain
of $\visinst$, we let $g'$ be the function that maps 
all values of the active domain of $\visinst$ to $a$, 
and finally we define $h''=h' \circ g'$. In this  way 
$h''$ maps $K'=g\big( K \cup \big\{ S(f(\bar z)) \big\} \big)$
to $\chase(\bfS,\calC,\visinst_{\{a\}})$.
\end{proof}

\smallskip
Now that we established the key lemmas, we can easily reduce the existence problem 
to an instance-based problem (recall that for the moment we assume that the 
background theory consists only of TGDs):

\smallskip
\begin{proof}[Proof of Theorem \ref{thm:psbexistcollapselin} (with TGDs only)]
One of the two directions is trivial: if $\PSB(Q,\calC,\bfS,\visinst_{\{a\}})=\true$, 
then clearly $\exists\PSB(Q,\calC,\bfS)=\true$.

For the converse direction, suppose that $\exists\PSB(Q, \calC, \bfS)=\true$.
This implies the existence of a realizable instance $\visinst$ such that 
$\PSB(Q,\calC,\bfS,\visinst)=\true$.
By Proposition \ref{prop:PSB-characterizion}, every instance in 
$\Chase(\bfS,\calC,\visinst)$ 
satisfies the query $Q$. 
Moreover, by Lemma \ref{lem:well-of-positivity}, every instance in 
$\Chase(\bfS,\calC,\visinst)$ 
maps homomorphically to 
$\chase(\bfS,\calC,\visinst_{\{a\}})$.
Hence the unique instance in $\Chase(\bfS,\calC,\visinst_{\{a\}})$,
i.e.~$\chase(\bfS,\calC,\visinst_{\{a\}})$,
also satisfies $Q$.
By applying Proposition \ref{prop:PSB-characterizion} again,
we conclude that $\PSB(Q,\calC,\bfS,\visinst_{\{a\}})=\true$.

Finally, the second statement of the theorem follows from the fact that 
the previous proofs are independent of the assumption that relational 
instances are finite.
\end{proof}

Now, we explain how to generalize the proof of Theorem \ref{thm:psbexistcollapselin} 
to combinations of TGDs and EGDs (still without constants). 
Recall that we can modify the  procedure 
for $\Chase(\bfS,\calC,\visinst)$ so as to also take into account  the EGDs 
in $\calC$ that can be triggered on the instances that emerge in the chase tree. 
Using this extended definition of 
$\Chase(\bfS,\calC,\visinst)$ 
at hand, the proof
of Lemma \ref{lem:well-of-positivity} does not pose particular problems,
as one just needs to handle the standard case of an EGD dependency. 
Finally, the proof of Theorem \ref{thm:psbexistcollapselin}
 directly uses Proposition \ref{prop:PSB-characterizion} 
and Lemma \ref{lem:well-of-positivity} as black boxes, and so 
carries over without any modification.

\medskip

It is worth remarking that, by pairing Theorem \ref{thm:psbexistcollapselin} 
with the upper bound and the finite controllability for instance-level problems
(Theorem \ref{thm:gnfoposinstancedecid}), one immediately obtains the following:

\begin{corollary}\label{cor:decidexistpsbgtgd}
$\exists\PSB(Q,\calC,\bfS)$  with $Q$ ranging
over  Boolean UCQs and  $\calC$ over sets of frontier-guarded TGDs %($\fgtgd$s)
without constants, 
is decidable in 
$\twoexp$, and is finitely controllable.
\end{corollary}

In contrast, we show that allowing disjunctions or constants in the background
theory sentences
leads to undecidability. We first prove this in the case
where the sentences include  disjunctions. This shows that the
interaction of disjunctive linear TGDs and linear EGDs (implicit
in the requirement that in a possible world for an instance $\fullinst$,
each fact of a visible relation $R$ world must be one of the $R$-facts of $\fullinst$)
causes the ``critical instance''
reduction to fail.

\begin{theorem}\label{thm:undecidexistpsbdisjunction}
The problem $\exists\PSB(Q,\calC,\bfS)$ is undecidable 
as $Q$ ranges over Boolean UCQs and $\calC$ over sets 
of disjunctive linear TGDs. 
\end{theorem}

The proof uses a technique that will be exploited for many of our schema-level
undecidability arguments. We will reduce the existence of a tiling to the $\exists \PSB$ problem.
The tiling itself will correspond to the visible instance that has a $\PSB$. The invisible relations
will store ``challenges'' to the correctness of the tiling.  The UCQ $Q$ will have disjuncts
that return $\true$ exactly when the challenge to correctness is passed. There will be challenges
to the labelling of adjacent cells, challenges to the correctness of the initial tile, and challenges
to the correct shape of the adjacency relationship -- that is, challenges that the tiling is really
grid-like. A correct tiling corresponds to every challenge being passed, 
and thus corresponds to a visible instance where every extension satisfies $Q$.
The undecidability argument also applies to the ``unrestricted version''  of $\exists \PSB$, in which
both quantifications over instances consider arbitrary instances. This will also be true for all other
undecidability results in this work, which always concern the schema-level problems.

\begin{proof}[Proof of Theorem \ref{thm:undecidexistpsbdisjunction}]
For simplicity, we deal with the ``unrestricted variant'' of the problem, which
asks if there is an arbitrary instance of the visible schema such that every superinstance
satisfying the sentences in $\calC$ also satisfies $Q$. Later we will show to modify the proof for 
dealing with finite instances.

We reduce the problem of tiling the infinite grid, which is known to be 
undecidable, to the problem $\exists\PSB$.
Recall that an instance of the tiling problem consists of a finite 
set $T$ of available tiles, some horizontal and vertical constraints,
given by two relations $H,V\subseteq T\times T$, and an initial tile $t_\bot\in T$ 
for the lower-left corner. The problem consists of deciding whether there is 
a tiling function $f:\bbN\times\bbN\then T$ such that 
\begin{enumerate}
  \item $f(0,0) = t_\bot$,
  \item $( f(i,j), f(i+1,j) ) \in H$ for all $i,j\in\bbN$,
  \item $( f(i,j), f(i,j+1) ) \in V$ for all $i,j\in\bbN$.
\end{enumerate}
Given an instance $(T,H,V,t_\bot)$ of the tiling problem,
we show how to construct a schema $\bfS$, a query $Q$, and a set of 
disjunctive linear TGDs %disjunctive IDs 
over $\bfS$ such that $\exists\PSB(Q,\calC,\bfS)=\true$
if and only if there is a tiling function for $(T,H,V,t_\bot)$.

The basic idea is that the visible instance that witnesses $\exists\PSB$ should 
represent a candidate tiling, and the invisible instances represent challenges 
to the correctness of the tiling.
Every cell of the grid is identified with some value, and we
use two visible binary relations $E_H, E_V$ to represent the 
horizontal and vertical edges of the grid. We also introduce 
a unary visible relation $U_t$, for each tile $t\in T$, to represent 
a candidate tiling function on the grid. 

We begin by enforcing the existence of an initial node with the associated tile
$t_\bot$. For this, we introduce another visible relation $\kw{Init}$, of arity 
$0$, and linear TGD
$$
  \kw{Init} ~\then~ \exists x ~ U_{t_\bot}(x) \ .
$$
It is also easy to guarantee that every node is connected to at least another 
node in the relation $E_H$ (resp., $E_V$), and that this latter node has an
associated tile that satisfies the horizontal constraints $H$ (resp., the
vertical constraints $V$). To do so we use the following disjunctive linear TGDs:
\begin{align*}
  U_t(x)        &~\then~ \exists y ~ E_H(x,y) ~\et~ \bigvee\nolimits_{(t,t')\in H} U_{t'}(y) 
                      \tag{for all tiles $t\in T$} \\[1ex]
  U_t(x)        &~\then~ \exists z ~ E_V(x,z) ~\et~ \bigvee\nolimits_{(t,t')\in V} U_{t'}(z)
                      \tag{for all tiles $t\in T$}
\end{align*}
We now explain how to enforce a grid structure on the relations
$E_H$ and $E_V$, and how to guarantee that each node has exactly one
tile associated with it.
Of course, we cannot directly use disjunctive TGDs in order to guarantee
that $E_H$ and $E_V$ correctly represent the horizontal and vertical edges 
of the grid.
However, we can introduce additional hidden relations that make it 
possible to mark certain nodes so as to expose the possible violations.
We first show how to expose violations to the fact that the horizontal
edge relation is a function.
The idea is to select nodes in $E_H$ in order to challenge functionality.
Formally, the horizontal challenge is captured by a hidden ternary relation
$\kw{HChallenge}_{\kw{funct}}$, by the linear TGDs
$$
\begin{array}{rclrcl}
  \kw{Init}                &\then& \exists ~x~y~y'~ 
                                   \kw{HChallenge}(x,y,y')   \\[1ex]
  \kw{HChallenge}(x,y,y')  &\then& E_H(x,y) ~\et~ E_H(x,y')
\end{array}
$$
and by the CQ
$$
  Q_H ~=~ \exists ~x~y~ \kw{HChallenge}_{\kw{funct}}(x,y,y) \ .
$$
Note that if the visible fact $\kw{Init}$ is present and the relation $E_H$ 
correctly describes the horizontal edges of the grid, then the above query 
$Q_H$ is necessarily satisfied by any instance of $\kw{HChallenge}_{\kw{funct}}$ 
that satisfies the above sentences: the only way to give a non-empty 
instance for $\kw{HChallenge}_{\kw{funct}}$ is to use triples of the 
form $(x,y,y)$. 
Conversely, if the relation $E_H$ is not a function, namely, if there exist 
nodes $x,y,y'$ such that $(x,y),(x,y')\in E_H$ and $y\neq y'$, then the singleton 
instance $\{ (x,y,y') \}$ for the hidden relation $\kw{HChallenge}_{\kw{funct}}$ 
will satisfy the associated setennces of the background theory and violate the query $Q_H$.
Note that we do not require that the relation $E_H$ is injective 
(this could be still done, but is not necessary for the reduction).
Similarly, we can use a hidden relation $\kw{VChallenge}$ and analogous 
background theory sentences and query $Q_V$ in order to challenge the functionality of 
$E_V$.

In the same way, we can challenge the confluence of the relations $E_H$ and $E_V$.
For this, we introduce a hidden relation $\kw{CChallenge}$ of arity $5$, which is
associated with the background theory sentences
$$
\begin{array}{rclrcl}
  \kw{Init}                             
    &\!\!\!\then&\!\!\! \exists ~x~y~z~w~w'~ 
                        \kw{CChallenge}(x,y,z,w,w')   \\[1ex]
  \kw{CChallenge}(x,y,z,w,w') 
    &\!\!\!\then&\!\!\! E_H(x,y) ~\et~ E_V(x,z) ~\et~ E_V(y,w) ~\et~ E_H(z,w')
\end{array}
$$
and the CQ
$$
  Q_C ~=~ \exists ~x~y~z~w~ \kw{CChallenge}(x,y,z,w,w) \ .
$$
As before, we can argue that there is a positive query implication for $Q_C$ iff
the horizontal and vertical edge relations are confluent, that is,
$(x,w)\in E_H \circ E_V$ and $(x,w') \in E_V \circ E_H$ imply $w=w'$.

We need now to ensure that every node is labeled with at most one tile, 
or equally that there are no relations $U_{t}$ and $U_{t'}$, for distinct tiles $t\neq t'\in T$,
that have non-empty intersection.
%such that $n$ is not in $U_t(\visinst)$ and $U_{t'}(\visinst)$.
For that we add the two following sentences, where $A$ and $B$ are hidden relations
\[
\begin{aligned}
  \kw{Init} ~&\then~ \exists ~x~ A(x) \vee B(x) \\
  B(x) ~&\then~  \bigvee\nolimits_{\!\!\!t \neq t'~~}  (~ U_t(x)  \wedge U_{t'}(x))
\end{aligned}
\]
Finally, we add the CQ 
\[
  Q_A ~=~ \exists ~x~ A(x)
\]

%We can also easily check that each node has at most one associated tile
%using the UCQ 
%$$
 % Q_T ~=~ \bigvee_{t\neq t' \,\in\, T} ~\exists ~x~ \big( U_t(x) ~\et~ U_{t'}(x) \big) \ .
%$$

Now that we described all the visible and hidden relations of the schema $\bfS$,
and the associated sentences $\calC$, we define the query for the $\exists\PSB$
problem as the conjunction of the atom $\kw{Init}$ and all previous UCQs 
(for this we distribute the disjunctions and existential quantifications 
over the conjunctions):
\[
  Q ~=~ \kw{Init} ~\et~ Q_A ~\et~  Q_H ~\et~ Q_V ~\et~ Q_C \ . %~\et~ Q_T \ .
\]
It remains to show that $\exists\PSB(Q,\calC,\bfS)=\true$ iff there is
a correct tiling of the infinite grid, namely, a function $f:\bbN\times\bbN \then T$ 
that satisfies the conditions 1), 2), and 3) above.

Suppose there is a correct tiling $f:\bbN\times\bbN \then T$. 
We construct the visible instance $\visinst$ that contains the fact
$\kw{Init}$ and the relations $E_H$, $E_V$, and $U_t$ with the intended semantics: 
$E_H = \big\{ \big((i,j),(i+1,j)\big) ~\big|~ i,j\in\bbN \big\}$,
$E_V = \big\{ \big((i,j),(i,j+1)\big) ~\big|~ i,j\in\bbN \big\}$, and
$U_t = \big\{ (i,j) ~\big|~ f(i,j) = t \big\}$ for all $t\in T$. 
Since no error can be exposed on the relations $E_H$, $E_V$, and $U_t$,
no matter how we construct a full instance $\fullinst$ that agrees with 
$\visinst$ on the visible part and satisfies the sentences in $\calC$,
we will have that $\fullinst$ satisfies all the components of the query $Q$, other than $Q_A$.
In addition, in any such $\fullinst$, $B$ must be empty, since otherwise tiling predicates for distinct
tiles would overlap, which is not the case. Since $\kw{Init}$ holds, we can conclude
via the first sentence above that $Q_A$ must hold.

Conversely, suppose that $\exists\PSB(Q,\calC,\bfS) = \true$ and let $\visinst$ be the
witnessing visible instance. Clearly, $\visinst$ contains the fact $\kw{Init}$ 
(otherwise, the query would be immediately violated).
%and it does not contain $W$, otherwise the query $Q_A$ would be violated. 
We can use the content of $\visinst$ and the knowledge that 
$\exists\PSB(Q,\calC,\bfS) = \true$ to inductively construct 
a correct tiling of the infinite grid. 
More precisely, by the first sentence in $\calC$, we know that 
$\visinst$ contains the fact $U_{t_\bot}(x)$, for some node $x$. 
Accordingly, we define $i_x = 0$, $j_x = 0$, and $f(i_x,j_x) = t_\bot$.
For the induction step, suppose that $f(i_x,j_x)$ is defined for a node $x$ 
with the associated coordinates $i_x$ and $j_x$. The sentences in $\calC$ 
enforce the existence of two cells $y$ and $z$ and two tiles $t$ and $t'$ 
for which the following facts are in the visible instance:
$E_H(x,y)$, $E_V(x,z)$, $U_t(y)$, and $U_{t'}(z)$.
Accordingly, we let $i_y = i_x + 1$, $j_y = j_x$, $i_z = i_x$, $j_z = j_y + 1$,
$f(i_y,j_y) = t$, and $f(i_z,j_z) = t'$. 
By the initial sentences in $\calC$, we know that the tiles associated
with the new cells $(i_y,j_y)$ and $(i_z,j_z)$ are consistent with the 
tile in $(i_x,j_x)$ and with the horizontal and vertical constraints $H$ and $V$.
We now argue that there is a unique choice for the nodes $y$ and $z$.
Indeed, suppose this is not the case; for instance, suppose that
there exist two distinct nodes $y,y'$ that are connected to $x$
via $E_H$. Then, we could construct a full instance in which
the relation $\kw{HChallenge}$ contains the single triple $(x,y,y')$. 
This will immediately violate the CQ $Q_H$, and hence $Q$. 
Similar arguments apply to the vertical successor $z$. 

We now argue that  there are unique choices for the tile $t$ 
associated with a node $y$. Suppose not. Then
we can let $A$ be empty and $B$ the set of all nodes with multiple tiles. 
All the sentences in $\calC$ are satisfied, but the query $Q_A$ is not. 
This contradicts the assumption that we have a $\PSB$.

Finally, we can argue along the same lines that, during the next steps of the 
induction, the $E_V$-successor of $y$ and the $E_H$-successor of $z$ coincide.
The above properties are sufficient to conclude that the constructed function 
$f$ is a correct tiling of the infinite grid.

The variant for finite instances is done by observing that the same reduction produces
a periodic grid, which can be represented as a finite instance.
\end{proof}

Perhaps even more surprisingly, 
we show that \emph{disjunction can be simulated using constants (under UNA)}. 
The proof
works by applying the technique of ``coding Boolean operations and truth values in the schema'' 
which has been used to eliminate the need for disjunction in hardness proofs in several past works
(e.g. \cite{georgchristos}). It is also similar to the proof idea used in Lemma 
\ref{lem:ucqtocqpsbinstanceid} from earlier in this paper.

\begin{proposition}\label{prop:existsPSB-disjunctions-to-constants}
There is a polynomial time reduction from $\exists\PSB(Q,\calC,\bfS)$, where $Q$ 
ranges over Boolean UCQs and $\calC$ over sets of disjunctive linear TGDs, 
to $\exists\PSB(Q',\calC',\bfS')$, where $Q'$ ranges over Boolean UCQs 
and $\calC'$ over sets of linear TGDs (with constants).
\end{proposition}

\begin{proof}
We transform the schema $\bfS$ to a new schema $\bfS'$ as follows.
For every visible (resp., hidden) relation $R$ of $\bfS$ of arity $k$, 
we add to $\bfS'$ a corresponding visible (resp., hidden) relation $R'$ 
of arity $k+1$. The idea is that the additional attribute of $R'$ 
represents a truth value, i.e. either the constant $0$ or the constant $1$, 
which indicates the presence of a tuple in the original relation $R$. 
For example, the fact $R'(\bar a,1)$ indicates the presence of the tuple 
$\bar a$ in the relation $R$. We can then simulate the disjunctions in 
the sentences of $\calC$ by using conjunctions and an appropriate look-up 
table, which we denote by $\kw{Or}$. Formally, we introduce three additional 
relations $\kw{Or}$, $\kw{Check}$, and $\kw{Init}$, of arities $2$, $1$, and $0$, 
respectively, and we let $\kw{Or}$ and $\kw{Init}$ be visible and $\kw{Check}$ 
be hidden in $\bfS'$.
Then, for every disjunctive linear TGD in $\calC$ of the form 
$$
  R(\bar x) ~\then~ \exists \bar y ~ S(\bar z) \vel T(\bar z')
$$
we add to $\calC'$ the linear TGD with constants
$$
  R'(\bar x,1) ~\then~ 
  \exists \bar y ~ b_1 ~ b_2 ~ S'(\bar z,b_1) \et T'(\bar z',b_2) \et \kw{Or}(b_1,b_2) \ .
$$
We further add to $\calC'$ the following sentences:
$$
\begin{aligned}
% Gabriele: I know that the third conjunct is not needed, 
% but it makes the translation look more ``natural'' and the proofs slightly easier.
  \kw{Init} &~\then~ \kw{Or}(0,1) \et \kw{Or}(1,0) \et \kw{Or}(1,1) \\[1ex]
  \kw{Init} &~\then~ \exists b_1 ~ b_2 ~ \kw{Or}(b_1,b_2) \et \kw{Check}(b_1) \et \kw{Check}(b_2) \ .
\end{aligned}
$$
Finally, we transform every CQ of $Q$ of the form $\exists \bar y ~ S(\bar y)$
to a corresponding CQ of $Q'$ of the form 
$$
  \exists \bar y ~ S'(\bar y,1) \et \kw{Check}(1) \et \kw{Init}
$$
Note that if needed, we can even rewrite the CQ above so as to avoid constants: we introduce
another hidden unary relation $\kw{One}$ and the sentence $\kw{Init} \then \kw{One}(1)$, 
and we replace the conjunct $\kw{Check}(1)$ with $\exists b ~ \kw{Check}(b) \et \kw{One}(b)$.
Below, we prove that $\exists\PSB(Q,\calC,\bfS)=\true$ iff
$\exists\PSB(Q',\calC',\bfS')=\true$.

For the easier direction, we consider a realizable $\bfS_v$-instance $\visinst$ 
such that $\PSB(Q,\calC,\bfS,\visinst)=\true$. We can easily transform $\visinst$
into a realizable $\bfS'_v$-instance $\visinst'$ that satisfies 
$\PSB(Q',\calC',\bfS',\visinst')=\true$.
For this it suffices to copy the content of the visible relations of $\visinst$ 
into $\visinst'$, by properly expanding the tuples with the constant $1$, and 
then adding the facts $\kw{Init}$, $\kw{Or}(0,1)$, $\kw{Or}(1,0)$, and $\kw{Or}(1,1)$.

As for the converse direction, we consider a realizable $\bfS'_v$-instance $\visinst'$
such that $\PSB(Q',\calC',\bfS',\visinst')=\true$. By the definition of $Q'$ it is 
clear that $\visinst'$ contains the fact $\kw{Init}$, and hence also the facts 
$\kw{Or}(0,1)$, $\kw{Or}(1,0)$, and $\kw{Or}(1,1)$.
We first claim that it suffices to 
show that 
%the relation $\kw{Or}$ contains no other tuples besides 
%$(0,1)$, $(1,0)$, and $(1,1)$ -- that is, 
for every fact $\kw{Or}(b_1,b_2)$ in 
$\visinst'$, we have $b_1=1$ or $b_2=1$. If this were the case, then we could easily transform $\visinst'$ 
into a realizable $\bfS_v$-instance $\visinst$ that satisfies $\PSB(Q,\calC,\bfS,\visinst)=\true$.
For this we simply select the facts $R'(\bar a,1)$ in $\visinst'$, where $R$ 
is a visible relation of $\bfS$, and project away the constant $1$. 

Thus it  remains to show that for every fact $\kw{Or}(b_1,b_2)$ in
$\visinst'$, we have $b_1=1$ or $b_2=1$. 
For the sake of contradiction,
suppose that $\visinst'$ contains a fact of the form $\kw{Or}(b_1,b_2)$, with 
$b_1\neq 1$ and $b_2\neq 1$. 
Since $\visinst'$ is realizable, there is a full $\bfS'$-instance $\fullinst'$ 
such that $\fullinst'\sat\calC'$ and $\visible(\fullinst') = \visinst'$. 
Note that $\fullinst'$ may satisfy $Q'$ and, in particular, the conjunct $\kw{Check}(1)$. 
However, removing the single fact $\kw{Check}(1)$ from $\fullinst'$ gives a new 
instance $\fullinst''$ that still satisfies the sentences in $\calC'$, 
agrees with $\fullinst'$ on the visible part, and violates the query $Q'$. 
This contradicts the fact that $\PSB(Q',\calC',\bfS',\visinst')=\true$.
\end{proof}

From the previous two results we immediately see that the addition of (distinct) constants  leads to undecidability:

\begin{corollary}\label{cor:existPSB-constants-undecidable}
The problem $\exists\PSB(Q,\calC,\bfS)$ is undecidable 
as $Q$ ranges over Boolean CQs and $\calC$ over sets of 
linear TGDs (with constants). 
\end{corollary}

We now turn to analysing how the complexity scales with less powerful
background theories, e.g. linear TGDs without constants. As before, we reduce 
$\exists\PSB(Q,\calC,\bfS)$ to $\PSB(Q,\calC,\bfS,\visinst_{\{a\}})$. 
We can then reuse some ideas from \cite{johnsonklug} to solve the 
latter problem in polynomial space:

\begin{theorem}\label{thm:linearPSB-in-pspace}
The problem $\PSB(Q,\calC,\bfS,\visinst_{\{a\}})$ 
as $Q$ ranges over Boolean UCQs and $\calC$ over sets of linear TGDs without constants,
is in $\pspace$,  and the same is true for
$\exists\PSB(Q,\calC,\bfS)$.
\end{theorem}

\begin{proof}
By Proposition \ref{prop:PSB-characterizion}, 
$\PSB(Q,\calC,\bfS,\visinst_{\{a\}})=\true$ is equivalent to checking that 
there is a homomorphism $h$ from $\canondb(Q_i)$
of some CQ $Q_i$ of $Q$ to the instance 
$\chase(\calC,\bfS,\visinst_{\{a\}})$.
We can easily guess in $\np$
a CQ $Q_i$ of $Q$, some homomorphism 
$h$ from  $\canondb(Q_i)$, and the corresponding image 
$I$ of  $\canondb(Q_i)$ under $h$. Then, it remains to decide 
whether $I$ is contained in 
$\chase(\calC,\bfS,\visinst_{\{a\}})$.
Below, we explain how to decide this in polynomial space.

Recall that the instance 
$\chase(\calC,\bfS,\visinst_{\{a\}})$
is obtained as the limit of a series of operations that consist of alternatively 
adding new facts according to the TGDs in $\calC$ and identifying the values that 
appear in some visible relation with the constant $a$.
Note that the second type of operation may also affect tuples that belong
to hidden relations (this happens when the values are shared with facts
in the visible instance). Also note that the affected tuples could 
have been inferred during previous steps of the chase. 
Nonetheless, at the exact moment when a new fact $R(b_1,\ldots,b_k)$ is 
inferred by chasing a linear TGD, we can detect whether a certain value 
$b_i$ needs to be eventually identified with the constant $a$, and in this 
case we can safely replace the fact $R(b_1,\ldots,b_k)$ with 
$R(b_1,\ldots,b_{i-1},a,b_{i+1},\ldots,b_k)$. More precisely, to decide 
whether the $i$-th attribute of $R(\bar b)$ needs to be instantiated with 
the constant $a$, we test whether $\calC$ entails a dependency of the form
$R(\bar x) \:\rightarrow\: \exists \bar y ~ S(\bar z)$,
where $\bar x$ is a sequence of (possibly repeated) variables that has the same 
equality type as $\bar b$ (i.e. $\bar x(j) = \bar x(j')$ iff $\bar b(j) = \bar b(j')$), 
$S$ is a visible relation, $\bar z$ is a sequence of variables among $\bar x,\bar y$, 
and $\bar x(i) = \bar z(j)$ for some $1\le j\le |\bar z|$.
% -- in this case we know that 
%the $i$-th attribute needs to be instantiated with the constant $a$. 
Note that the above entailment can be rephrased as a containment problem between 
two CQs -- i.e. $R(\bar x)$ and $\exists \bar y ~ S(\bar z)$ -- under a given 
set of linear TGDs $\calC$, and we know from \cite{johnsonklug} that the 
latter problem is in $\pspace$. 
We also observe that, in order to discover all the values in $R(\bar b)$ 
that need to be identified with the constant $a$, it is not sufficient
to execute the above analysis only once on each position $1\le i\le\arity(R)$, 
as identifying some values with the constant $a$ may change the equality 
type of the fact and thus trigger new dependencies from $\calC$ 
(notably, this may happen when the linear TGDs are not IDs). 
We thus repeat the above analysis on all positions of $R$ and until the 
corresponding equality type stabilizes -- this can be still be done in 
polynomial space. After this, we add the resulting fact to the chase.

What we have just described is an alternative construction of 
$\chase(\calC,\bfS,\visinst_{\{a\}})$ 
%\gabriele{Change this: put a dot, then say each chase step can be done in pspace}
in which every chase step can be done using a $\pspace$ sub-procedure.
% to detect ``on-the-fly'' which 
%values need to be instantiated with the constant $a$, so as to
%avoid chasing the special EGDs induced by the visible instance.
We omit the routine details showing that this alternative construction
gives the same result, in the limit, as the version of the chase that 
we introduced at the beginning of Section \ref{sec:existsPSB}
(the arguments are similar to the proof of Lemma \ref{lem:disjunctive-chase-universal}).

Below, we explain how to adapt the techniques from \cite{johnsonklug}
to this alternative variant of the chase, in order to decide 
whether the homomorphic image $I$ of some CQ of $Q$ is contained in 
$\chase(\calC,\bfS,\visinst_{\{a\}})$.
For this, it is convenient to think of 
$\chase(\calC,\bfS,\visinst_{\{a\}})$ 
as a directed graph, where the nodes represent the facts in 
$\chase(\calC,\bfS,\visinst_{\{a\}})$ 
and the edges describe the inference steps that derive new facts 
from existing facts and sentences in $\calC$.
Note that, because the background theory sentences are linear TGDs, 
each inference step depends on at most one fact.
In particular, the nodes of this graph that have no incoming edge 
(we call them \emph{roots}) are precisely the facts from the instance 
$\visinst_{\{a\}}$, and all the other nodes are reachable from some root.
Moreover, by the previous arguments, one can check in polynomial 
space whether an edge exists between two given nodes.

Now, we focus on the minimal set of edges that connects all the facts of $I$
to some roots in the graph. The graph restricted to this set of edges
is a forest, namely, every node in it has at most one incoming edge.
Moreover, the height of this forest is at most exponential in $|I|$, 
and each level in it contains at most $|I|$ nodes. 
Thus, the restricted graph can be explored by a non-deterministic 
polynomial-space algorithm that guesses the nodes at a level on 
the basis of the nodes at the previous level and the linear TGDs in $\calC$.
The algorithm terminates successfully once it has visited all 
the facts in $I$, witnessing that $I$ is contained in 
$\chase(\calC,\bfS,\visinst_{\{a\}})$.
Otherwise, the computation is rejected after seeing 
exponentially many levels.
\end{proof}

\smallskip
We can derive matching lower 
bounds by reducing Open-World Query Answering (OWQ) to $\exists\PSB$:

\begin{proposition} \label{prop:OWtoexistPSB} 
For any class of sentences containing 
%IDs 
linear TGDs,
%(thus all classes considered here),
$\OWQ$ reduces to $\exists\PSB$.
\end{proposition}

\begin{proof}
Let $Q$ be a query, $\calC$ a set of sentences over a schema $\bfS$, 
and $\fullinst$ an instance of the schema $\bfS$. We show how to reduce 
the Open-World Query Answering problem for $Q$, $\calC$, $\bfS$, and $\fullinst$ 
to a problem $\exists\PSB(Q',\calC',\bfS')$.
The idea is to create a copy of the instance $\fullinst$ in the hidden 
part of the schema, which can  then be  extended arbitrarily. 

Formally, we let the transformed schema $\bfS'$ consist of all the relations in 
$\bfS$, which are assumed to be hidden, plus an additional visible relation 
$\good$ of arity $0$. We then introduce a variable $y_b$ for each value 
in the active domain of $\fullinst$, and we let $\calC'$ contain all the 
sentences from $\calC$, plus the sentence 
$\good \:\rightarrow\: \exists \bar y ~ Q_\fullinst$,
where $\bar y$ contains one variable $y_b$ for each value $b$ in the active domain
of $\fullinst$ and $Q_\fullinst$ is the conjunction of the atoms of the form 
$A(y_{b_1},\ldots,y_{b_k})$, for all facts $A(b_1,\ldots,b_k)$ in $\fullinst$. 
Note that the visible instance $\visinst_\good$ that contains the atom $\good$ is realizable, 
since it can be completed (using the chase) to an $\bfS'$-instance $\fullinst'$
that satisfies the sentences in $\calC'$.
Let $Q'= Q \et \good$. We claim that $\exists\PSB(Q',\calC',\bfS')=\true$ if and only
if $Q$ is certain with respect to $\calC$ on $\fullinst$. In one direction, suppose $\exists \PSB(Q',\calC',\bfS')=\true$  holds.
The witness visible instance having $\PSB$ can only be the instance $\visinst_\good$.
Consider an instance $\fullinst'$ containing all facts of $\fullinst$ and satisfying
the original sentences $\calC$. By setting $\good$ to true in $\fullinst'$, we have an instance
satisfying $\calC'$, and since $\visinst_\good$ has a $\PSB$ then
we know that this instance must satisfy $Q'$ and hence $Q$.  Thus $Q$ is certain with respect to
$\calC$ on $\fullinst$ as required.
Conversely, suppose $Q$ is certain with respect to
$\calC$ on $\fullinst$. Letting $C_\fullinst$ be the chase of $\fullinst$ with respect to $\calC$,
we see that $C_\fullinst$ satisfies $Q$.  We will show there is a $\PSB$ for $Q',\calC',\bfS$ on $\visinst_\good$.
Thus fix an instance $\fullinst'$ where $\good$ and $\calC'$ holds. The additional sentence
implies that $\fullinst'$ contains the
image of $\fullinst$ under some homomorphism $h$. But $h$ extends to a homomorphism of $C_\fullinst$ into
$\fullinst'$. Thus $\fullinst'$ satisfies $Q$, and therefore satisfies $Q'$.
Thus there is a $\PSB$ on $\visinst_\good$ as required.

Thus we have reduced the Open-World Query Answering problem for $Q$, $\calC$, and $\bfS$ 
to the problem $\exists\PSB(Q',\calC',\bfS')$.
\end{proof}

\smallskip
From this and existing lower bounds on the Open-World Query Answering 
(\cite{casanova} coupled with a reduction from implication to $\OWQ$ for 
%IDs,
linear TGDs,  
\cite{taming} for $\fgtgd$s), we see that the prior upper bounds from
Theorem \ref{thm:linearPSB-in-pspace} and Corollary \ref{cor:decidexistpsbgtgd} 
are tight:

\begin{corollary}\label{cor:linearExistPSB-pspace-hard}
The problem $\exists\PSB(Q,\calC,\bfS)$, where $Q$ ranges over 
CQs and $\calC$ over sets of 
%IDs, 
linear TGDs,
is $\pspace$-hard.
\end{corollary}

\begin{corollary}\label{cor:fgtgdsrPSB-2exptime-hard}
The problem $\exists\PSB(Q,\calC,\bfS)$, where $Q$ ranges over
CQs and $\calC$ over sets of $\fgtgd$s without constants, is $2\exptime$-hard.
\end{corollary}

%% file: possummary.tex
\subsection{Summary for Positive Query Implication}

The main results on positive query implication are highlighted in the table below.

\begin{center}
\scriptsize
\begin{tabular}{@{}l|lll@{}}
 Background Theory $\Sigma$ & $\PSB$ data complexity & $\PSB$ combined complexity & $\exists\PSB$ \\[1ex]
  \hline\\
  $\cf$        & $\exptime$-cmp    & $\twoexp$-cmp & $\pspace$-cmp\\
  ~~Linear TGD & Thm.~\ref{thm:expdatacomplexitypsb}~/~Thm \ref{th:psb-instance-hard} 
               & Thm.~\ref{thm:gnfoposinstancedecid}~/~Thm \ref{thm:psbidtwoexphardcomb} 
               & Thm.~\ref{thm:linearPSB-in-pspace}~/~Cor \ref{cor:linearExistPSB-pspace-hard}\\[1ex]
  $\cf$        & $\exptime$-cmp    & $\twoexp$-cmp & $\twoexp$-cmp\\
  ~~$\fgtgd$   & Thm.~\ref{thm:expdatacomplexitypsb}~/~Thm \ref{th:psb-instance-hard}
               & Thm.~\ref{thm:gnfoposinstancedecid}~/~Thm \ref{thm:psbidtwoexphardcomb} 
               & Cor.~\ref{cor:decidexistpsbgtgd}/Cor \ref{cor:fgtgdsrPSB-2exptime-hard}\\[1ex]
  $\cf$ Disj.  & $\exptime$-cmp    & $\twoexp$-cmp & undecidable\\
  ~~Linear TGD & Thm.~\ref{thm:expdatacomplexitypsb}~/~Thm \ref{th:psb-instance-hard} 
               & Thm.~\ref{thm:gnfoposinstancedecid}~/~Thm \ref{thm:psbidtwoexphardcomb} 
               & Thm.~\ref{thm:undecidexistpsbdisjunction}\\[1ex]
  \hline\\
  %Disj.        & $\exptime-cmp$    & $\twoexp$ & undecidable\\ 
 % ~~$\fgtgd$   & Thm \ref{thm:expdatacomplexitypsb}/Thm \ref{th:psb-instance-hard}  
%               & Thm \ref{thm:gnfoposinstancedecid}/Thm \ref{thm:psbidtwoexphardcomb}  
%               & Thm \ref{thm:undecidexistpsbdisjunction}\\[0.5ex]
%  \hline\hlineskip
%michael: I think the table is a bit too large
  Linear TGD    & $\exptime$-cmp    & $\twoexp$-cmp & undecidable\\
  \& $\fgtgd$   & Thm.~\ref{thm:expdatacomplexitypsb}~/~Thm.~\ref{th:psb-instance-hard} 
                & Thm.~\ref{thm:gnfoposinstancedecid}~/~Thm.~\ref{thm:psbidtwoexphardcomb}  
                & Cor.~\ref{cor:existPSB-constants-undecidable}\\
  \& GNFO       & & &          \\[1ex]
  %$\fgtgd$     & $\exptime$-cmp    & $\twoexp$ & undecidable\\
 %              & Thm \ref{thm:expdatacomplexitypsb}/Thm \ref{th:psb-instance-hard} 
 %              & Thm \ref{thm:gnfoposinstancedecid}/Thm \ref{thm:psbidtwoexphardcomb}
 %              & Cor \ref{cor:existPSB-constants-undecidable}\\[0.5ex]
  % \hline\hlineskip
  % $\gnfo$      &  $\exptime$-cmp   & $\twoexp$-cmp & undecidable\\
  %              & Thm \ref{thm:expdatacomplexitypsb}/Thm \ref{th:psb-instance-hard} 
  %              & Thm \ref{thm:gnfoposinstancedecid}/Thm \ref{thm:psbidtwoexphardcomb}  
  %              & Thm \ref{thm:undecidexistpsbdisjunction}\\[0.5ex]
  \hline
\end{tabular} \label{fig:possum}
\end{center}

%% file: neg.tex
\section{Negative Query Implication} \label{sec:negative}

\input{neginstance}

\input{negexists}

\input{negsummary}

%% file: neginstance.tex
\subsection{Instance-level problems} \label{subsec:neginstance}

Here we analyze the complexity of the problem $\NSB(Q,\calC, \bfS, \visinst)$.
As in the positive case, we begin with an upper bound that holds for a very rich 
class of background theories, which go far beyond referential constraints (and  $\fgtgd$s).

\begin{theorem} \label{thm:nsbgnf}
The problem $\NSB(Q,\calC, \bfS, \visinst)$,
as $Q$ ranges over Boolean UCQs and $\calC$ over sets  of $\gnfo$ sentences,
has $\twoexp$ combined complexity, $\exptime$ data complexity, and it is finitely controllable. 
\end{theorem}

\begin{proof}
As in the positive case, we reduce to unsatisfiability of a $\gnfo$ formula.
We use a  variation of the same formula, where $\neg Q$ is now replaced by $Q$:
%$\nsbof(Q, \calC, \bfS, \visinst)$ defined as:
$$
\begin{aligned}
  \phi^\nsbof_{Q,\calC,\bfS,\visinst} ~=~
  & Q ~\et~ \calC ~\et~ \\[-1ex]
  & \!\!\!\!\!
    \bigwedge_{R\in \bfS_v} \!\! \Big( \!\!\bigwedge_{R(\bar{a})\in\visinst}\!\!\!\!\! R(\bar{a}) 
                                  \:\et\: 
                                  \forall \bar{x} ~ \big( R(\bar{x}) \then \!\!\!
                                                          \bigvee_{R(\bar{a})\in\visinst} \!\!\!\!
                                                          \bar{x}=\bar{a} \big) \Big)
\end{aligned}
$$
The data complexity analysis is as in Theorem \ref{thm:expdatacomplexitypsb}, since
the formulas agree on the part that varies with the instance.
\end{proof}

We can show that this bound is tight if the class of background theories is rich enough. 
This will follow from our lower bounds for positive query implication problems, since we can show
that $\NSB$ is at least as difficult as $\PSB$ for sentences in a powerful logical language.

\begin{theorem}\label{th:PSBtoNSB}
For any class of sentences that include connected $\fgtgd$s and for any UCQ $Q$,
$\PSB(Q,\calC,\bfS,\visinst)$ reduces in polynomial time to $\NSB(Q',\calC',\bfS',\visinst')$.
When $Q, \calC, \bfS$ are fixed in the input to this reduction, then $Q', \calC', \bfS'$ are fixed
in the output.
\par\noindent
Thus, for these background theories, the lower bounds for combined and data complexity given in
Theorems \ref{th:psb-instance-hard} and \ref{thm:psbidtwoexphardcomb} 
apply to negative query implications as well.
\end{theorem}

\begin{proof}
We first provide a reduction that works with any class of  TGDs
allowing arbitrary conjunctions in the left-hand sides (e.g. frontier-guarded TGDs).
Subsequently, we show how to modify the constructions in order to preserve 
connectedness.

The schema $\bfS'$ is obtained by copying both the visible and the hidden relations 
from $\bfS$ and by adding the following relations: a visible relation 
$\kw{Error}$ of arity $0$ and a hidden relation $\kw{Good}$ of arity $0$. 
The sentences $\calC'$ will contain the  sentences from $\calC$, 
plus one frontier-guarded TGD of the form
$$
  Q_i(\bar y) \et \kw{Good} ~\then~ \kw{Error}
$$
for each disjunct $\exists \bar y ~ Q_i(\bar y)$ of the UCQ $Q$.
Finally, the query and the visible instance for $\NSB$ are defined as follows:
$Q'=\kw{Good}$ and $\visinst' = \visinst$ (in particular, we initialize 
the visible relation $\kw{Error}$ with the empty set).

We now verify that $\PSB(Q,\calC,\bfS,\visinst)=\false$ iff $\NSB(Q',\calC',\bfS',\visinst')=\false$.
Suppose that $\PSB(Q,\calC,\bfS,\visinst)=\false$, namely, that there 
is an $\bfS$-instance $\fullinst$ such that $\fullinst \nsat Q$, 
$\fullinst \sat \calC$, and $\visible(\fullinst) = \visinst$.
Let $\fullinst'$ be the $\bfS'$-instance obtained from $\fullinst$ by adding
the single hidden fact $\kw{Good}$. Clearly, $\fullinst'$ satisfies the query $Q'$
and also the sentences in $\calC'$. In particular, it satisfies every 
sentence $Q_i(\bar y) \et \kw{Good} \then \kw{Error}$ because $\fullinst$
violates every disjunct $\exists \bar y ~ Q_i$ of $Q$. 
Hence, we have $\NSB(Q',\calC',\bfS',\visinst')=\false$.
Conversely, suppose that $\NSB(Q',\calC',\bfS',\visinst')=\false$, namely,
that there is an $\bfS'$-instance $\fullinst'$ such that
$\fullinst'\sat Q'$, $\fullinst'\sat\calC'$, and $\visible(\fullinst') = \visinst'$.
By copying the content of $\fullinst'$ for those relations belong to the schema
$\bfS$, we obtain an $\bfS$-instance $\fullinst$ that satisfies the sentences $\calC$.
Moreover, because $\fullinst'$ contains the fact $\kw{Good}$ but not the fact 
$\kw{Error}$, $\fullinst'$ violates every conjunct $\exists \bar y ~ Q_i(\bar y)$ 
of $Q$, and so $\fullinst$ does. 
This shows that $\PSB(Q,\calC,\bfS,\visinst)=\false$.

We observe that the sentences in the above reduction use left-hand sides
that are not connected. In order to preserve connectedness, it is sufficient 
to modify the above constructions by adding a dummy variable that is shared 
among all atoms. 
More precisely, we expand the relations of the schema $\bfS$ and the relation
$\kw{Good}$ with a new attribute, and we introduce a new visible relation $\kw{Check}$ 
of arity $1$. 
The dummy variable will be used to enforce connectedness in the left-hand sides,
and the relation $\kw{Check}$ will gather all the values associated with the
dummy attribute. Using the visible instance, we can also check that the relation
$\kw{Check}$ contains exactly one value.
The sentences in the background theory are thus modified as follows. 
Every sentence
$R_1(\bar x_1) \et \ldots \et R_m(\bar x_m) \then \exists \bar y ~ S(\bar z)$ 
in $\calC'$ is transformed into 
$R_1(\bar x_1,w) \et \ldots \et R_m(\bar x_m,w) \then \exists \bar y ~ S(\bar z,w)$.
In particular, note that the sentence
$Q_i(\bar y) \et \kw{Good} \then \kw{Error}$ becomes
$Q_i(\bar y,w) \et \kw{Good}(w) \then \kw{Error}(w)$,
which is now a connected frontier-guarded TGD.
Furthermore, for every relation $R(\bar x)$ in $\bfS$, we add the sentence
$$
  R(\bar x,w) ~\then~ \kw{Check}(w)
$$
and we do the same for the relation $\kw{Good}$:
$$
  \kw{Good}(w) ~\then~ \kw{Check}(w) \ .
$$
Finally, the query is transformed into $Q' = \exists w ~ \kw{Good}(w)$ and
the visible instance $\visinst'$ is expanded with a fresh dummy value $a$
on the additional attribute and with the visible fact $\kw{Check}(a)$.
\end{proof}

%As a matter of fact, this also shows that, in the presence of frontier-guarded
%TGDs, we cannot hope to reduce the instance-based $\NSB$ problem to the 
%evaluation of a Datalog program, as this would give a polynomial-time
%upper bound for data complexity. On the other hand, we will see later 
%(Theorem \ref{th:definability-active-domain-controllable}) that an approach 
%based on Datalog is possible when we restrict to constraints given by linear TGDs.

\smallskip
As mentioned above,
combining the above reduction with Theorems \ref{th:psb-instance-hard} and
\ref{thm:psbidtwoexphardcomb}, we get the following hardness results for
instance-based $\NSB$.

\begin{corollary}\label{cor:nsb-instance-hard}
There are a Boolean UCQ $Q$ and a set $\calC$ of IDs over a schema $\bfS$
for which the problem $\NSB(Q,\calC,\bfS,\visinst)$ is $\exptime$-hard
in data complexity (that is, as $\visinst$ varies over instances).
\end{corollary}

\begin{corollary}\label{cor:nsb-cb-hard}
The problem $\NSB(Q,\calC,\bfS,\visinst)$, as $\calC$ ranges
over  sets  of connected frontier-guarded TGDs,  $\bfS$  over schemas, $Q$  over conjunctive queries
and $\visinst$ over instances, is $\twoexp$-hard.
\end{corollary}

\medskip
Thus far, the negative query implication 
results have been similar to the positive ones.  We  will now show
a strong contrast in the case of IDs and linear TGDs.  Recall that the $\PSB$
problems were highly intractable even for fixed schema, query, and background theory.
We begin by showing that  $\NSB(Q,\calC,\bfS,\visinst)$ 
can be solved easily by looking only at full instances that agree with $\visinst$ 
on the visible part and whose active domains are almost the same as that of $\visinst$:

\begin{definition}\label{def:active-domain-controllable}
The problem $\NSB(Q,\calC,\bfS,\visinst)$ is said to be \emph{active domain controllable}
if it is equivalent to asking that for every instance $\fullinst$ 
{\sl over the active domain of $\visinst$}, 
if $\fullinst$ satisfies $\calC$ and $\visinst=\visible(\fullinst)$,
then $Q(\fullinst)=\false$.
\end{definition}

It is clear that the the problem $\NSB(Q,\calC,\bfS,\visinst)$ is simpler
when it is active domain controllable, as in this case we could guess 
a full instance $\fullinst$ over the active domain of $\visinst$ and
then reduce the problem to checking whether $Q$ holds on $\fullinst$.

We give a simple argument that $\NSB$ under IDs is active domain controllable.
Let $\calC$ be a set of IDs over a schema $\bfS$, $Q$ be a UCQ, and $\visinst$ 
be a visible instance such that $\NSB(Q,\calC,\bfS,\visinst)=\false$.
Without loss of generality --- that is, by adding a dummy visible fact over
a visible relation that does not occur in the sentences of the background theory ---
we can assume that the active domain $\adomain(\visinst)$ of $\visinst$
contains at least one element.
The fact that $\NSB(Q,\calC,\bfS,\visinst)=\false$ implies the existence 
of a full instance $\fullinst$ such that $\fullinst\sat\calC$, $
\visible(\fullinst)=\visinst$, and $\fullinst\sat Q$.
Now take any element $a \in \adomain(\visinst)$ and let 
$h$ be the homomorphism that is the identity over $\adomain(\visinst)$ and 
maps any other value from $\adomain(\fullinst)\setminus\adomain(\visinst)$ to $a$.
Since the sentences $\calC$ are IDs (in particular, since the left-hand side atoms
do not have 
%repeated occurrences of the same variable),
constants or repeated occurrences of the same variable), 
we know that $h(J)\sat\calC$. 
Similarly, we have $h(J)\sat Q$. Hence, $h(J)$ is an instance over the active domain 
of $\visinst$ that equally witnesses $\NSB(Q,\calC,\bfS,\visinst)=\false$.

Note that our hardness results for $\PSB$ (in particular, Theorem \ref{thm:psbidtwoexphardcomb}),
imply that $\PSB$ is \emph{not} active domain controllable even for IDs, since 
such a result would easily give membership in $\conp$.

The following example shows that linear TGDs are not always active domain controllable.

\begin{example}\label{ex:active-domain-controllability}
Let $\bfS$ be the schema with a hidden relation $R$ of arity $2$, with two visible 
relations $S,T$ of arities $1$, $0$, respectively, and with the sentences:
$$
  R(x,y) ~\then~ S(x) 
  \qquad\qquad
  R(x,x) ~\then~ T \ .
$$
Note that the sentences are linear TGDs and they are even \emph{full} 
-- no existential quantifiers on the right.
The conjunctive query is $Q = \exists x~y ~ R(x,y)$. 
Further let the visible instance $\visinst$ consist of the single fact $S(a)$.
Clearly, every full instance $\fullinst$ over the active domain $\{a\}$ that 
satisfies both $\calC$ and $Q$ must also contain the facts $R(a,a)$ and $T$, 
and so such an instance cannot agree with $\visinst$ in the visible part.
On the other hand, the instance that contains the facts $S(a)$ and $R(a,b)$, 
for a fresh value $b$, satisfies both $\calC$ and $Q$ and moreover agrees with 
$\visinst$. This shows that $\NSB(Q,\calC,\bfS,\visinst)$ is not active domain 
controllable.
\end{example}

The example shows that we need to weaken the notion of active domain controllability to allow
some elements outside of the active domain. The following definition allows a fixed number of exceptions.
\begin{definition}\label{def:almost active-domain-controllable}
For a number $k$, the problem $\NSB(Q,\calC,\bfS,\visinst)$ is said to be \emph{active domain controllable modulo $k$}
if it is equivalent to asking that for every instance $\fullinst$ whose active domain contains at most
$k$ elements outside of the active domain of $\visinst$,
if $\fullinst$ satisfies $\calC$ and $\visinst=\visible(\fullinst)$,
then $Q(\fullinst)=\false$.
\end{definition}

\input{almostadomcontrollable}

%\gabriele{Should we anticipate a bit the discussion that we have
%          later for adom controllability modulo $k$?}
Now we show how to exploit active domain controllability 
to prove that $\NSB$ problems can be solved not only efficiently,
but ``definably'' using well-behaved query languages.
For this, we introduce a variant of Datalog programs, called 
\emph{GFP-Datalog} programs, whose semantics is given by greatest 
fixpoints. 
GFP-Datalog programs are defined syntactically in 
the same way as Datalog programs \cite{AHV}, that is, as finite 
sets of rules of the form $U(\bar x) \:\leftarrow\: Q(\bar x)$
where the variables in $\bar x$ are implicitly universally quantified 
and $Q$ is a conjunctive query whose free variables are exactly $\bar x$.
As for Datalog programs, we distinguish between \emph{extensional} 
(i.e., input) predicates and \emph{intensional} (i.e., output) predicates.
In the above rules we restrict the left-hand sides
to contain only intensional predicates. 
Given a GFP-Datalog program $P$, the \emph{immediate consequence operator} 
for $P$ is the function that, given an instance $M$ consisting of 
both extensional and intensional relations, returns the 
instance $M'$ where the extensional relations are as in $M$
and the tuples of each intensional relation $U$ are 
those satisfying $Q(M)$, where $Q$ is any query appearing  on the right of a rule with $U$.
The immediate consequence operator is monotone, 
and the  semantics of the GFP-Datalog program on instance $I$ for the extensional
relations
is defined as the
greatest fixpoint of this operator starting at  the instance $I^+$ that extends
$I$ by setting each intensional relation ``maximally'' --- that is, to
the tuples of values from the active domain of $I$ plus the constants 
appearing in the GFP-Datalog program. 
A program may also include a distinguished intensional predicate, the \emph{goal predicate} $G$,
% OLD VERSION
%and then the result is taken to be the projection of the greatest fixpoint onto $G$.
% BALDER'S SUGGESTION
%in which case it defines the query that maps every instance 
%to the set of tuples belonging to $G$ in the greatest fixpoint.
% A MIX OF THE TWO THAT I PREFER
in which case it defines the query that maps every instance 
to the set of tuples satisfying $G$ in the greatest fixpoint.
We now show that under active domain controllability, we can 
use GFP-Datalog to decide $\NSB(Q,\calC,\bfS,\visinst)$:
   
\begin{theorem}\label{th:definability-active-domain-controllable}
If $Q$ is a Boolean UCQ, $\calC$ a set of linear TGDs (with constants), 
and $\NSB(Q,\calC,\bfS,\visinst)$ is active domain controllable, then 
$\neg\NSB(Q,\calC,\bfS,\visinst)$, viewed as a Boolean query over the 
visible part $\visinst$, is definable by a GFP-Datalog program that 
can be constructed in \ptime from $Q$, $\calC$, and $\bfS$.
\end{theorem}

\begin{proof}
First observe that $\NSB(Q,\calC,\bfS,-)$ can be seen as a Boolean function
that takes as input an instance $\visinst$ for the visible relations of $\bfS$
and returns $\true$ iff the query $Q$ does {\sl not} hold on {\sl every} instance
$\fullinst$ that satisfies the sentences $\calC$ and such that $\visible(\fullinst)=\visinst$.
Accordingly, $\neg\NSB(Q,\calC,\bfS,-)$ is the negation of the function $\NSB(Q,\calC,\bfS,-)$,
and thus maps an instance $\visinst$ to $\true$ when $Q$ {\sl does} hold on {\sl some}
instance $\fullinst$ that satisfies $\calC$ and agrees with $\visinst$ on the visible relations.

Below, we implement the function $\neg\NSB(Q,\calC,\bfS,-)$ by means of a GFP-Datalog program.
Thanks to active domain controllability, it is sufficient to consider only 
full instances constructed over the active domain of $\visinst$.
More precisely, it is sufficient to show that a witnessing instance $\fullinst$ 
can be obtained as a greatest fixpoint starting from the values in the 
active domain of $\visinst$. 
Below, we describe the GFP-Datalog program that computes $\fullinst$ starting from $\visinst$. 

%\gabriele{I tried to improve and give more intuition about the program. Please check!}
The extensional relations are the ones in the visible part $\visinst$, while
the intensional relations are the ones in the hidden part of the schema $\bfS$,
plus an extra intensional relation $A$ that collects the values in the active 
domain of $\visinst$. 
For each extensional (i.e.~visible) relation $R$ and each position $i\in\{1,\ldots,\arity(R)\}$,
we add the rule $A(x_i) \:\leftarrow\: R(\bar x)$,
which collects all the values of the active domain into the relation $A$.
In addition, for each intensional (i.e.~hidden) relation $R$, we have the rule
$$
  R(\bar x) ~~\leftarrow~~
  \bigwedge_i A(x_i) ~~\et 
  \bigwedge_{\substack{\text{linear TGD in  $\calC$ of the} \\ 
                       \text{form } R(\bar x) \:\then\: \exists\bar y\: S(\bar z)}} 
            S(\bar z) \ .
$$
Intuitively, the above rule permits the existence of a fact $R(\bar a)$ 
only when $\bar a$ consists of values from the active domain 
and every linear TGD $R(\bar x) \:\then\: \exists\bar y\: S(\bar z)$
of $\calC$ is satisfied by some fact $S(\bar b)$ when substituting $\bar x$ for $\bar a$. 
This semantics is consistent with the goal of finding the biggest instance $\fullinst$ 
over the active domain of $\visinst$ that satisfies the UCQ $Q$ --- so as to have 
$\NSB(Q,\calC,\bfS,\visinst)=\false$ --- while guaranteeing that the linear TGDs 
remain valid. 

We finally add the rule
\[
  \goal ~~\leftarrow~~ S_1(\bar z_1) \et \ldots \et S_n(\bar z_n)
\] 
for each CQ $\exists \bar y ~ S_1(\bar z_1) \et \ldots \et S_n(\bar z_n)$ of $Q$,
and take $\goal$ to be the final output of our program.

Let us now prove that the Datalog program does compute the function $\neg\NSB(Q,\calC,\bfS,-)$
under the greatest fixpoint semantics.
Consider an instance $\fullinst$ computed by the GFP-Datalog program starting 
from input $\visinst$.
Clearly, the extensional (visible) part of $\fullinst$ agrees with $\visinst$.
We claim that $\fullinst$ also satisfies the sentences in $\calC$.  Indeed,
if $R(\bar x) \:\then\: \exists\bar y\: S(\bar z)$ is a linear TGD in $\calC$
and $R(\bar a)$ is a fact of $\fullinst$, with $R(\bar a)$ image of $R(\bar x)$
via some homomorphism $h$, then $\fullinst$ contains a fact of 
the form $S(\bar b)$, where $\bar b$ is the image of $S(\bar z)$ via some
homomorphism $h'$ that extends $h$.
To conclude, we observe that the predicate $\goal$ holds iff $\fullinst$ 
satisfies some disjunct $S_1(\bar z_1) \et \ldots \et S_n(\bar z_n)$ of the UCQ $Q$, 
namely, iff $\NSB(Q,\calC,\bfS,\visinst)=\false$.
\end{proof}

In the case of Linear TGDs that are active domain controllable modulo $k$,
 we can similarly use a GFP Datalog program, but first pre-processing the active domain
to contain the $k$ additional constants. The extension of 
Theorem \ref{th:definability-active-domain-controllable} clearly holds:
\begin{theorem}\label{th:definability-active-domain-controllable-mod}
If $Q$ is a Boolean UCQ, $\calC$ a set of linear TGDs (with constants), 
and $\NSB(Q,\calC,\bfS,\visinst)$ is active domain controllable modulo $k$, then
$\neg\NSB(Q,\calC,\bfS,\visinst)$, viewed as a Boolean query over the
visible part $\visinst$, is definable by a GFP-Datalog program that
can be constructed in \ptime from $Q$, $\calC$, and $\bfS$.
\end{theorem}

Recall that the na\"ive fixpoint algorithm for a GFP-Datalog program 
takes exponential time in the maximum arity of the intensional relations, 
but only polynomial time in the size of the extensional relations and 
the number of rules. This is true even if one extends the active domain by $k$ elements, where
$k$ is the maximal arity.
Thus we can get bounds on the $\NSB$ problem for IDs
using the simple argument for active domain controllability for IDs given above
along with Theorem \ref{th:definability-active-domain-controllable}.
We can likewise get bounds for linear TGDs 
using Theorem \ref{thm:almostadomcontrollable}  and
 Theorem \ref{th:definability-active-domain-controllable-mod}.

\begin{corollary}\label{cor:complnsbinstlin}
%If   $\calC$ is restricted to range over sets of linear TGDs, 
%then $\NSB(Q,\calC,\bfS,\visinst)$ has data complexity in $\ptime$ and 
%combined complexity in $\exptime$.
When $\calC$ ranges over sets of linear TGDs and $Q$ over Boolean UCQs,
$\NSB(Q,\calC,\bfS,\visinst)$ has data complexity in $\ptime$ 
and combined complexity in $\exptime$.
\end{corollary}

\begin{example} \label{ex:return}
%\gabriele{I rephrased, trying to correct and improve the explanation}
Returning to the medical example from the introduction, Example \ref{ex:one}, we see that the GFP-Datalog program
is quite intuitive: since we have a referential constraint from $\appoint$ into $\patient$
and the visible instance does not contain the fact $\patient(\text{Smith})$,
%is empty in the instance and 
all tuples of the form $(\text{Smith},a,d)$ are removed from the relation $\appoint$.
The program then simply evaluates the query on the resulting instance, which returns false, 
indicating that an $\NSB$ does hold on the original visible instance.
\end{example}

%myeat

We give a tight $\exptime$ lower bound for the combined complexity of $\NSB$ with linear TGDs (and even IDs):

\begin{theorem}\label{thm:nsbinstcombexphard}
The combined complexity of $\NSB(Q,\calC, \bfS, \visinst)$, where 
$Q$ ranges over UCQs and $\calC$ ranges over IDs, is $\exptime$-hard.
\end{theorem}

\begin{proof}
We reduce the acceptance problem for an alternating $\pspace$ 
Turing machine $M$ to $\NSB(Q,\calC,\bfS,\visinst)$.
As in the proof of Theorem \ref{th:psb-instance-hard}, we assume that 
the transition function of $M$ maps each universal configuration to a 
set of exactly $2$ target configurations. Moreover, we assume that
there is at least one target configuration for each existential configuration.
In particular, $M$ never halts.
The computation begins with the head on the second
position and never visits the first and last position of the tape.
The acceptance condition of $M$ is defined by distinguishing two 
special control states, $q_{\kw{acc}}$ and $q_{\kw{rej}}$, that once
reached will `freeze' $M$ in its current configuration. 
We say that $M$ accepts (the empty input) if for all paths 
in the computation tree, the state $q_{\kw{acc}}$ is 
eventually reached; otherwise, we say that $M$ rejects.

Differently from the proofs of Theorem \ref{th:psb-instance-hard}
and Theorem \ref{thm:psbidtwoexphardcomb}, the configurations of
$M$ can be described by simply specifying the label of each cell
of the tape, the position of the head, and the control state of 
the Turing machine $M$.
We thus define \emph{cell values} as elements of  
$V = (\Sigma \times Q) \uplus \Sigma$,
where $\Sigma$ is the alphabet of $M$ and $Q$ is the set of its control states.
If a cell has value $(a,q)$, this means that the associated letter is $a$, the 
control state of $M$ is $q$, and the head is on this cell. 
Otherwise, if a cell has value $a$, this means that the associated letter is 
$a$ and the head of $M$ is not on this cell.

Now, let $n$ be the size of the tape of $M$.
We begin by describing the initial configuration of $M$.
This is encoded by a visible relation $C_0$ of arity $n+1$, where the first 
attribute gives the identifier of the initial configuration and the remaining
$n$ attributes give the values of the tape cells. As the relation $C_0$
is visible, we can immediately fix its content to be a singleton consisting
of the tuple $(x_0,y_1,y_2,y_3,\ldots,y_n)$, where
$x_0$ is the identifier of the initial configuration, 
$y_1 = \bot$, $y_2 = (\bot,q_0)$, $y_3 = \ldots = y_n = \bot$.
As for the other configurations of $M$, we store them into two distinct hidden
relations $C^\exists$ and $C^\forall$, depending on whether the control states
are existential or universal. Each fact in one of these two relation consists 
of $n+1$ attributes, where the first attribute specifies an identifier and the 
remaining $n$ attributes specify the cell values.
We can immediately give the first sentence, which requires the initial 
configuration to be existential and stored also in the relation $C^\exists$:
$$
  C_0(x,y_1,\ldots,y_n) ~\then~ C^\exists(x,y_1,\ldots,y_n) \ .
$$

To represent the computation tree of $M$, we encode pairs of subsequent 
configurations. In doing so, we not only store the identifiers of the configurations,
but also their contents, in such a way that we can later check the correctness of the transitions using inclusion dependencies. We use different relations to record %michael: not recall!
 whether the current 
configuration is existential or universal and, in the latter case, 
whether the successor configuration is the first or the second one in the transition
set (recall that the transition rules of $M$ define exactly two successor configurations 
from each universal configuration). Formally, we introduce three hidden relations
$S^\exists$, $S^\forall_1$, and $S^\forall_2$, all of arity $2n+2$.
We can easily enforce that the first $n+1$ and the last $n+1$ attributes 
in every tuple of $S^\exists$, $S^\forall_1$, and $S^\forall_2$ describe
configurations in $C^\exists$ and $C^\forall$:
$$
\begin{array}{rclrcl}
  S^\exists(x,\bar y,x',\bar y') &\then& C^\exists(x,\bar y) &\qquad\qquad
  S^\exists(x,\bar y,x',\bar y') &\then& C^\exists(x',\bar y') \\[1ex]
  S^\forall_1(x,\bar y,x',\bar y') &\then& C^\forall(x,\bar y) &\qquad\qquad
  S^\forall_1(x,\bar y,x',\bar y') &\then& C^\forall(x',\bar y') \\[1ex]
  S^\forall_2(x,\bar y,x',\bar y') &\then& C^\forall(x,\bar y) &\qquad\qquad
  S^\forall_2(x,\bar y,x',\bar y') &\then& C^\forall(x',\bar y') \ .
\end{array}
$$
Similarly, we guarantee that every existential (resp., universal) configuration
has one (resp., two) successor configuration(s) in $S^\exists$ (resp., $S^\forall_1$
and $S^\forall_2$):
$$
\begin{array}{rcl}
  C^\exists(x,\bar y) &\then& \exists ~x'~\bar y'~ S^\exists(x,\bar y,x',\bar y')   \\[1ex]
  C^\forall(x,\bar y) &\then& \exists ~x'~\bar y'~ S^\forall_1(x,\bar y,x',\bar y') \\[1ex]
  C^\forall(x,\bar y) &\then& \exists ~x'~\bar y'~ S^\forall_2(x,\bar y,x',\bar y') \ .
\end{array}
$$

We now turn to explaining how we can enforce the correctness of the transitions
represented in the relations $S^\exists$, $S^\forall_1$, and $S^\forall_2$.
Compared to the proof of Theorem \ref{th:psb-instance-hard}, the goal is simpler
in this setting, as we can simply compare the values $z_{-1},z_0,z_{+1}$ for the 
cells at positions $i-1,i,i+1$ in a configuration with the value $z'$ for the cell 
at position $i$ in the successor configuration. We thus introduce new visible relations
$N^\exists$, $N^\forall_1$, and $N^\forall_2$ of arity $4$. Each of these
relations is initialized with the possible quadruples of cell values 
$z_{-1},z_0,z_{+1},z'$ that are allowed by the transition function of $M$.
Consider, for example, the case where the transition function specifies that, when $M$ is in
the universal control state $q$ and reads the letter $a$, then the first of the
two subcomputations spawned by $M$
  begins by rewriting $a$ with $a'$, moving 
the head to the left, and switching to control state $q'$. In this case we add
to $N^\forall_1$ all the tuples of the form $\big( a_{-1}, (a,q), a_{+1}, a' \big)$
or $\big( a_{-2}, a_{-1}, (a,q), (a_{-1},q') \big)$, with $a_{-2},a_{-1},a_{+1}\in\Sigma$.
Accordingly, we introduce the following IDs, for all $1<i<n$:
$$
\begin{array}{rcl}
  S^\exists(x,\bar y,x',\bar y')   &\then& N^\exists(y_{i-1},y_i,y_{i+1},y'_i) \\[1ex]
  S^\forall_1(x,\bar y,x',\bar y') &\then& N^\forall_1(y_{i-1},y_i,y_{i+1},y'_i) \\[1ex]
  S^\forall_2(x,\bar y,x',\bar y') &\then& N^\forall_2(y_{i-1},y_i,y_{i+1},y'_i) \ .
\end{array}
$$
Furthermore, we constrain the values of the extremal cells to never change:
$$
\begin{array}{rclrcl}
  S^\exists(x,\bar y,x',\bar y')   &\then& E(y_1,y'_1) &\qquad\quad
  S^\exists(x,\bar y,x',\bar y')   &\then& E(y_n,y'_n) \\[1ex]
  S^\forall_1(x,\bar y,x',\bar y') &\then& E(y_1,y'_1) &\qquad\quad
  S^\forall_1(x,\bar y,x',\bar y') &\then& E(y_n,y'_n) \\[1ex]
  S^\forall_2(x,\bar y,x',\bar y') &\then& E(y_1,y'_1) &\qquad\quad
  S^\forall_2(x,\bar y,x',\bar y') &\then& E(y_n,y'_n)
\end{array}
$$
where $E$ is another visible binary relation interpreted by the
singleton instance $\{(\bot,\bot)\}$.

It remains to specify the query that checks that the Turing machine $M$ 
reaches the rejecting state $q_{\kw{rej}}$ along some path of its computation tree.
For this, we introduce a last visible relation $V_{\kw{rej}}$ that contains all cell 
values of the form $(a,q_{\kw{rej}})$, with $a\in\Sigma$. 
The query that checks this property is
$$
  Q ~=~ \bigvee_{1<i<n} \exists ~ x ~ \bar y ~
        \big(~ C^\exists(x,\bar y) ~\et~ V_{\kw{rej}}(y_i) ~\big) \ .
$$

\medskip
Let $\visinst$ be the instance that captures the intended semantics of the visible 
relations $V$, $C_0$, $N^\exists$, $N^\forall_1$, $N^\forall_2$, $E$, and $V_{\kw{rej}}$,
The proof that $\NSB(Q,\calC,\bfS,\visinst)=\true$ iff $M$ accepts (namely, has a 
computation tree where all paths visit the control state $q_{\kw{acc}}$) goes
along the same lines of the proof of Theorem \ref{th:psb-instance-hard}.
\end{proof}

%% file: almostadomcontrollable.tex
\begin{theorem} \label{thm:almostadomcontrollable}
For any collection $\calC$ of Linear TGDs, the problem
$\NSB(Q,\calC,\bfS,\visinst)$ is active domain controllable modulo $k$, where
$k$ is the maximal arity of any relation in the schema.
\end{theorem}

\begin{proof}
Let $k$ be the maximal arity of any relation in the schema.
The main idea is to compress an arbitrary counterexample instance to $\NSB$ by
one with at most $k$ elements outside the active domain, by taking $k$ ``representative elements'' outside the active domain and
replacing arbitrary tuples outside the active domain with these $k$ elements.
In doing this replacement, we should take into account equalities within each tuple.

Formally, we say that two tuples $\vec t$ and $\vec t'$ of the same length are \emph{equality equivalent} if:
$t_i=t_j$ if and only if $t'_i=t'_j$ and for every schema constant
$c$, $t_i=c$ if and only if $t'_i=c$.

Suppose that $\NSB(Q,\calC,\bfS,\visinst)=\false$, namely, that 
there is an $\bfS$-instance $\fullinst$ such that $\fullinst\sat\calC$,
$\fullinst\sat Q$, and $\visible(\fullinst) = \visinst$. We need to give a  instance
$\fullinst'$ whose active domain has only $k$ elements outside the active domain of $\visinst$
that witnesses $\NSB(Q,\calC,\bfS,\visinst)=\false$.

We fix an extension $\bbD$ of the active domain of $\visinst$
that contains $k$ additional fresh values.
For each fact $R(\bar a)$  in $\fullinst$
and each tuple $\bar b\in\bbD^{\arity(R)}$, if $\bar b$ and $\bar a$ are equality-equivalent
and agree on each position whose value is in the active domain of $\visinst$,
we add the fact $R(\bar b)$ to $\fullinst'$.
By definition, the instance $\fullinst'$ 
agrees with $\fullinst$ on the visible part, and has only $k$ elements outside the active domain
of $\visinst$.

Below we show that $\fullinst'$ satisfies the sentences of $\calC$ and the query $Q$. 
Consider any linear TGD $\tau$ of $\calC$ of the form 
$$
  R(\bar x) ~\then~ \exists \bar y ~ S(\bar z)
$$
and any fact $R(\bar a)$ that is the image under some homomorphism $h$ 
of the left-hand side atom $R(\bar x)$. 
Let $I$ be the set of positions $i\in\{1,\ldots,\arity(R)\}$ such that 
$\bar a(i)\in\adomain(\visinst)$. We know that there is $\bar u$ such that $R(\bar u)$ 
holds in $\fullinst$ such that $\bar a|I = \bar u|I$ and $\bar a$ is equality-equivalent
to $\bar u$. Since $\fullinst$ satisfies $\tau$, and $\bar u$ is equality-equivalent
to $\bar a$, we know that there is a fact $S(\bar v)$ in $\fullinst$
agreeing with $\bar u$ on the positions corresponding to exported variables of $\tau$. Let $\bar b$ be any tuple in
$\bbD^{\arity(R)}$ equality-equivalent
to $\bar b$ and  agreeing with $\bar v$ on all the positions corresponding to exported variables of $\tau$.
Since $k$ is at least the arity of $R$, such a $\bar b$ must exist.
Then $\bar b$ witnesses that $\tau$ holds for $\bar a$. This completes the proof that the sentences of $\calC$ hold.

A similar argument shows that $Q$ holds in $\fullinst'$.  Thus $\fullinst'$ witnesses that
$\NSB(Q,\calC,\bfS,\visinst)$ is active domain controllable modulo $k$.
\end{proof}

\begin{example} \label{ex:kadomproof}
As an example of the prior argument, consider a TGD $\tau$
\[
R(x,y,y) \rightarrow \exists z ~ S(y,z, z)
\] 
and suppose the instance $\fullinst$ has a tuple
$R(a,b,b)$ where $a_1$ is in the active domain of the visible instance and $b$ is outside of the
active domain of the visible instance. 
Thus  there is a  homomorphism from the left side
of $\tau$ to $R(a,b,b)$.
Since $\fullinst$ satisfies $\tau$, it must contain $S(b,c,c)$ for some  value $c$.

The instance $\fullinst'$ produced by the prior argument will replace
$R(a,b,b)$ by
$R(a,c_1, c_1)$, where $c_1$ is one of the $k$ additional constants.
We explain why this replacement will not break the satisfaction
of $\tau$. There is a homomorphism $h'$ of the left hand side of $\tau$
to $R(a, c_1, c_1)$.
If the witness $c$ of $S(b,c,c)$ is in the active domain of the visible instance, then
 $\fullinst'$ has $S(c_1, c, c)$, and thus we have the witness we need for 
$\tau$ with respect to $h'$.
If $c$ is not in the active domain of the visible instance, then
$\fullinst'$ will also have $S(c_1, c_2, c_2)$, for $c_2$ another of the additional constants.
Either way the required value is present.
\end{example}

%% file: negexists.tex
\subsection{Existence problems} \label{subsec:negexists}

Here we consider the complexity of the schema-level question,
$\exists \NSB(Q, \calC, \bfS)$.
We first show that when the background theories are preserved under disjoint unions 
(e.g., connected frontier guarded TGDs), the existence of an $\NSB$ can be
checked by considering a single ``negative critical instance'',
namely the empty visible instance $\emptyset$.
This instance is easily seen to be realizable: the variant of the chase procedure that
we introduced in Section \ref{sec:existsPSB} terminates immediately when initialized
with the empty instance $\fullinst_0=\emptyset$ and returns the singleton collection 
$\Chase(\calC,\bfS,\emptyset)$ consisting of the empty $\bfS$-instance satisfying
$\calC$.

\begin{theorem}\label{thm:nsbexistscollapse}
If the query $Q$ is monotone and the background theory $\calC$ is preserved under 
disjoint unions of instances,
then $\exists\NSB(Q, \calC, \bfS)=\true$ iff 
$\NSB(Q,\calC, \bfS, \emptyset)=\true$.
\end{theorem}

\begin{proof}
It is immediate to see that $\NSB(Q,\calC, \bfS, \emptyset)=\true$ implies 
$\exists\NSB(Q, \calC, \bfS)=\true$.
We prove the converse implication by contraposition.

Suppose that $\NSB(Q,\calC,\bfS,\emptyset)=\false$, namely, that there is an
$\bfS$-instance $\fullinst$ satisfying $\calC$ and $Q$ and such that 
$\visible(\fullinst)=\emptyset$. 
We aim at proving that $\NSB(Q,\calC, \bfS, \visinst)=\false$ for all realizable 
visible instances $\visinst$.
Let $\visinst$ be such a realizable instance and let $\fullinst'$ be an $\bfS$-instance 
that satisfies $\calC$ and such that $\visible(\fullinst')=\visinst$. 
We define the new instance $\fullinst''$ as a disjoint union of $\fullinst$ 
and $\fullinst'$. Since the background theory $\calC$ is preserved under disjoint unions, 
$\fullinst''$ satisfies $\calC$. Moreover, $\fullinst''$ satisfies the query $Q$, 
by monotonicity. Since $\visinst=\visible(\fullinst')=\visible(\fullinst'')$, 
we have $\NSB(Q,\calC,\bfS,\visinst)=\false$.
Finally, since $\visinst$ was chosen in an arbitrary way, this proves that 
$\exists\NSB(Q, \calC, \bfS)=\false$.
\end{proof}

Using the ``negative critical instance'' result above and Theorem \ref{thm:nsbgnf}, 
we immediately see that $\exists \NSB(Q,\calC, \bfS)$ is decidable in $\twoexp$ 
for $\gnfo$ sentences that are closed under disjoint unions, and in particular 
for connected frontier-guarded TGDs.
Combining with Corollary \ref{cor:complnsbinstlin} also gives an $\exptime$ bound 
for linear TGDs.
In fact, we can improve this upper bound by observing that the $\NSB$ problem over the
empty visible instance reduces to classical Open-World Query answering:

\begin{proposition}\label{prop:nsbemptyandowq} 
For any  Boolean CQ $Q$, $\NSB(Q,\calC,\bfS,\emptyset)$  holds
% here we need to restrict to CQ to talk of the frozen body and have a polynomial-time reduction...
iff $\OWQ(Q', \calC, \canondb(Q))$  holds, where
$$
  Q' ~=~ \bigvee\nolimits_{\!\!\!R\in \bfS_v} \exists\bar{x} ~ R(\bar{x})
$$
and $\canondb(Q)$ is the canonical instance of the CQ $Q$.
% ...but we need UCQ in the reduced NSB problem, even if we start with CQ!
\end{proposition}

\begin{proof}
Suppose that $\NSB(Q,\calC,\bfS,\emptyset) = \true$.
This means that every $\bfS$-instance that satisfies the 
sentences in $\calC$ and has empty visible part, must violate
the query $Q$. By contraposition, every $\bfS$-instance that 
satisfies the sentences $\calC$ and contains $\canondb(Q)$
(i.e., satisfies $Q$), must contain some visible facts, and hence 
satisfy the UCQ $Q'$. This implies that $\OWQ(Q',\calC,\canondb(Q)) = \true$.

The proof that $\OWQ(Q',\calC,\canondb(Q)) = \true$ implies
$\exists\NSB(Q,\calC,\bfS,\emptyset) = \true$ follows symmetric arguments.
\end{proof}

We know from previous results \cite{bgo} that $\OWQ$ for Boolean UCQs and linear TGDs 
is in $\pspace$. From the above reduction, we immediately get that the problem 
$\NSB(Q,\calC, \bfS, \emptyset)$, and hence (by Theorem \ref{thm:nsbexistscollapse}) 
the problem $\exists\NSB(Q,\calC,\bfS)$, for a set of linear TGDs is also in $\pspace$.

\begin{corollary}\label{cor:nsbemptyandowq}
The problem $\exists\NSB(Q,\calC,\bfS)$, as $Q$ ranges over  
Boolean UCQ and $\calC$ over sets of linear TGDs, is in $\pspace$.
\end{corollary}

Matching lower bounds for $\exists\NSB$ come by a converse reduction 
from Open-World Query answering. 

%\gabriele{Here is the Chase characterization for NSB. 
%          It might be moved somewhere else, as you prefer,
%          but keep in mind that this is just for the empty instance,
%          so it might be better to introduce it in this section}%
%michael: I agree
To prove this reduction, we first provide a characterization 
of the $\NSB$ problem over the empty visible instance, which is based,
like Proposition \ref{prop:PSB-characterizion}, on our chase
procedure:

\begin{proposition}\label{prop:NSB-characterization}
If $Q$ is a Boolean CQ and $\calC$ is a set of TGDs and EGDs without constants over a 
schema $\bfS$, then $\NSB(Q,\calC,\bfS,\emptyset)=\true$ 
iff either $Q$ contains a visible atom, or it does not and in this case
$\Chase(\calC,\bfS,\canondb(Q))=\emptyset$.
\end{proposition}

\begin{proof}
Suppose that $Q$ does not contain visible atoms and 
$\Chase(\calC,\bfS,\canondb(Q))$ is non-empty.
Let $K$ be some instance in $\Chase(\calC,\bfS,\canondb(Q))$
and observe that, by construction, $K$ satisfies the sentences
in $\calC$ and the query $Q$, and has the same visible part
as $\canondb(Q)$, which is empty.
This means that $K$ is a witness of the fact that 
$\NSB(Q,\calC,\bfS,\emptyset)=\false$.

Conversely, suppose that $\NSB(Q,\calC,\bfS,\emptyset)=\false$.
This means that there is an $\bfS$-instance $\fullinst$ with
no visible facts that satisfies the sentences in $\calC$ and 
the query $Q$.
Since $\fullinst\sat Q$, there is a homomorphism $g$ from $\canondb(Q)$
to $\fullinst$. Moreover, since $Q$ contains no visible atoms,
the two instances $\fullinst$ and $\canondb(Q)$ agree on the visible
part.
%\gabriele{I don't understand this ``Note that...'', neither the argument nor its use}
%Note that $g$ acts only on hidden facts and hence,
%without loss of generality, we can assume that $g$ is identity.
By Lemma \ref{lem:disjunctive-chase-universal}, letting 
$\fullinst_0 = \canondb(Q)$, we get the existence of an 
instance $K$ in $\Chase(\calC,\bfS,\canondb(Q))$.
\end{proof}

As in the positive case, the upper bounds are  tight:

\begin{theorem}\label{thm:eNSB-twoexp}
$\exists \NSB(Q, \calC, \bfS)$ is $\twoexp$-hard
as $Q$ ranges over Boolean CQs and $\calC$ over sets of 
connected $\fgtgd$s.
\end{theorem}

\begin{theorem}\label{thm:complexistsnsblin}
$\exists \NSB(Q, \calC, \bfS)$ is $\pspace$-hard 
as $Q$ ranges over Boolean CQs and $\calC$ over sets of linear $\tgd$s.
\end{theorem}

The first theorem will be proven by  reducing the open-world query answering problem to
$\exists \NSB$, and then applying a prior $\twoexp$-hardness result from 
Cal{\`{\i}} et al. \cite{taming}.
The $\pspace$ lower bound will be shown by a reduction from the implication problem for IDs,
shown  $\pspace$-hard by Casanova et al. \cite{casanova}. 

%To prove lower bounds for $\exists \NSB$, 
We begin with the reduction from Open-World Query answering:

\begin{proposition}\label{prop:ENSBtoOP} 
There is a polynomial time reduction from the Open-World Query answering problem
over a set of connected $\fgtgd$s without constants and a connected Boolean CQ to 
an $\exists \NSB$ problem over a set of connected $\fgtgd$s without constants and
a Boolean CQ.
\end{proposition}

\begin{proof}
Consider the Open-World Query answering problem over a schema $\bfS$,
a set $\calC$ of sentences without constants and closed under disjoint union,
a Boolean CQ $Q$, and an $\bfS$-instance $\fullinst$.
We reduce this problem to an $\exists\NSB$ problem over a new schema $\bfS'$,
a new set of sentences $\calC'$, and a new Boolean CQ $Q'$. 
The schema $\bfS'$ is obtained from $\bfS$ by adding a relation $\kw{Good}$ 
of arity $0$, which is assumed to be the only visible relation in $\bfS'$.
The set of sentences $\calC'$ is equal to $\calC$ unioned with the sentence
$$
  S_1(\bar x_1) \et \ldots \et S_m(\bar x_m) ~\then~ \kw{Good}
$$
where $S_1(\bar x_1)$, \dots, $S_m(\bar x_m)$ are the atoms in the CQ $Q$.
The query $Q'$ is defined as the \emph{canonical query} of the instance $\fullinst$, 
obtained by replacing each value $v$ with a variable $y_v$ and by quantifying 
existentially over all these variables. 
Note that $\canondb(Q')$ is isomorphic to the input instance $\fullinst$.

Now, assume that the original sentences in $\calC$ were connected $\fgtgd$s
and the CQ $Q$ was also connected. By construction, the sentences in $\calC'$ 
turn out to be also connected $\fgtgd$s. In particular, the satisfiability of 
these sentences are preserved under disjoint unions, and hence from Theorem 
\ref{thm:nsbexistscollapse}, 
$\exists\NSB(Q',\calC',\bfS') = \true$ iff $\NSB(Q',\calC',\bfS',\emptyset) = \true$.
Thus, it remains to show that $\NSB(Q',\calC',\bfS',\emptyset) = \true$ iff 
$\OWQ(Q,\calC, \fullinst) = \true$.

By contraposition, suppose that $\OWQ(Q,\calC,\fullinst) = \false$.
This means that there is a $\bfS$-instance $\fullinst'$ that contains $\fullinst$,
satisfies the sentences in $\calC$, and violates the query $Q$.
In particular, $\fullinst'$, seen as an instance of the new schema $\bfS'$, without
the visible fact $\kw{Good}$, satisfies the query $Q'$ and the sentences in $\calC'$
(including the sentence that derives $\kw{Good}$ from the satisfiability of $Q$). 
The $\bfS'$-instance $\fullinst'$ thus witnesses the fact that 
$\NSB(Q',\calC',\bfS',\emptyset) = \false$.

Conversely, suppose that $\NSB(Q',\calC',\bfS',\emptyset) = \false$.
Recall that the sentences in $\calC'$ do not use constants and $Q'$
contains no visible facts. We can thus apply Proposition \ref{prop:NSB-characterization}
and derive $\Chase(\calC',\bfS',\canondb(Q'))\neq\emptyset$. Note that
$\canondb(Q')$
is clearly isomorphic to the original instance $\fullinst$.
In particular, there is an instance $K$ in $\Chase(\calC',\bfS',\canondb(Q'))$ 
that contains the original instance $\fullinst$, satisfies the sentences in 
$\calC'$, and does not contain the visible fact $\kw{Good}$.
From the latter property, we derive that $K$ violates the query $Q$.
Thus $K$, seen as an instance of the schema $\bfS$, witnesses the fact
that $\OWQ(Q,\calC,\fullinst) = \false$.
\end{proof}

We note that  there are two variants of $\OWQ$, corresponding to finite and infinite instances.
However, by finite-controllability of  $\fgtgd$s, inherited from
the finite model property of $\gnfo$ (see Theorem \ref{thm:gnfsat})
 these two variants agree. Hence we do not distinguish them. Similar remarks
hold for other uses of $\OWQ$ within proofs in the paper.

\medskip
We are now ready to prove Theorem~\ref{thm:eNSB-twoexp}, namely, the 
$\twoexp$-hardness of the problem $\exists\NSB(Q,\calC,\bfS)$, 
where $Q$ ranges over Boolean CQs and $\calC$ ranges over sets of connected FGTGDs. 
%We recall the statement:
%
%\begin{statement}
%The problem $\exists \NSB(Q, \calC, \bfS)$ is $\twoexp$-hard
%as $Q$ ranges over Boolean CQs and $\calC$ ranges over sets of 
%connected $\fgtgd$s.
%\end{statement}

\begin{proof}[Proof of Theorem~\ref{thm:eNSB-twoexp}]
Theorem 6.2 of Cal{\`{\i}} et al.~\cite{taming} shows $\twoexp$-hardness of
open-world query answering for $\fgtgd$s. An inspection of the proof shows
that only connected $\fgtgd$s are required. Thus, the theorem follows immediately 
from Proposition~\ref{prop:ENSBtoOP}.
\end{proof}

\medskip
We now turn towards proving Theorem~\ref{thm:complexistsnsblin}, namely,
the $\pspace$ lower bound for $\exists \NSB$ under linear $\tgd$s. 
Recall that the reduction in Proposition \ref{prop:ENSBtoOP} 
does not preserve smaller classes of sentences, such as linear $\tgd$s.
We thus prove the theorem using a separate reduction. 
%Recall the statement of Theorem \ref{thm:complexistsnsblin}:
%
%\begin{statement}
%The problem $\exists \NSB(Q, \calC, \bfS)$ is $\pspace$-hard as 
%$Q$ ranges over Boolean CQs and $\calC$ ranges over sets of linear $\tgd$s.
%\end{statement}

\begin{proof}[Proof of Theorem~\ref{thm:complexistsnsblin}]
We reduce from the implication problem for inclusion dependencies (IDs), 
which is known to be $\pspace$-hard from Casanova et al. \cite{casanova}.
Consider a set of IDs $\calC$ and an additional ID 
$\delta = S_\star(\bar x_\star) ~\rightarrow~ \exists \bar y ~ T_\star(\bar z_\star)$, 
where $\bar x_\star,\bar y$ are sequences of pairwise distinct variables and 
$\bar z_\star$ is a sequence of variables from $\bar x_\star$ and $\bar y$. 
We denote by $F(\delta)$ the sequence of variables shared between $\bar x_\star$ and $\bar z_\star$ and $m$ the length of this vector.
Note that we annotated relations and variables in $\delta$
with the subscript $_\star$ 
in order to make it clear when refer later to these particular objects.

We create a new schema $\bfS'$ that contains, for each relation $R$ of arity $k$
in the original schema $\bfS$, a relation $R'$ of arity $k+m$. 
We also add to $\bfS'$ a copy of each relation $R$ from $\bfS$, without changing
the arity. Furthermore, we add a $0$-ary relation $\good$, which is the only 
visible relation of $\bfS'$.
For each ID in $\calC$ of the form
$$
  R(\bar x) ~\then~ \exists \bar y ~ S(\bar z)
$$
we introduce a corresponding ID in $\calC'$ of the form 
%\gabriele{This definition may violate the previous assumption that the arity of $S'$ is double of the arity of $S$...
 %         What are we doing here??
%          Maybe the arity increases by just $|x_\star|$}
%\pierre{I changed the proof}
$$ 
  R'(\bar x, \bar x') ~\then~ \exists \bar y ~ S'(\bar z, \bar x')
$$
where the variables in $\bar x'$ are distinct from the variables in $\bar x$. 
We also add the sentences
$$
\begin{array}{rcl}
  S_\star(\bar x_\star)                       &\rightarrow& S'_\star(\bar x_\star, F(\delta)) \\[1ex]
  T'_\star(\bar z_\star, F(\delta)) &\rightarrow& \good
\end{array}
$$
where the elements of $\bar z_\star$ are arranged as in the atom $T_\star(\bar z_\star)$ 
that appears on the right-hand side of the ID $\delta$.
Note that the sentence that copies the content from $R$ to $R'$ 
and duplicates the attributes is not an ID, but is still a linear TGD.
The query of our $\exists\NSB$ problem is defined as 
$$
  Q' ~=~ \exists \bar x ~ S_\star(\bar x) \ .
$$

The sentences that we just defined are preserved under disjoint unions. 
Thus, by Theorem \ref{thm:nsbexistscollapse}, we know that
$\exists\NSB(Q',\calC',\bfS') = \true$ iff $\NSB(Q',\calC',\bfS',\emptyset) = \true$.
Below, we prove that the latter holds iff the ID $\delta$ is implied by the set of
IDs in $\calC$.

In one direction, suppose that the implication holds.
From this, we can easily infer that in the schema $\bfS'$ the following dependency holds:
$$
  S'_\star(\bar x_\star, F(\delta)) ~\then~ \exists \bar y ~ T'_\star(\bar z_\star, F(\delta))
$$
Consider now a full $\bfS'$-instance $\fullinst'$ with empty visible part.
We show that the query $Q'$ is not satisfied, namely, $\fullinst'$ cannot 
satisfy  $\exists \bar x_\star ~ S_\star(\bar x_\star)$.
If it did, then, by the copy of the sentences on the primed relations, this would 
yield  $\exists \bar x_\star ~  S'_\star(\bar c, F(\delta))$. Hence, by the sentences
in the background theory,
we infer  that $\exists \bar z_\star ~ T'_\star(\bar z_\star, F(\delta))$ holds, and thus that
$\good$ holds.
This however would contradict the hypothesis that $\fullinst'$ has empty visible part.

In the other direction, suppose that the implication fails and consider a 
witness $\bfS$-instance $\fullinst$ that contains the fact $S_\star(\bar x_\star)$ 
but not the corresponding $T_\star$ fact.
We create a full $\bfS'$-instance $\fullinst'$ with empty visible part 
where $Q'$ holds, thus showing that $\exists\NSB(Q',\calC',\bfS',\emptyset)=\false$.
We first copy in $\fullinst'$ the content of all relations $R$ from $\fullinst$.
In particular, $\fullinst'$ contains the fact $S_\star(\bar x_\star)$, but no 
$T_\star$ fact.
The primed relations $R'$ in $\fullinst'$ are set to contain all and only the facts 
of the form $R'(\bar x, F(\delta))$, where $R(\bar x)$ is a fact in $\fullinst$.
Finally, we set $\good$ to be the empty relation in $\fullinst'$.
Clearly, $Q'$ holds in $\fullinst'$ and the visible part is the empty instance.
It is also easy to verify that all the sentences in $\calC'$ are satisfied by
$\fullinst'$, and this completes the proof.
\end{proof}

Note that the reduction above does not create a schema with  IDs, but rather
with general linear TGDs (variables can be repeated on the right).
We do not know whether $\exists \NSB(Q, \calC, \bfS)$
is  $\pspace$-hard even for background theories consisting of IDs.

\smallskip
We can show that the connectedness requirement is critical for decidability:

\begin{theorem}\label{th:eNSB-und}
%\michael{Should check for a connection with the results of Nash, Segoufin Vianu 
%         on cases where determinacy of UCQ views and queries is undecidable -- 
%         the examples (see Proposition 4.2 on page 19 of their journal paper)
%         also revolve around non-connected queries).}%
The problem $\exists\NSB(Q,\calC,\bfS)$ is undecidable 
as $Q$ ranges over Boolean CQs and $\calC$ over sets of $\fgtgd$s. 
\end{theorem}

\begin{proof}
We give a reduction from the \emph{model conservativity problem} for 
$\el$ TBoxes, which is shown undecidable in \cite{conservextdl}.  
Intuitively, $\el$ is a logic that defines $\fgtgd$s over relations 
of arity $2$, called ``TBoxes''.
Given some TBoxes $\phi_1$ and $\phi_2$ over two schemas $\bfS_1$ and $\bfS_2$, 
respectively, with $\bfS_1\subseteq\bfS_2$, we say that $\phi_2$ is a 
\emph{model conservative extension} of $\phi_1$ if every $\bfS_1$-instance 
$\visinst$ that satisfies $\phi_1$ can be extended to an $\bfS_2$-instance 
that satisfies $\phi_2$ without changing the interpretation of the predicates 
in $\bfS_1$, that is, by only adding an interpretation for the relations that 
are in $\bfS_2$ but not in $\bfS_1$.
The model conservativity problem consists of deciding whether $\phi_2$ 
is a model conservative extension of $\phi_1$. The proof in \cite{conservextdl} shows
that this problem is undecidable for both finite instances and arbitrary instances.

We reduce the above problem to the complement of $\exists\NSB(Q,\calC,\bfS)$, 
for suitable $Q$, $\calC$, and $\bfS$, as follows.
Given some TBoxes $\phi_1$ and $\phi_2$ over the schemas $\bfS_1\subseteq\bfS_2$,
let $\bfS$ be the schema obtained from $\bfS_2$ by adding a new predicate $\good$ 
of arity $0$ and by letting the visible part be $\bfS_1$ (in particular,
the relation $\good$ is hidden).
Further let $\calC=\{\phi_1, \good \then \phi_2\}$, where $\good\then \phi_2$ 
is shorthand for the collection of $\fgtgd$s obtained by adding $\good$ as a 
conjunct to the left-hand side of each dependency of $\phi_2$
(note that this makes the dependency unconnected).
Finally, consider the query $Q=\good$. We have that $\exists\NSB(Q,\calC, \bfS)=\true$ 
iff there is an $\bfS_1$-instance $\visinst$ satisfying $\phi_1$, none of whose 
$\bfS_2$-expansions satisfies $\phi_2$.
\end{proof}

%% file: negsummary.tex
\subsection{Summary for Negative Query Implication}

A summary of results on negative implication is below. We notice
that the decidable cases are orthogonal to those for positive implications. 
Note also that unlike in the positive cases, we have tractable cases for data complexity.

\begin{center}
\scriptsize
\begin{tabular}{@{}l|lll@{}}
 background theory $\Sigma$                     & $\NSB$ data complexity     &  $\NSB$ combined complexity & $\exists\NSB$ \\[1ex]
  \hline\\
%  ID                 & $\ptime$   & $\exptime$ & $\pspace$\\
%                     & Cor \ref{cor:complnsbinstlin}/Prop \ref{prop:graph-reachability}   
%                     & Cor \ref{cor:complnsbinstlin}/Thm \ref{thm:nsbinstcombexphard}  
%                     & Cor \ref{cor:complexistsnsblin}\\[0.5ex]
%  \hline\hlineskip
  Linear             &  In $\ptime$     & $\exptime$-cmp   & $\pspace$-cmp\\
  ~~ TGD             & Cor.~\ref{cor:complnsbinstlin}~%/~Prop.~\ref{prop:graph-reachability}
                     & Cor.~\ref{cor:complnsbinstlin}~/~Thm.~\ref{thm:nsbinstcombexphard}
                     & Cor.~\ref{cor:nsbemptyandowq}~/~Thm.~\ref{thm:complexistsnsblin}\\[1ex]
  \hline\\
  Conn.~Disj.        & $\exptime$-cmp & $\twoexp$-cmp  & $\twoexp$-cmp\\
  ~~ $\fgtgd$        & Thm.~\ref{thm:nsbgnf}~/~Thm.~\ref{th:PSBtoNSB}
                     & Thm.~\ref{thm:nsbgnf}~/~Thm.~\ref{th:PSBtoNSB}
                     & Thm.~\ref{thm:nsbexistscollapse}/Thm.~\ref{thm:eNSB-twoexp}\\[1ex]
  \hline\\
  $\fgtgd$           & $\exptime$-cmp & $\twoexp$-cmp    & undecidable\\
  \& GNFO            & Thm.~\ref{thm:nsbgnf}~/~Thm.~\ref{th:PSBtoNSB}
                     & Thm.~\ref{thm:nsbgnf}~/~Thm.~\ref{th:PSBtoNSB}
                     & Thm.~\ref{th:eNSB-und}\\[1ex]
  % \hline\hlineskip
  % GNFO               & $\exptime$ & $\twoexp$  & undecidable
  %                    & Thm \ref{thm:nsbgnf}/Thm \ref{th:PSBtoNSB}
  %                    & Thm \ref{thm:nsbgnf}/Thm \ref{th:PSBtoNSB}
  %                    & Thm \ref{th:eNSB-und}\\[0.5ex]
  \hline
\end{tabular}
\end{center}

%% file: combo.tex
\section{Extensions and special cases} \label{sec:extspec}
We present some results concerning natural extensions of the framework.

\smallskip
\myparagraph{Non-Boolean queries}
Throughout this work we have restricted to queries to be given as sentence.
The natural extension of the notion of query implication for non-Boolean queries is
to consider inference of information concerning membership of any visible
tuple in the query output. E.g.~$\PSB(Q, \calC, \bfS, \visinst)$ would hold if 
there is a tuple $\bar t$ over the active domain of $\visinst$ such that
$\bar t \in Q(\fullinst)$ for all instances $\fullinst$ of $\bfS$ satisfying 
the background theory $\calC$ and having visible part $\visinst$. 
As usual, the schema-level problem $\exists\PSB(Q, \calC, \bfS)$ 
(resp.~$\exists\NSB(Q, \calC, \bfS)$) for a non-Boolean query $Q$ 
amounts at deciding whether there is a realizable visible instance $\visinst$ 
witnessing $\PSB(Q, \calC, \bfS, \visinst)$ (resp.~$\NSB(Q, \calC, \bfS, \visinst)$).

We show that \emph{all of our results carry over to the non-Boolean case}.
Since the lower-bounds for Boolean problems are clearly inherited by the 
non-Boolean ones, we focus on arguing that the upper bounds carry over.

All the complexity upper bounds for the instance-level problem carry over 
in a rather simple way.
%using the simple approach of substituting in each potential output a tuple from $\visinst$
%and utilizing the prior algorithms on the resulting Boolean queries.
%The complexity for each substitution preserves the upper bounds since they hold in the presence of constants,
%and the iteration over tuples can be absorbed in the complexity classes given
%in our upper bounds: for data complexity the iteration is polynomial,
%while for combined complexity the number of tuples can be exponential, but our bounds are at least exponential.
%Moreover, the iteration over the tuples $\bar t$ can be absorbed in the 
%complexity classes given in our upper bounds: for data complexity the iteration is polynomial,
%while for combined complexity the number of tuples can be exponential, but our bounds are at least exponential.
For example, given $\bfS$, $\calC$, $\visinst$ as usual, and given a non-Boolean query $Q$
and a visible tuple $\bar t$, the problem of deciding whether $\bar t$ appears in every potential 
output $Q(\fullinst)$, for any instance $\fullinst$ satisfying $\calC$ and having visible 
part $\visinst$, reduces to the problem $\PSB(Q_{\bar t}, \calC, \bfS, \visinst)$, 
where $Q_{\bar t}$ is the Boolean query obtained by substituting the $i$-th free
variable of $Q$ with the $i$-th constant in $\bar t$, for all $i$'s.
A similar reduction holds for negative implication.
Thus the instance-level problem in the non-Boolean case reduces to 
a series of instance-level problems in the Boolean case, one for 
each choice of a tuple $\bar t$ over the active domain of $\visinst$.
Our upper bounds can be applied to the latter problems, since they hold 
in the presence of constants in the query. 
Moreover, the iteration over the tuples $\bar t$ can be absorbed in the 
complexity classes of our upper bounds: for data complexity the iteration 
is polynomial, while for combined complexity the number of tuples can be exponential, 
but our bounds are at least exponential.
Further, GFP-Datalog definability for negative implications also extends straightforwardly
to the non-Boolean case:
Theorem \ref{thm:almostadomcontrollable}
%thm:enforce-adom-controllability} 
extends with the same statement and proof,
while the argument in Theorem \ref{th:definability-active-domain-controllable}  is easily extended to  show
that there is a GFP-Datalog program
that returns the complement of $\NSB(Q, \calC, \bfS)$ within the active domain.

The complexity results for $\exists \PSB$ also generalize to the non-Boolean case:
we can revise Theorem \ref{thm:psbexistcollapselin} to state $\exists \PSB(Q,\calC,\bfS)=\true$
iff there is a positive query implication for the tuple $(a, \ldots, a)$ and the instance $\visinst_{\{a\}}$.
For $\exists \NSB$, we can extend Theorem \ref{thm:nsbexistscollapse} to show that
for logical sentences preserved under disjoint union, if there is a positive query implication 
involving some visible instance $\visinst$ and a tuple $\bar t$, then there is one involving 
the empty instance and 
%some tuple $\bar t$
the same tuple $\bar t$. 
From this it follows that the complexity
bounds for $\exists \NSB$ carry over to the non-Boolean case.

\smallskip
\myparagraph{Beyond unions of conjunctive queries}
So far we have considered only the case where the query $Q$ does not contain negation or universal quantification.
It is natural to extend the query language  even further, to Boolean combinations 
of Boolean conjunctive queries ({BCCQs}).
We note that the problem $\PSB(Q, \calC, \bfS, \visinst)$, as $Q$ ranges over BCCQs,
subsumes both $\PSB(Q, \calC, \bfS, \visinst)$ and $\NSB(Q, \calC, \bfS, \visinst)$  
for $Q$ a UCQ. Thus all lower bounds for either of these two problems are inherited 
by the BCCQ problem. The corresponding instance level problems are still decidable.
Indeed, this holds even when $Q$ is a $\gnfo$ sentence, since  we can  
%just apply using 
use the same translation to $\gnfo$ satisfiability applied in 
Theorems \ref{thm:gnfoposinstancedecid} and \ref{thm:nsbgnf}.
However, for the schema-level problems $\exists\PSB$ and $\exists\NSB$ 
we immediately run into problems:

\begin{theorem}\label{thm:booleancombund}
The problem $\exists\PSB(Q, \calC, \bfS)$ for a Boolean combination $Q$ 
of Boolean CQs is undecidable, even when the sentences in the background theory are 
IDs. The same holds for $\exists\NSB(Q, \calC, \bfS)$.
\end{theorem}

\begin{proof}
As in the previous undecidability results,
we reduce a tiling problem with tiles $T$, initial tile $t_\bot \in T$
and horizontal and vertical
constraints $H,V \subseteq T \times T$ 
to the 
problem  $\exists \PSB(Q,\calC,\bfS)$.  Again, for convenience we deal with the infinite
variant of the problem. The idea will be that the visible instance
witnessing $\exists \PSB$ represents the tiling, and invisible instances represent
challenges to the correctness of the tiling.

We model the infinite grid to be tiled by visible relations $E_H$ and $E_V$, and 
the tiling function by a collection of unary visible relations $U_t$, for all tiles
$t \in T$.

The invisible relations represent markings of the grid for possible errors.
There are several kinds of challenges.
We focus on the horizontal consistency challenge, which selects two nodes in 
the $E_H$ relation, to challenge whether the nodes satisfy the horizontal 
constraint.
Formally, the challenge is captured by a binary invisible predicate
$\horchallenge(x,y)$, with an associated sentence in the background theory
$$
  \horchallenge(x,y) ~\then~ E_H(x,y) \ .
$$
The query $Q$ will be satisfied only when the following 
\emph{negated} CQs hold, for all pairs $(t,t')\nin H$: 
$$
  \neg \exists ~ x ~ y ~~ \horchallenge(x,y) ~\et~ U_t(x) ~\et~ U_{t'}(y) \ .
$$
Note that this can only happen if the relation $\horchallenge$ has selected
two horizontally adjacent nodes whose tiles violate the horizontal constraints.
The vertical constraints are enforced in a similar way using an invisible 
relation $\vertchallenge$ and another negated CQ.

Recall that in the infinite grid, we have unique vertical and horizontal 
successors of each node, and the horizontal and vertical successor functions 
commute. Thus far we have not enforced that $E_V$ and $E_H$ have this property.
We will use additional hidden relations and IDs to enforce that every element is 
related to at least one other via $E_H$ and $E_V$.

We first show how to enforce that every element has at most one horizontal 
successor (``functionality challenge'').
We introduce a hidden relation $\hfuncchallenge(x,y,y')$ 
and a background theory sentence
$$
\begin{array}{rcl}
  \hfuncchallenge(x,y,y') &\then& E_H(x,y) \\[1ex]
  \hfuncchallenge(x,y,y') &\then& E_H(x,y') \ .
\end{array}
$$
We also add to the query $Q$ the conjunct:
$$
  \big( ~ \neg \exists~ x ~ y ~ y' ~ \hfuncchallenge(x,y,y') ~ \big)
  ~~\vee~~ 
  \big( ~ \exists ~ x ~ y ~ \hfuncchallenge(x,y,y) ~ \big) \ .
$$
We claim that if there is a visible instance witnessing $\exists \PSB$, 
then $E_H$ is functional. Indeed, if $E_H$ were not functional in the
visible instance, then we could choose a node $x$ with two distinct 
$E_H$-successors $y$ and $y'$, add only the tuple $(x,y,y')$ to 
$\hfuncchallenge$, and obtain a full instance that satisfies the 
sentences of the background theory but not the query $Q$.
Conversely, suppose that $E_H$ is functional in a visible instance $\visinst$,
and consider any full instance $\fullinst$ that satisfies the background theory and
agrees with $\visinst$ on the visible part.
If there are no tuples in $\hfuncchallenge$, the conjunct above is 
clearly satisfied by its first disjunct.
If there is some tuple $(x,y,y')$ in $\hfuncchallenge$, then by the 
background theory, we must have $E_H(x,y)$ and $E_H(x,y')$, and hence, by
functionality, $y=y'$. In this case, the conjunct above holds via 
the second disjunct.
The functionality of the vertical relation $E_V$ is enforced in an 
analogous way. 

Commutativity of $E_H$ and $E_V$ can be also enforced using a similar technique.
We add a hidden relation $\confchall(x,y,z,u,v)$
with the following sentences in the background theory:
$$
\begin{array}{rcl}
  \confchall(x,y,z,u,v) &\then& E_H(x,y) \\[1ex]
  \confchall(x,y,z,u,v) &\then& E_V(y,u) \\[1ex]
  \confchall(x,y,z,u,v) &\then& E_V(x,z) \\[1ex]
  \confchall(x,y,z,u,v) &\then& E_H(z,v) \ .
\end{array}
$$
A potential tuple in $\confchall(x,y,z,u,v)$ represents the join 
of a triple of nodes moving first horizontally and then vertically 
from $x$ (i.e., $x,y,u$) and a triple going first vertically and 
then horizontally from $x$ (i.e., $x,z,v$). 
For the relations to commute, we must satisfy the query 
$$
  \big(~ \neg \exists ~ x ~ y ~ z ~ u~ v ~ \confchall(x,y,z,u,v) ~\big) 
  ~~\vee~~ 
  \big(~ \exists ~ x ~ y ~ z ~ u ~ \confchall(x,y,z,u,u) ~\big)
$$
in the full instance.
Thus, we add the above conjunct to $Q$.

Putting the various components of $Q$ for different challenges 
together as a Boolean combination of CQ, completes the proof of 
the theorem.
\end{proof}

\myparagraph{The case of conjunctive query views}
As mentioned earlier, the database community has studied the $\PSB$ problem in 
the case where the background theory consist exactly of CQ-view definitions that
determine each visible relation in terms of invisible relations.
Formally, a CQ-view based scenario consists
of a schema $\bfS=\bfS_v \cup \bfS_h$, namely, the
union of a schema for the visible relations and a schema
for the hidden relations, and a set of sentences $\calC$ 
between visible and hidden relations that must be of a particular form.
For each visible relation $R \in \bfS_v$, $\calC$ must contain two
dependencies of the form 
$$
\begin{array}{rcl}
  R(\bar x)              &\then& \exists \bar y ~ \phi_R(\bar x, \bar y) \\[1ex]
  \phi_R(\bar x, \bar y) &\then& R(\bar x)
\end{array}
$$
where $\phi_R$ is a conjunction of atoms over the hidden schema $\bfS_h$,
Furthermore, all sentences in $\calC$ must be of the above forms.
Note that this CQ-view scenario is incomparable in expressiveness to $\gnf$ sentences.

%\michael{I don't see what is wrong with my explanation below.}
The instance-level problems are still well-behaved, because given a visible instance $\visinst$, 
the sentences can be rewritten as $\calC_1 \wedge \calC_2$, where $\calC_1$ 
consists of TGDs from the view relations to the base relations, and $\calC_2$ 
consists of sentences of the form
$V(\vec x) \rightarrow \bigvee_{\vec a \in V(\visinst)} \vec x= \vec a$.
%\todo{Balder says that this is strange, and I agree. I think here we meant to 
%      say that the chase terminates in polynomially many rounds, but only along
 %     one disjunctive branch. But still I don't understand the last sentence about
 %     polynomial size.}
%\michael{I propose to eliminate the polynomial stuff and just claim decidability}
Thus the ``disjunctive chase'' of $\visinst$ with these dependencies will terminate,
since after the first round (where we fire $\calC_1$ dependencies), no new
elements will be created.
%From this we can directly argue that a counterexample
%superinstance for either
%$\PSB$ and $\NSB$ must be of polynomial size.
%michael: could say if we have room one is in NP the other is in conp

The decidability of the $\exists\PSB$ problem follows immediately from these observations 
and Theorem \ref{thm:psbexistcollapselin}, which applies to background theories capturing CQ-view 
definitions. In contrast, for the $\exists\NSB$ problem we prove that

\begin{theorem} \label{thm:undecidexistscqview} 
The $\exists \NSB$ problem under background knowledge given as CQ-view definitions is undecidable.
\end{theorem}

\begin{proof}
As in earlier undecidability results, such as Theorem \ref{thm:undecidexistpsbdisjunction},
we will give the proof for the unrestricted version of the problem, which asserts
the existence of an instance with a $\NSB$, finite or infinite.

We give a reduction from a tiling problem that is specified
by a set of tiles $T$, an initial tile $t_\bot \in T$, and horizontal 
and vertical constraints $H, V \subseteq T \times T$.
In order to match the unrestricted version of $\exists \NSB$, we will deal with the infinite tiling variant, 
thus considering the problem of tiling the infinite grid $\bbN\times\bbN$.

As before, we will have visible relations $E_H$ and $E_V$ representing the 
horizontal and vertical edges of the grid. 
Recall that every visible relation must be associated with a CQ-view definition
on a subset of hidden relations.
In particular, for the relations $E_H,E_V$ it is sufficient to introduce hidden 
copies $E'_H,E'_V$ and enforce the trivial dependencies:
$$
\begin{array}{rcl}
  E_H(x,y) &\iff& E'_H(x,y) \\[1ex]
  E_V(x,y) &\iff& E'_V(x,y) \ . 
\end{array}
$$
Similarly, each node of the grid has to be associated with a tile in $T$, 
and this will be represented by some visible unary relations $U_t$,
together with the corresponding hidden copies $U'_t$. We have  associated sentences in
the background theory:
%$U_t(x) ~\iff~ U'(t)$, for all $t \in T$.
$U_t(x) ~\iff~ U'_t(x)$, for all $t \in T$.

\medskip
As in earlier undecidability results, such as Theorem \ref{thm:booleancombund},  
the first goal is to ensure that for each node, there exists at most one predecessor 
and at most one successor for the relations $E_H$ and $E_V$.
We explain how to ensure this for the successor case and the relation $E_H$,
but similar constructions work for the other cases.
%\gabriele{Removed this, since I think I gave it already above}
%The relation $E'_H$ is associated with another CQ-view definition 
%$\phi_H = C_H(x,y)$, where $C_H$ is a copy of $E_H$ that can be used in view 
%definitions to generate other visible predicates. 
%In addition, 
We introduce a hidden relation $\hfuncchallenge$ of arity $4$, and a visible 
relation $\errhfun$ of arity $3$ with the associated CQ-view definition
$$
  \errhfun(x,y,x') ~\iff~ \hfuncchallenge(x,y,x',y) \ .
$$
Our query $Q$ will contain as a conjunct the following UCQ:
$$
\begin{array}{rcl}
  Q_{\hfuncchallenge} &=& 
  \big(\: \exists ~ x ~ y ~ y' ~ \errhfun(x,y,y') \:\big)
  ~~\vee~~ \\[0.5ex]
  & & \big(\: \exists ~ x ~ y ~ y' ~ 
              \hfuncchallenge(x,y,x,y') \:\et\: E_H(x,y) \:\et\: E_H(x,y') \:\big).
\end{array}
$$
We explain how the subquery $Q_{\hfuncchallenge}$ enforces that every 
element has at most one successor in the relation $E_H$.

Suppose that $\exists\NSB(Q_{\hfuncchallenge},\calC,\bfS)=\true$, 
namely, that there exists an $\bfS_v$-instance $\visinst$ 
such that $\NSB(Q_{\hfuncchallenge},\calC,\bfS,\visinst)=\true$.
The visible relation $\errhfun$ must be empty in $\visinst$,
as otherwise the query $Q_{\hfuncchallenge}$ would be satisfied in every
full instance that agrees with $\visinst$ on the visible part (note that
$\visinst$ is clearly realizable).
Moreover, as $\errhfun$ is empty in $\visinst$, every full instance that 
satisfies the background theory and agrees with $\visinst$ does not contain a 
fact of the form $\hfuncchallenge(x,y,x',y)$.
Now, suppose, by way of contradiction, that there is an element $x$ with 
two distinct $E_H$-successors $y$ and $y'$.
We can construct a full instance that extends $\visinst$ with the single fact 
$\hfuncchallenge(x,y,x,y')$. This full instance satisfies all the  sentences
in $\calC$ and also the query $Q_{\hfuncchallenge}$, thus contradicting
$\NSB(Q_{\hfuncchallenge},\calC,\bfS,\visinst)=\true$.

For the converse direction, we aim at proving that there is a negative query implication
 on $Q_{\hfuncchallenge}$ for those instances that encode valid tilings and
are realizable.
More precisely, we consider a visible instance $\visinst$ in which the relation 
$E_H$ is a function and the relation $\errhfun$ is empty (note that the latter 
condition on $\errhfun$ is safe, in the sense that the considered instance 
$\visinst$ could be obtained from a valid tiling and, being realizable, could 
be used to witness a negative query implication).
We claim that $\NSB(Q_{\hfuncchallenge},\calC,\bfS,\visinst)=\true$. 
Consider an arbitrary full instance $\fullinst$ that agrees with $\visinst$
on the visible part and satisfies the sentences in $\calC$,
and suppose by way of contradiction that $Q_{\hfuncchallenge}$ holds on $\fullinst$.
Then, $\fullinst$ would contain the following facts, for a triple of nodes $x,y,y'$:
$\hfuncchallenge(x,y,x,y')$, $E_H(x,y)$, $E_V(x,y')$.
On the other hand, $\fullinst$ cannot contain the fact $\hfuncchallenge(x,y,x',y)$,
as otherwise this would imply the presence of the visible fact $\errhfun(x,y,x')$.
From this we conclude that $y\neq y'$, which contradicts the functionality of $E_H$.

Very similar constructions and arguments can be used to enforce single successors
in $E_V$, single predecessors in $E_H$ and $E_V$, as well as confluence of $E_H$
and $E_V$.

\medskip
We now explain how we enforce the existential properties of the grid, 
such as $E_H$ being non-empty.
We introduce two nullary relations $\hemptyerror$ and $\hemptyhiddenerror$, 
where the former is visible and the latter is hidden, and we constrain them
via the CQ-view definition
$$
  \hemptyerror ~\iff~ \exists ~ x ~ y ~ \big(\: E_H(x,y) ~\et~ \hemptyhiddenerror \:\big) \ .
$$
We add as a conjunct of our query the following UCQ:
$$
  Q_{\hemptyerror} ~=~ \hemptyerror ~\vee~ \hemptyhiddenerror \ .
$$
Below, we show how this enforces non-emptiness of $E_H$. 

Suppose that $\visinst$ is an $\bfS_v$-instance such that 
$\NSB(Q_{\hemptyerror},\calC,\bfS,\visinst)=\true$. 
We show that in this case the relation $E_H$ is non-empty.
First, note that the fact $\hemptyerror$ must not appear in $\visinst$, since 
otherwise all full instances extending $\visinst$ would satisfy $Q_{\hemptyerror}$
(as $\visinst$ is realizable, there is at least one such full instance).
If $E_H$ were empty, we could set $\hemptyhiddenerror$ to non-empty and thus get 
a contradiction of $\NSB(Q_{\hemptyerror},\calC,\bfS,\visinst)=\true$.

For the converse direction, we consider a visible instance $\visinst$ 
in which the relation $E_H$ is non-empty and $\hemptyerror$ is empty 
(again, such an instance can be obtained from a valid tiling of the infinite grid
and thus can be used to witness a negative query implication).
In any full instance that agrees with $\visinst$ on the visible part, 
$\hemptyhiddenerror$ must agree with $\hemptyerror$, and hence must be empty.
This implies that the query $Q_{\hemptyerror}$ is violated, whence
$\NSB(Q_{\hemptyerror},\calC,\bfS,\visinst)=\true$.

\medskip
Besides requiring that $E_H$ and $E_V$ are non-empty, we must also guarantee 
that for every pair $(x,y)\in E_H$ (resp., $(x,y)\in E_V$), there is a pair 
$(y,z)\in E_V$ (resp., $(y,z)\in E_H$). 
Note that once we have performed this, functionality and confluence will 
ensure that $E_H$ and $E_V$ correctly encode the horizontal and vertical 
edges of the grid.
We explain how to enforce that every pair $(x,y)\in E_H$ has a successor
pair $(y,z)\in E_V$ -- a similar construction can be given for the symmetric property.
We add to our schema another visible relation $\herrorsucc$ of arity $0$, and 
a hidden relation $\herrorhiddensucc$ of arity $1$. 
The associated CQ-view definition is 
%The relation $\herrorsucc_{H,t}$ is defined by $\exists y~ \herrorhiddensucc(x)
%\wedge E_H(x,y) \wedge U_t(x) $.
$$
  \herrorsucc ~\then~ \exists ~x~y~z~~ E_H(x,y) ~\et~ \herrorhiddensucc(y) ~\et~ E_V(y,z) \ .
$$
Moreover, we add as a conjunct of our query the following UCQ:
$$
  Q_{\herrorsucc} ~=~ \herrorsucc ~~\vee~~ 
                      \big(\: \exists ~x~y~ E_H(x,y) ~\et~ \herrorhiddensucc(y) \:\big) \ .
%\big( \bigvee_{t \in T_H}
%(\exists x ~\herrorsucc_{H,t}(x))\big) ~~\vee~~ \bigvee_{t \in T_H} \big( \exists x ~
%\herrorhiddensucc(x) \wedge U_t(x) \big).
$$
We show how this enforces the desired property.

Suppose that there is a visible instance $\visinst$ 
such that $\NSB(Q_\herrorsucc,\calC,\bfS,\visinst)=\true$.
First, observe that the visible relation $\herrorsucc$ must be empty, 
as otherwise all extensions of $\visinst$ would satisfy $Q_\herrorsucc$.
Now, suppose, by way of contradiction, that there is a pair $(x,y)\in E_H$
that has no successor pair $(y,z)\in E_V$. 
In this case, we can construct a full instance that extends $\visinst$
with the hidden fact $\hlabelhiddenerror(y)$. This full instance has $\visinst$
as visible part and satisfies the sentences in the background theory and the query $Q_{\herrorsucc}$.
As this contradicts the hypothesis $\NSB(Q_\herrorsucc,\calC,\bfS,\visinst)=\true$, 
we conclude that for every pair $(x,y)\in E_H$, there is a successor pair $(y,z)\in E_V$.

Conversely, consider a visible instance $\visinst$ that represents a correct 
encoding of the infinite grid and where the visible relation $\herrorsucc$ is empty.
In any full instance that agrees with $\visinst$ on the visible part, 
$\herrorsucc$ must be the same as 
$\exists~x~y~z~ E_H(x,y) ~\et~ \herrorhiddensucc(y) ~\et~ E_V(y,z)$.
In particular, because every node has both a successor in $E_H$ and 
a successor in $E_V$, this implies that the hidden relation $\herrorhiddensucc$ 
cannot contain the node $y$, for any pair $(x,y)\in E_H$.
Hence the query $Q_\herrorsucc$ is necessarily violated, and this
proves that $\NSB(Q_\herrorsucc,\calC,\bfS,\visinst)=\true$.

\medskip
Now that we have enforced a grid-like structure on the relations $E_H$ and $E_V$, 
we consider the relations $U_t$ that encode a candidate tiling function.
Using similar techniques, we can ensure that every node of the grid has an associated tile.
More precisely, we enforce that, for every pair $(x,y)\in E_H$, the element $x$ 
must also appear in $U_t$, for some tile $t\in T$
% -- analogous constructions and arguments 
%can be given for the pairs $(x,y)\in E_V$.
%Gabriele: not needed
%The visible schema has a visible relation $C_t$ for each tile $t$.
We add a visible relation $\hlabelerror_t$ of arity $0$ for each tile $t\in T$ 
and a hidden relation $\hlabelhiddenerror$ of arity $1$.
The associated CQ-view definitions are of the form
$$
  \hlabelerror_t ~\iff~ \exists ~ x ~ y ~ E_H(x,y) ~\et~ \hlabelhiddenerror(x) ~\et~ U_t(x) \ .
$$
We add as conjunct of our query the following UCQ:
$$
  Q_\hlabelerror ~=~ \bigvee_{t \in T} \exists ~ x ~ y ~ 
                     \big(\: \hlabelerror_t(x,y) \:\big) ~\vee~ 
                     \big(\: E_H(x,y) \et \hlabelhiddenerror(x) \:\big) \ .
$$
We prove that the above definitions enforce that all nodes that appear
in the first column of the relation $E_H$ have at least one associated tile.

Consider a visible instance $\visinst$ such that $\NSB(Q_\hlabelerror,\calC,\bfS,\visinst)=\true$.
For each tile $t$, the visible relation $\hlabelerror_t$ must be empty, as otherwise all
extensions of $\visinst$ would satisfy $Q_\hlabelerror$.
Suppose, by way of contradiction, that there is a node $x$ that appears in the first 
column of the visible relation $E_H$, but does not appear in any relation $U_t$, with 
$t\in T$. We can construct a full instance where the relation $\hlabelhiddenerror$ 
contains the element $x$. This instance would then satisfy the query $Q_\hlabelerror$,
thus contradicting $\NSB(Q_\hlabelerror,\calC,\bfS,\visinst)=\true$.

For the converse, consider a visible instance $\visinst$ in which the relation
$E_H$ is non-empty (as enforced in the previous steps) and, for all pairs $(x,y)\in E_H$,
there is a tile $t\in T$ such that $x\in U_t$.
Furthermore, assume that all the relations $\hlabelerror_t$, with $t\in T$, 
in this visible instance are empty. Note that such an instance $\visinst$ is 
realizable and hence can be obtained from a valid tiling (if there is any)
and used as a witness of a negative query implication. 
In every full instance that agrees with $\visinst$ and satisfies the background theory, 
$\hlabelerror_t$ must be the same as 
$\exists ~x~y~ \hlabelhiddenerror(x) \et E_H(x,y) \et U_t(x)$.
In particular, because every node is associated with some tile,
this implies that the hidden relation $\hlabelhiddenerror$ cannot
contain the node $x$, for any pair $(x,y)\in E_H$.
Hence the query $Q_\hlabelerror$ is necessarily violated, and this
proves that $\NSB(Q_\hlabelerror,\calC,\bfS,\visinst)=\true$.

\medskip
We also need to guarantee that each node has at most one associated tile. 
This property can be easily enforced by the subquery 
$$
  Q_{\twolabel} ~=~ \bigvee_{t \neq t'} \exists x~ U_t(x) \et U_{t'}(x) \ .
$$
Finally, we enforce that the encoded tiling function respects the horizontal 
and vertical constraints using the following UCQ:
$$
\begin{array}{rcl}
  Q_{\cerror} &=& \displaystyle\bigvee\nolimits_{(t,t') \not \in H} ~
                  \big(\: \exists ~x ~ y ~ E_H(x,y) \et U_t(x) \et U_{t'}(y) \:\big)
                  ~~\vee~~ \\[0.5ex]
              & & \displaystyle\bigvee\nolimits_{(t,t') \not \in V} ~
                  \big(\: \exists ~x ~ y ~ E_V(x,y) \et U_t(x) \et U_{t'}(y) \:\big) \ .
\end{array}
$$

\medskip
Summing up, if we let $Q$ be the 
disjunction %conjunction  %michael: yes, thanks for catching this
of all previous queries, we know 
that $\exists\NSB(Q,\calC,\bfS)=\true$ if and only if there exists a valid tiling
of the infinite grid $\bbN\times\bbN$. 
%This shows that the former problem, where the 
%instances are allowed to be infinite, is undecidable. 
%\gabriele{This is an attempt to explain the reduction for the finite setting.}
%We could strengthen this undecidability result to the settings where instances 
%are required to be finite. This can be done by modifying the above constructions
%so as to encode tilings of portions of the grid that are infinite, but where every 
%row and every column is finite. For example, tiling can be restricted to cells of 
%the form $(i,j)$, with $i\le j\le 2i$.
\end{proof}

%% file: conc.tex
\section{Conclusions} \label{sec:conc}
This work gives a detailed examination of inference of information
from complete knowledge about a subset of the signature coupled with background
knowledge about the full signature. Both the information and
the background knowledge are expressed by  logical sentences.
In future work we will look at mechanisms for ``restricted access'' that are finer-grained
than just exposing the full contents of a subset of the schema relations.
One such mechanism consists 
language-based restrictions -- the ability to  evaluate open formulas over the schemas in a fragment of the logic.
Another mechanism consists of functional interfaces -- for example, the ``access method'' interfaces studied
in works such as \cite{thebook,ustods}.
% and  binding-pattern based restrictions (user can do lookups
%with values that have been acquired from prior lookups).

%% file: gnftoaut/gnftoaut.tex
\input{gnftoaut/macros}

\section{Proof of exponential time satisfiability for $\gnf$ with fixed width, fixed CQ-rank, and fixed arity of schema}
\label{sec:appendixgnfaut}
\input{gnftoaut/gnfbasics}

\subsection{Tree-like models and automata} \label{subsec:treelike}
\input{gnftoaut/basictreecodes}
\subsection{Automata background} \label{subsec:aut}
\input{gnftoaut/automatagnf}

\subsection{Decision procedure for normal-form $\gnf$ without
equality and constants} \label{subsec:normal}
In giving the automata constructions in this section,
we will start by working with $\gnf$ formulas $\varphi$ in a normal form,
 \emph{\gnnf}, similar to one
introduced in \cite{gnfj}.
Throughout this section, we also  assume  that the formulas 
\emph{do not use equality or constants}.

\input{gnftoaut/gnnormal}

\input{gnftoaut/gnfautdecidespecializations}

%\subsection*{Dropping the normal form restriction} \label{sec:gnfminus}
\subsection{Handling equality and constants} \label{subsec:equality}
\input{gnftoaut/equality}

\subsection{Additional remarks: relationship to bounds for general $\gnf$} \label{subsec:gnfminus}
\input{gnftoaut/convertgnnorm}

%%%%%%%%%%%%%%%%%%%%%%%%%%%%%%%%%%%%%%%%%%%%%%%%%%%%%%%%%%%%%%%%%%%
%%%%%%%%%%%%%%%%%%%%%%%%%%%%%%%%%%%%%%%%%%%%%%%%%%%%%%%%%%%%%%%%%%%

%% file: gnftoaut/macros.tex
\newcommand{\utuple}{\kw{GTuple}}
\newcommand{\treelogic}{\kw{Tree}}
\newcommand{\gnfpk}{\text{$\gnfp^k$}\xspace}
\newcommand{\sigoriginal}{\sigL}
\newcommand{\sigtarget}{\sigma'}
\renewcommand{\L}{\mathrm{L}}
\newcommand{\phiL}{\phi_{\L}}
\newcommand{\R}{\mathrm{R}}
\newcommand{\phiR}{\phi_{\R}}
\newcommand{\cform}[2]{\text{consistent}_{#1,#2}}
\newcommand{\consistent}{\text{consistent}}

\newcommand{\sigL}{\sigma_{\L}}
\newcommand{\sigR}{\sigma_{\R}}
\newcommand{\codepath}{\textit{Path}}
\newcommand{\sigcode}[2]{\Sigma^{\text{code}}_{#1,#2}}

\newcommand{\owq}{\kw{OWQ}}
\newcommand{\nfta}{\textup{NTA}_{\textup{fin}}}
\newcommand{\reachablestates}{\textup{ReachState}}
\newcommand{\gnfp}{\kw{GNFP}}
\newcommand{\gfp}{\kw{GFP}}
\newcommand{\unfp}{\kw{UNFP}}
\newcommand{\lfplogic}{\kw{LFP}}
\newcommand{\LFP}{\mu}
\newcommand{\GFP}{\nu}
\newcommand{\LFPA}[2]{\LFP_{#1,#2}}
\newcommand{\GFPA}[2]{\GFP_{#1,#2}}
\newcommand{\guarded}{\operatorname{gdd}}

\newcommand{\Card}{\kw{Card}}
\newcommand{\Eq}{\kw{Eq}}

\newcommand{\nf}{\text{normal form}\xspace}
\newcommand{\calA}{\autA}
\newcommand{\sigk}{\sigcode{\sigma}{k}}
\newcommand{\otree}{t_<}
\newcommand{\bagnames}[1]{\operatorname{names}(#1)}
\newcommand{\decode}[1]{\operatorname{decode}(#1)}
\newcommand{\trans}{+}
\newcommand{\sigmab}{\sigma_b}
\newcommand{\sigmad}{\sigma_d}
\newcommand{\mysize}[1]{|#1|}
\newcommand{\mysubparagraph}[1]{\subparagraph{#1.}}
\newcommand{\guardedb}{\operatorname{gdd}}
\newcommand{\calN}{\mathcal{N}}
\newcommand{\paramsk}{N_k}
\newcommand{\instance}{I}
\newcommand{\elems}[1]{\operatorname{elems}(#1)}
\newcommand{\twowayaltinf}{\kw{2TA}^\omega}
\newcommand{\btwowayaltinf}{\kw{2ABT}^\omega}
\newcommand{\ptwowayaltinf}{\kw{2APT}^\omega}
\newcommand{\ponewayndinf}{\kw{1NPT}^\omega}
\newcommand{\bonewayndinf}{\kw{1NBT}^\omega}

\newcommand{\fA}{\mathfrak{A}}
\newcommand{\fB}{\mathfrak{B}}
\newcommand{\fU}{\mathfrak{U}}
\newcommand{\fM}{\mathfrak{M}}
\newcommand{\strM}{M}%change to frak later
\newcommand{\N}{\mathbb{N}}

\newcommand{\gnkinvar}{\gnf^k}
\newcommand{\sgnkinvar}{B\gnf^k}
\newcommand{\sunkinvar}{\logic{BUN}^k}
\newcommand{\gnlinvar}{\logic{GN}^l}
\newcommand{\gnninvar}{\logic{GN}^n}
\newcommand{\gnminvar}{\logic{GN}^m}
\newcommand{\ginvar}{\logic{G}}

\newcommand{\bunravelk}[1]{\cU_{\textup{BGN}^k}( #1 )}

\newcommand{\unravelk}[1]{\cU_{\gnkunravel}( #1 )}
\newcommand{\unravelm}[2]{\cU_{\textup{GN}^m[#2]}( #1 )}
\newcommand{\unravell}[2]{\cU_{\textup{GN}^l[#2]}( #1 )}
\newcommand{\sunravelk}[2]{\cU_{\gnkunravel, #2}( #1 )}

\renewcommand{\restriction}{\mathord{\upharpoonright}}
\newcommand{\restrict}[2]{ #1\restriction_{#2} }
\newcommand{\gnfk}{\gnf^k}
%\newcommand{\restrict}[2]{ #1\restriction_{#2} }

%\newcolumntype{C}[1]{>{\centering\let\newline\\\arraybackslash\hspace{0pt}}m{#1}}
%\newcolumntype{L}[1]{>{\let\newline\\\arraybackslash\hspace{0pt}}m{#1}}

%\usepackage[thmmarks]{ntheorem}
%\theoremheaderfont{\bfseries}
%\theorembodyfont{\normalfont}
%\theoremseparator{:}
%\theoremsymbol{$\blacksquare$}

\newenvironment{myexmp}{\refstepcounter{myexmp}\par\medskip
\noindent\textbf{Example~\themyexmp.}}{\null\hfill$\triangleleft$\medskip}
%\newcounter{myexmp}[chapter]
%\renewcommand{\themyexmp}{\thechapter.\arabic{myexmp}}

\def\DIAMOND{\LTLdiamond}
\def\DIAMONDM{\LTLdiamondminus}
\def\BOX{\LTLsquare}
\def\BOXM{\LTLsquareminus}

%\newenvironment{myexmp}{\refstepcounter{myexmp}\par\medskip
%\noindent\textbf{Example~\thechapter.\themyexmp.}}{\null\hfill$\triangleleft$\medskip}
%\newcounter{myexmp}
%\makeatletter
%\@addtoreset{myexmp}{chapter}
%\makeatother

\newcommand{\munravel}{\kw{ModalUnravel}}
\newcommand{\gunravel}{\kw{GUnravel}}
\newcommand{\gunraveltree}{\kw{GUnravalTree}}
\newcommand{\gnkunravel}{\kw{GNUnravel}_k}
\newcommand{\gnkunraveltree}{\kw{GNUnravelTree}_k}
\newcommand{\mroot}{\kw{Root}}
\newcommand{\mcopy}{\kw{Copy}}

\newcommand{\mso}{\kw{MSO}}
\newcommand{\gso}{\kw{GSO}}
\newcommand{\so}{\kw{SO}}

\newcommand{\related}{\kw{Related}}
\newcommand{\evans}{\text{``$\kw{Evans}$''}}
\newcommand{\kennedy}{\text{``$\kw{Kennedy}$''}}
\newcommand{\thompson}{\text{``$\kw{Thompson}$''}}

\newcommand{\rqfo}{\kw{RQFO}}
\newcommand{\inducedaxs}{\Gamma_{\kw{Ind}}}
\newcommand{\atts}{\kw{atts}}
\newcommand{\joe}{\mbox{``joe''}}
\newcommand{\cseqrep}{\kw{ChaseSeqRep}}
\newcommand{\vecinf}{\overrightarrow{\kw{inf}}}
\newcommand{\dom}{\domain}
\newcommand{\adom}{\kw{adom}}
\newcommand{\push}{\kw{Push}}
\newcommand{\popswap}{\kw{PopSwap}}
\newcommand{\stackcontent}{\kw{StackCon}}
\newcommand{\return}{\kw{Return}}
\newcommand{\free}[1]{\kw{Free}(#1)}
\newcommand{\rulesof}{\kw{RulesOf}}
\newcommand{\touspjnegquery}{\kw{ToUSPJ^\neg}}
\newcommand{\seqconsts}{\kw{SeqConsts}}
\newcommand{\indomain}{\kw{InDomain}}
\newcommand{\inmap}{\kw{InMap}}
\newcommand{\outmap}{\kw{OutMap}}

\newcommand{\aschemaconst}{\kw{Const(\aschema)}}
\newcommand{\constsigma}{\aschemaconst}
\newcommand{\infacccon}{\kw{InfAccCopy}}
\newcommand{\methods}{\kw{Methods}}
\newcommand{\bagguard}{$A$}
\newcommand{\IS}{induced-subinstance}
\newcommand{\ISA}{induced-subinstance-access}

\newcommand{\accbind}{\kw{AccBind}}
\newcommand{\smith}{\text{``$\kw{Smith}$''}}
\newcommand{\jones}{\text{``$\kw{Jones}$''}}
\newcommand{\mathematics}{\kw{mathematics}}
\newcommand{\fintable}{\mathit{fin}}

\newcommand{\relquan}{\kw{RelQuan}}
\newcommand{\cqof}{\kw{CQOf}}
\newcommand{\planof}{\kw{PlanOf}}

\newcommand{\wrt}{w.r.t.~}
\newcommand{\ie}{i.e.,~}
\newcommand{\eg}{e.g.~}
\newcommand{\footurl}[1]{\footnote{\textsf{\small #1}}}

\newcommand{\includedin}{\subseteq}
\newcommand{\origconsts}{\kw{OrigConsts}}
\newcommand{\paramtab}{\kw{ParamTable}}
\newcommand{\inp}{\kw{input}}
\newcommand{\allatts}{\kw{allatts}}

\newcommand{\lookupname}{\kw{lookupname}}
\newcommand{\facultynames}{\kw{facultynames}}
\newcommand{\getnames}{\kw{getnames}}
\newcommand{\facultyname}{\kw{facultyname}}
\newcommand{\localnames}{\kw{localname}}
\newcommand{\homm}{\kw{hom}}
\newcommand{\latestnull}{\kw{latestnull}}
\newcommand{\commands}{\kw{Commands}}

\newcommand{\mapsfrom}{:=}
\newcommand{\chaseconfig}{\kw{ChaseConfig}}
\newcommand{\success}{\kw{Success}}

\newcommand{\subq}{\kw{SQ}}

\newcommand{\fortrans}{\forax}
\newcommand{\forwardax}{\forax}
%\kw{AcSch}_{Forward}}
\newcommand{\tnode}{\kw{TabNode}}
\newcommand{\orm}[1]{{\color{orange}#1}}
\newcommand{\grm}[1]{{\color{green}#1}}
\newcommand{\brm}[1]{{\color{brown}#1}}
\newcommand{\bnm}[1]{{\color{blue}#1}}
\newcommand{\redm}[1]{{\color{red}$#1$}}
\newcommand{\red}[1]{{\color{red}$#1$}}
\newcommand{\rednm}[1]{{\color{red}#1}}
\newcommand{\blue}[1]{\bnm{#1}}

\newcommand{\greenc}{\kw{Green}}
\newcommand{\orangec}{\kw{Orange}}
\newcommand{\bluec}{\kw{Blue}}

\newcommand{\asch}{\kw{Sig}}
\newcommand{\aschema}{\asch}
\newcommand{\abag}{\kw{b}}
\newcommand{\bags}{\kw{Bags}}
\newcommand{\rootbag}{\kw{RootB}}
\newcommand{\Ann}{\kw{Ann}}
\newcommand{\posexistsineq}{\exists^{+, \neq}}
\newcommand{\posexists}{\exists^+}
\newcommand{\existsineq}{\exists^{\neq}}
\newcommand{\existsfo}{\exists}% existential formula
\newcommand{\deacc}{\kw{DeAcc}}
\newcommand{\lostarski}{\L{}o\'s-Tarski}
\newcommand{\relc}{\kw{RelAcc}}
\newcommand{\bindpat}{\kw{BindPatt}}
\newcommand{\trule}[2]{\ensuremath{\begin{array}{c}#1 \\ \hline #2\end{array}}}
\newcommand{\interp}{\xrightarrow{~\text{int}~}}
\newcommand{\nneg}{\mathop{\sim}}
\newcommand{\Left}{\textsc{left}}
\newcommand{\Right}{\textsc{right}}
\newcommand{\T}{{\cal T}}
\newcommand{\B}{{\cal B}}
\renewcommand{\S}{{\cal S}}

% Uncomment either of the next 2 lines to turn on/off side notes
% \newcommand{\sidenote}[4]{ #4 }
\newcommand{\sidenote}[4]{\todo[fancyline,color=#2,linecolor=#2]{\textbf{[#1]} #3}\textcolor{#2}{#4}}

\newcommand{\backax}{\kw{Ax}_{\kw{Back}}}
\newcommand{\forax}{\kw{Ax}_{\kw{For}}}
\newcommand{\aseq}{\kw{Seq}}
\newcommand{\toplan}{\kw{ToPlan}}
\newcommand{\tuples}{\kw{Tuples}}
\newcommand{\select}{\sigma}
\newcommand{\project}{\pi}
\newcommand{\join}{\bowtie}
\newcommand{\union}{\cup}
\newcommand{\diff}{-}
\newcommand{\adiff}{-}
\newcommand{\rename}{\rho}
\newcommand{\renaming}{\leadsto}
\newcommand{\accpart}{\kw{AccPart}}
\newcommand{\accs}{\kw{AcSch}}
\newcommand{\accsb}{\kw{AcSch}^\leftrightarrow}
\newcommand{\accsbp}{\kw{AltAcSch}^\leftrightarrow}
\newcommand{\accsnegalt}{\kw{AltAcSch}^\neg}
\newcommand{\False}{\false}
\newcommand{\True}{\true}
\newcommand{\accsneg}{\kw{AcSch}^\neg}
\newcommand{\command}{\kw{Command}}
\newcommand{\comms}{\kw{Comms}}
\newcommand{\accout}{\kw{AccOut}}

\newcommand{\anacc}{\kw{Acc}}

\newcommand{\globdom}{\geq_{GD}}

\newcommand{\AD}{\ensuremath{AD}}
\newcommand{\spj}{\ensuremath{{SPJ}}^{\neq}\xspace}
\newcommand{\spjneg}{\ensuremath{{SPJAD}^{\neq}} \xspace}
\newcommand{\espjneg}{\ensuremath{{SPJAD}} \xspace}

\newcommand{\uspj}{\ensuremath{{USPJ}}^{\neq}\xspace}
\newcommand{\USPJ}{\uspj}
\newcommand{\uspjneg}{\ensuremath{{USPJAD}^{\neq}}\xspace}
\newcommand{\ESPJ}{\ensuremath{{SPJ}}\xspace}
\newcommand{\espj}{\ESPJ}
\newcommand{\EUSPJ}{\ensuremath{{USPJ}}\xspace}
\newcommand{\euspj}{\EUSPJ}
\newcommand{\EUSPJneg}{\ensuremath{{USPJAD}}\xspace}
\newcommand{\euspjneg}{\EUSPJneg}

\newcommand{\dept}{\kw{Dept}}
\newcommand{\emp}{\kw{Emp}}
\newcommand{\uemployee}{\kw{UEmployee}}
\newcommand{\employee}{\kw{Employee}}
\newcommand{\manager}{\kw{Manager}}

\newcommand{\deptid}{\kw{deptid}}
\newcommand{\mgrid}{\kw{mgrid}}
\newcommand{\ename}{\kw{ename}}
\newcommand{\did}{\deptid}

\newcommand{\profname}{\kw{profname}}
\newcommand{\studentid}{\kw{studid}}
\newcommand{\studid}{\studentid}
\newcommand{\dname}{\kw{dname}}
\newcommand{\profid}{\kw{profid}}
\newcommand{\lname}{\kw{lname}}
\newcommand{\student}{\kw{Student}}
\newcommand{\professor}{\kw{Professor}}

\newcommand{\profinfo}{\kw{Profinfo}}
\newcommand{\studinfo}{\kw{Studentinfo}}
\newcommand{\univdirectory}{\kw{Udirectory}}
\newcommand{\udirectory}{\univdirectory}
\newcommand{\studentinfo}{\studinfo}
\newcommand{\researcher}{\kw{Researcher}}
\newcommand{\lecturer}{\kw{Lecturer}}
\newcommand{\universitydirectory}{\univdirectory}
\newcommand{\eid}{\kw{eid}}
\newcommand{\employeeid}{\kw{employeeid}}
\newcommand{\studname}{\kw{studname}}
\newcommand{\lastname}{\kw{lastname}}
\newcommand{\onum}{\kw{onum}}
\newcommand{\accessible}{\kw{accessible}}
\newcommand{\accessconst}[1]{\accessible(\config_{#1})}
\newcommand{\acc}[1]{\kw{Accessed} #1}
\newcommand{\accq}[1]{\kw{InfAcc} #1}
\newcommand{\accfactq}{\kw{InfAccQuery}}
\newcommand{\accfacts}{\kw{AccFacts}}
\newcommand{\accqp}[1]{\kw{AltInfAcc} #1}
\newcommand{\infacc}[1]{\kw{InfAcc} #1}
\newcommand{\bestplan}{\kw{BestPlan}}
\newcommand{\bestcost}{\kw{BestCost}}
\newcommand{\oldbestcost}{\kw{OldBestCost}}
\newcommand{\candidates}{\kw{Candidates}}
\newcommand{\plan}{\kw{Plan}}
\newcommand{\plantree}{\kw{PlanTree}}
\newcommand{\prooftree}{\kw{ProofTree}}

\newcommand{\config}{\kw{config}}
\newcommand{\mt}{\kw{mt}}
\newcommand{\aplan}{\kw{PL}}
\newcommand{\outcome}[3][]{\ensuremath{\llbracket#2\ifstrempty{#1}{}{\mid#1}\rrbracket_{#3}}}
\newcommand{\interpret}[2]{\outcome{#1}{#2}}
\newcommand{\cost}{\kw{Cost}}
\newcommand{\pointer}{\kw{Pointer}}
\newcommand{\parent}{\kw{Parent}}
\newcommand{\children}{\kw{Children}}
\newcommand{\childof}{\kw{ChildOf}}
\newcommand{\descendantof}{\kw{DecendantOf}}
\newcommand{\oldcost}{\kw{OldCost}}
\newcommand{\newcost}{\kw{NewCost}}
\newcommand{\atomiccost}{\kw{AtomicCost}}
\newcommand{\haspointer}{\kw{HasPointer}}
\newcommand{\isroot}{\kw{IsRoot}}
\newcommand{\prune}{\kw{Prune}}
\newcommand{\goodplans}{\kw{GoodPlans}}
\newcommand{\oldplans}{\kw{OldPlans}}
\newcommand{\ids}{\kw{Ids}}
\newcommand{\names}[1]{\bagnames{#1}}
\newcommand{\uname}{\kw{uname}}
\newcommand{\addr}{\kw{addr}}
\newcommand{\uid}{\kw{uid}}
\newcommand{\plandag}{\kw{PlanDag}}
\newcommand{\planspace}{\kw{PlanSpace}}
\newcommand{\arule}{\kw{Ax}^{mt}_R}
\newcommand{\ruleapp}{\kw{AtomicProof}}
\newcommand{\gcompose}{\kw{CompProof}}
\newcommand{\configs}{\kw{Configs}}
\newcommand{\backup}{\kw{BackUpBestCost}}
\newcommand{\uplan}{\kw{ChasePlan}}
\newcommand{\initialize}{\kw{Initialize}}
\newcommand{\chooseax}{\kw{Choose}}
\newcommand{\reps}{\kw{Reps}}
\newcommand{\repof}{\kw{RepOf}}
\newcommand{\witnessmap}{\kw{WitnessMap}}
\newcommand{\candidateconfigs}{\kw{CandConf}}
\newcommand{\dominated}{\kw{dominated}}
\newcommand{\sucdom}{\kw{SucDom}}
\newcommand{\candidatepairs}{\kw{CandPairs}}
\newcommand{\reppairs}{\kw{RepPairs}}
\newcommand{\pequiv}{\equiv_{\kw{Fact}}}
\newcommand{\factdomh}{\preceq^h_\kw{Fact}}
\newcommand{\factdom}{\preceq_\kw{Fact}}
\newcommand{\costdom}{\preceq_\kw{Cost}}
\newcommand{\methdom}{\preceq_\kw{mt}}
\newcommand{\rtdom}{\preceq_\kw{RT}}

\renewcommand{\B}{\mathcal{B}}
\newcommand{\correct}[1]{\textcolor{red}{\textbf{Check:} \textbf{#1}}}
\renewcommand{\phi}{\varphi}

\newtheorem{op}{Open question}
%\newtheorem{lemma}{Lemma}
 %\newtheorem{algo}{Algorithm}
% \newenvironment{proof}{\medskip \noindent \bf Proof \rm}{$\Box$ \medskip}
%\newenvironment{example}                        % Example environment.
%  {%\refstepcounter{chapter}
%\trivlist\item       %   Acts just like a theorem
%  [\hskip\labelsep\bf Example \thechapter]}%    %   environment, except that
%  {\endtrivlist}

% \newcommand{\myparagraph}[1]{{\bf #1.}}

\newcommand{\totalit}{\kw{TotalIt}}
\newcommand{\firstmatch}{\kw{ItFirstMatch}}
\newcommand{\bestmatch}{\kw{ItBestPl}}
\newcommand{\totalfact}{\kw{TotalFacts}}
\newcommand{\dominance}{\kw{DomCh}}
\newcommand{\globalequiv}{\kw{GlobalEq}}
\newcommand{\costprune}{\kw{CostPrune}}
\newcommand{\totalfacts}{\kw{Facts}}
\newcommand{\lra}{\longrightarrow}
\newcommand{\factsof}{\kw{FactsOf}}

\newcommand{\ev}{\kw{EV}}
\newcommand{\aninst}{\kw{I}}
\newcommand{\aquery}{\kw{Q}}
\newcommand{\dtime}{\kw{DTIME}}

\newcommand{\pdq}{$\textsc{PDQ}$\xspace}

\newcommand{\accessop}{\kw{AccessOp}}
\newcommand{\depjoin}{\overrightarrow{\bowtie}}
\newcommand{\officeinfo}{\kw{OfficeInfo}}
\newcommand{\phone}{\kw{phone}}
\newcommand{\bname}{\kw{Bname}}
\newcommand{\ax}{\kw{Ax}}
\newcommand{\name}{\kw{Name}}
\newcommand{\pid}{\profid}
\newcommand{\offid}{\kw{Offid}}
\newcommand{\officein}{\kw{OfficeIn}}
\newcommand{\id}{\kw{Id}}
\newcommand{\office}{\kw{Office}}

\newcommand{\SELECT}{\kw{SELECT}}
\newcommand{\FROM}{\kw{FROM}}
\newcommand{\WHERE}{\kw{WHERE}}
\newcommand{\AND}{\kw{AND}}
\newcommand{\advisorname}{\kw{advisorname}}
\newcommand{\advisorid}{\kw{advisorid}}
\newcommand{\blank}{\_}
\newcommand{\act}{\kw{Direction}}
\newcommand{\stay}{\kw{Stay}}
\newcommand{\down}{\kw{Down}}
\newcommand{\up}{\kw{Up}}
\newcommand{\leftm}{\kw{Left}}
\newcommand{\rightm}{\kw{Right}}
\newcommand{\outputalph}{O}
\newcommand{\manyone}{\leq_{manyone}}
\newcommand{\ptimered}{\leq_{PTIME}}
\newcommand{\timetm}{\kw{Time}}
\newcommand{\propsat}{\kw{PropSAT}}
\newcommand{\threesat}{\kw{3SAT}}

\newcommand{\accept}{\kw{Accept}}
\newcommand{\noaccept}{\kw{NoAccept}}
\newcommand{\machtime}{\kw{Time}}
\newcommand{\machspace}{\kw{Space}}
\newcommand{\lttime}{<_{time}}
\newcommand{\ltspace}{<_{space}}
\newcommand{\hashead}{\kw{HasHead}}
\newcommand{\hassymbol}{\kw{HasSymbol}}
\newcommand{\fosat}{\kw{FOSat}}
\newcommand{\power}{{\cP}}
\newcommand{\V}{{\cal V}}
\newcommand{\code}{\kw{code}}
\newcommand{\treecode}{\code}

\newcommand{\autA}{\mathcal{A}}
\newcommand{\autB}{\mathcal{B}}
\newcommand{\autC}{\mathcal{C}}
\newcommand{\autM}{\mathcal{M}}
\newcommand{\autN}{\mathcal{N}}
\newcommand{\mvb}[1]{{\bf mvb:} \textcolor{red}{#1}}
\newcommand{\set}[1]{\left\{ #1\right\}}
\newcommand{\sset}[1]{\{ #1\}}
\newcommand{\QEve}{Q_E}
\newcommand{\Qand}{Q_\wedge}
\newcommand{\QAdam}{Q_A}
\newcommand{\Qor}{Q_\vee}
\newcommand{\ddown}{\downarrow}
\newcommand{\dup}{\uparrow}
\newcommand{\dstay}{0}
\newcommand{\directions}{\kw{Dir}}
\newcommand{\closure}[1]{\kw{cl}(#1)}
\newcommand{\dneighbor}{\updownarrow}
\newcommand{\qmovesibling}{\text{Eve-choose-sibling}}
\newcommand{\qmovesiblinguniv}{\text{Adam-choose-sibling}}
\newcommand{\qmoveparent}{\text{goto-parent}}
\newcommand{\qselect}{\text{check}}
\newcommand{\state}[1]{\langle #1 \rangle}
\newcommand{\Buchi}{B\"uchi\xspace}
\newcommand{\autsig}{\Sigma}
\newcommand{\namesk}{\paramsk}

\newcommand{\specializations}[1]{\kw{Spec}(#1)}

\newcommand{\clgnf}[1]{\kw{cl}_{\kw{GN}}(#1)}
\newcommand{\const}[1]{\kw{Const}(#1)}

%% file: gnftoaut/gnfbasics.tex
In this appendix, we give details of the following result:

\medskip

Satisfiability of $\gnf$ sentences 
is decidable
in exponential time if the arity of the relations
used in the sentence
is fixed, and further certain parameters of the sentence are fixed:
the \emph{width}, and the \emph{CQ-rank}.

\medskip

A doubly-exponential bound on satisfiability of $\gnf$ was proven in  the papers
where $\gnf$ was introduced \cite{gnficalp,gnfj}. However the argument
was by reduction to satisfiability of the guarded fragment.
Conversions of $\gnf$ formulas to automata, and comments about what controls their complexity, are implicit in a number
of other works \cite{uslics15bounded,uslics14,icalp17}. 
But the conversions are performed for richer logics than $\gnf$.
This means firstly that they introduce  many 
complications that are unnecessary for $\gnf$, and secondly that they
do not provide the precise statements for $\gnf$  that we require in
our analysis of  inference problems.

Here we give a direct reduction of satisfiability of  $\gnf$ 
to emptiness testing for a tree automaton. The translation allows us to track
the complexity of satisfiability in a more fine-grained way,
including the collapse to exponential time when the arity of the relations,
the width, and the CQ-rank is fixed.

We will start in Subsection \ref{subsec:treelike} explaining
the tree-like model property,  and in Subsection \ref{subsec:aut}
giving background on the automaton formalism we use.
In Subsection \ref{subsec:normal} we show
decidability for the case of $\gnf$ without equality and
constants, restricting to sentences of a special kind (``normal form'').
Finally in Subsection \ref{subsec:equality} we extend to full $\gnf$, with equalities
and constants.
We close in Subsection \ref{subsec:gnfminus} with some remarks relating the results here
with the bounds in the absence of any normal form restriction.

%% file: gnftoaut/basictreecodes.tex
%\section{Tree codes for $\gnf$ without equality}
The first step in showing decidability of $\gnf$ satisfiability
is to show that for any sentence $\phi$ there is a number
$k$, easily computed from $\phi$, such that:
 if $\phi$ is satisfiable, it is satisfiable
over structures that are ``$k$-tree-like'': that is
a structure that is coded by a tree, where each
vertex in the tree represents at most $k$ elements in the structure.

In this section, we will explain the tree-like model
property. In doing so we will \emph{restrict to $\gnf$
sentences that do not have equality or constants}. The extension
to equality and constants will be given in Subsection \ref{subsec:equality}.

We start by describing  what these tree codes look like in detail.

For a number $k$ we let $\paramsk= \set{1, \dots, 2 \cdot k}$.
This is a finite set of \emph{names} that will be used to describe the elements
represented in a given node in the tree.

Given a relational signature  $\sigma$ and a number
$k$,  the \emph{$k$-code signature}, $\sigk$
contains:
\begin{itemize}
\item a unary predicate $D_a$ for all $a \in \paramsk$
\item unary predicates
$R_{\vec{a}}$ for all $R \in \sigma$ of arity $j$ and all $\vec{a} \in \paramsk^j$
\end{itemize}
Informally, $D_a(v)$ indicates that $a$ is a name
in the node $v$ in the tree code, while
$R_{\vec{a}}(v)$ indicates that $R$ holds for
the elements represented by the names $\vec{a}$ at~$v$.

Neighboring nodes may describe
overlapping pieces of the structure.
This will be implicitly coded based on repeated use of names:
if some name appears in two neighboring nodes,
then the same element is being described in both nodes.
This is why $\paramsk$ has $2k$ names,
even though at most $k$ names are used in a single node.

For a vertex $v$ in a $\sigk$ tree $\tree$, let $\names{v} := \set{ a \in \paramsk : D_a \mbox{ holds of } v }$.
This denotes the
set of \emph{names} used for elements in node~$v$.

A \emph{consistent $\sigk$-tree} is a $\sigk$-tree such that
every node $v$ satisfies
\begin{compactitem}
\item $\size{\names{v}} \leq k$ %+ \size{\const{\sigma}}$; 
\item for all $R_{\vec{a}} \in \sigk$, if $R_{\vec{a}}(v)$ then $\vec{a} \subseteq \names{v}$;
\end{compactitem}
When $\sigma$ is clear from context,
such a tree will also be called a \emph{$k$-code}.

We now describe the structure coded by a $k$-code formally.
Given a consistent tree $\tree$ and a local name $a$,
we say nodes $u$ and $v$ are $a$-connected
if
there is a sequence of nodes $u = w_0, w_1, \dots, w_j = v$
such that $w_{i+1}$ is a parent or child of $w_i$,
and $a \in \names{w_{i}}$ for all $i \in \set{0,\dots,j}$.
We write $[v,a]$ for the equivalence class of $a$-connected nodes of $v$.
For $\vec{a} = a_1 \dots a_n$,
we often abuse notation and write $[v,\vec{a}]$ for the tuple 
$[v,a_1],\dots,[v,a_n]$

The \emph{decoding} of $\tree$ is the
$\sigma$-structure $\decode{\tree}$
with universe 
\[\set{ [v,a] : \text{$v \in \dom(\tree)$ and $a \in \names{v}$}}\]
such that
for each relation $R$, we have $R^{\decode{\tree}}([v_1,a_1],\dots,[v_j,a_j])$ iff
there is $w \in \dom(\tree)$ such that
$R_{\vec{a}}(w)$ holds and $[w,a_i] = [v_i,a_i]$ for all~$i$.

We are now ready to state the result that satisfiable $\gnf$ sentences 
have $k$-tree-like models. 
The original papers on $\gnf$  \cite{gnficalp,gnfj} show
that every satisfiable $\gnf$ sentence (even with equality and constants) has
a satisfying model with a \emph{tree decomposition} in which each vertex of the
tree is associated with $k$ elements of the model. We will not need the definition of
tree decomposition here, but it is easy to see (and explained in other works,
such as \cite{icalp17}) that structures with such a decomposition
have codes of the type given above. Hence we have:
\begin{proposition} \label{prop:treemod} 
\cite{gnfj} Suppose $\phi$ is a $\gnf$ sentence without equality and constants
having width $k$. If $\phi$  is satisfiable, then
it is satisfiable in a structure that is the decoding of
some  $k$-code. 
\end{proposition}

Tree codes like this
can generally have unbounded (possibly infinite) degree.
It is well-known that if a first-order sentence $\phi$ is satisfiable,
there is a structure $M$ that is countable such that $M \models \phi$ --
this follows from the Lowenheim-Skolem theorem \cite{ChangKeisler}.
Using this fact, one can refine
the proof of Proposition \ref{prop:treemod} to show that
$M$ is satisfiable in a countable model that has a $k$-tree code
where the branching degree is countable.

For technical reasons, it is more convenient to use full binary trees for our encodings.
Any tree code $T$ where each node has at most countably many children
can be converted to a binary tree code in the following way.
First, for each node $u$, we add infinitely many new children to~$u$, each child
being the root of an 
infinite full binary tree where each node has the same label as~$u$ in~$T$.
This ensures that each node of~$T$ now has infinitely many (but still countably many)
children.
Second, we convert~$T$
into a full binary tree:
starting from the root, each node $u$ with children $(v_i)_{i \in \mathbb{N}}$
is replaced by
the subtree consisting of $v_1,v_2,\dots$ and new nodes $u_1, u_2, \dots$
such that
the label at each $u_i$ is the same as the label at $u$,
the left child of $u_i$ is $v_i$ and the right child of $u_i$ is $u_{i+1}$.
In other words, instead of having a node $u$ with infinitely many children $(v_i)_{i \in \mathbb{N}}$,
we create an infinite spine of nodes with the same label as $u$,
and attach each $v_i$ to a different copy $u_i$ of $u$ on this spine.

%% file: gnftoaut/automatagnf.tex
\subsubsection{Alternating \Buchi automata}
We will consider infinite full binary trees: that is infinite trees in 
which the outdegree of every vertex is two.
We assume a set of unary predicates $A_1 \ldots A_n$ for such trees, and
let $\Sigma$ be $\{A_1 \ldots A_n\}$.

We will look at automata that can move up and down in such trees.
Let $\act_2$ be the set of (movement) \emph{directions}:
$\stay$, $\down_1$, $\down_2$, and $\up$.

For any set $J$, let $B^+(J)$ be the set of positive Boolean combinations
of propositions in $J$. Given a set $I$ of elements from $J$ and a
formula $\phi \in B^+(J)$, the notion of $\phi$ holding in $I$
($I \models \phi$) is defined as usual in propositional logic:
a single element $j \in J$ holds in $I$ if $j \in I$,
a disjunction holds in $I$ if one of its disjuncts holds, while a conjunction
holds if all of its conjuncts hold.
We will be interested in positive Boolean combinations over $\act_2 \times Q$;
these formulas will be used to describe possible moves of the automaton.

We will translate $\gnf$ sentences to  a
\emph{two-way alternating automaton over infinite 
trees}. Such an automaton  is specified as $(Q, \Sigma, q_0, \delta, \Omega)$, where
\begin{itemize}
\item $Q$ is a finite set of states
\item $\Sigma$ is as above
\item $q_0 \in Q$ is the \emph{initial state}
\item $\delta \in  Q \times \powerset{\Sigma} \rightarrow B^+(\act_2 \times Q)$ is the \emph{transition relation}
\item $\Omega$ is an acceptance condition, which we discuss below.
\end{itemize}
A \emph{run} of the automaton starting at vertex $v$  of a tree $\tree$,
  is another tree
$t'$ whose labelling function $\lambda_{t'}$ labels vertices $n$  with a 
vertex of $\tree$ and
 a state $q \in Q$. We now describe further
properties that are required for the  run to be \emph{accepting}.

First we require that the root of $t'$ is assigned to $(v,q_0)$.
That is, the computation starts at the initial state with the specified
vertex $v$.

Second, we require that the relationship between parent and children
labels in $t'$ be consistent with the transition function $\delta$.
Suppose a vertex $n'$ of $t'$ is associated by $\lambda_{t'}$ to a vertex
$n$ of $\tree$ whose predicates correspond to subset
$S_{n}$, and also  to a state $q'$, and let $C_{n'}$ be the children of $n'$ in $t'$.
Then we require that $\lambda_{t'}$ associate each $c' \in C_{n'}$  with a vertex
of $\tree$ that is either $n$, a parent of $n$, or child of $n$.

Given the above requirement  we can associate
each child $c'_i \in C_{n'}$ with a direction $d'_i \in \act_2$ as well as a state $q'_i \in Q$.
Let $P_{n'}$ be the set of pairs $(d'_i,q'_i)$ associated with some child of $n'$.
We require that $P_{n'} \models \delta(q, S)$.

Finally, we require that every branch of $t'$ obeys the acceptance condition $\Omega$.
There are a number of different acceptance  conditions defined for automata
over infinite trees. 
We will make use of
the \emph{\Buchi acceptance condition}.
This is specified by a set $\Omega \subseteq Q$ of accepting states.
The requirement is that along each branch in $t'$,
there is a state in $\Omega$ that occurs infinitely often.
Let $\btwowayaltinf$ denote 2-way alternating tree automata equipped with this \Buchi acceptance condition.

Given an automaton $\cA$ the \emph{language of $\cA$},
denoted $L(\cA)$, is
the set of trees $\tree$ such $\cA$ has an accepting run
starting at the root of $\tree$. The \emph{non-emptiness problem} for
a class of automata is the analog of the satisfiability problem
for a logic: given an automaton $\cA$ in the class, determine
if $L(\cA) \neq \emptyset$.

Vardi \cite{Vardi98} showed that non-emptiness is decidable for $\btwowayaltinf$ in $\exptime$
(in fact, this was shown for parity automata, which includes \Buchi automata).

\begin{theorem}[\cite{Vardi98}]\label{thm:vardi}
If $\cA$ is a $\btwowayaltinf$ automaton $\cA$,
then it is decidable in $\exptime$ if $L(\cA) \neq \emptyset$.
More specifically, the running time is  $f(\size{\cA})^{f(s)}$ where $s$ is the number of states of $\cA$ and $f$ is a polynomial independent of $\cA$.
\end{theorem}
Thus if we can convert our satisfiability
problem into an emptiness check for a $\btwowayaltinf$  automaton
with size doubly exponential in the size of the formula and number of states
exponential in the size of the formula, we will obtain a doubly-exponential
bound on satisfiability.
Similarly, if we can construct  a $\btwowayaltinf$  automaton
with size exponential in the formula and number of states
polynomial in the formula, we will obtain a singly-exponential
bound on satisfiability.

%% file: gnftoaut/gnnormal.tex
\subsubsection{Normal form for GNF}
We present the normal form that we use.

Formulas $\varphi$ in  \gnnf
can be generated using the following grammar:
\begin{align*}
\varphi ::= \
&{\textstyle \bigvee_i \exists \vec{x} ~ \bigwedge_j \psi_{ij}} \\
\psi ::= \,
&\alpha
\, \mid \,
\alpha \wedge \varphi 
\, \mid \,
\alpha \wedge \neg \varphi \, \mid \,
\\
&\varphi \mbox{ if } \varphi \mbox{ has at most one free variable} \, \mid \, \\
&\neg \varphi \mbox{ if } \varphi \mbox{ has at most one free variable}
\end{align*}
where $\alpha$ is an atomic formula,
and in the case of $\alpha \wedge \neg \varphi$ and $\alpha \wedge \varphi$,
$\free{\alpha} \supseteq \free{\varphi}$.
%As with $\gnf$, we also allow the guard $\alpha$ to be omitted
The $\varphi$ are referred to as \emph{UCQ-shaped formulas}, with
each of the disjuncts being a \emph{CQ-shaped formula}.

Note that if $\phi_i$ for $i=1 \ldots n$ are \emph{sentences} in normal form then
their conjunction $\bigwedge_i \phi_i$ is also in normal form.

A formula is  \emph{answer-guarded} if it has at most one free variable or is of the form
$\alpha \wedge \chi$ where $\alpha$ is an atom that contains all the free variables of $\chi$.
The idea of the normal form is that the grammar generates UCQ-shaped formulas
where each conjunct is an answer-guarded subformula. 

Later we will see that we can convert arbitrary $\gnf$ formulas to this normal form.

The \emph{width} of a  $\gnf$ formula $\phi$  in the normal form above
is the maximum number of free variables in any subformula.

The \emph{CQ-rank} of a formula $\phi$ in $\gnnf$,  
denoted $\rankcq{\phi}$, is the maximum number of conjuncts $\psi_i$
in any CQ-shaped subformula
$\exists \vec{x} ~ \bigwedge_i \psi_i$ of $\phi$
where $\vec{x}$ is non-empty.
Recall that the $\psi_i$ in such a CQ-shaped formula are of the form 
$\alpha$, $\alpha \wedge \neg \phi''$, or $\alpha \wedge \varphi''$,
but
for the purposes of counting conjuncts for the CQ-rank,
each $\psi_i$
is treated as a single conjunct.

%Note that  if the maximal arity of the schema is $a$, then for any
%$\phi$, $\width{\phi} \leq (\rankcq{\phi} \cdot a)+1$. Thus fixing
%the maximal arity and the CQ-rank implies fixing the width.

%(see Proposition~\ref{prop:normalformsize}).

%% file: gnftoaut/gnfautdecidespecializations.tex
\subsubsection{Automata for $\gnf$}

We now explain how to construct an automaton for a $\gnf$ sentence $\phi$ in normal form
without equality.

\myparagraph{Specializations}
The rough idea will be that an automaton has 
 states for all subformulas of $\phi$ -- the ``subformula closure of $\phi$''.
The automaton being in a state corresponding to 
subformula $\psi$  at a vertex $v$ of a tree $\tree$
will indicate that it is verifying that  $\psi$ holds at $v$ in $\tree$.
The statement above   is not precise because
in  $\gnf$, the notion of ``subformula'' needs to be more expansive
than the usual one  in order to be able to correctly verify
the CQ-shaped formulas.

Before we define the relevant closure, we need to think more carefully about CQ-shaped formulas,
and how they can be satisfied in a tree-like structure.
For this, we need to describe \emph{specializations}.

Consider a CQ-shaped \nf $\gnf$ formula
\[
\rho(\vec{x}) = \exists\vec{y} ~ \bigwedge_{j \in \set{1,\dots,r}} \psi_j(\vec{x},\vec{y}) .
\]
A \emph{specialization} of $\rho$
is a formula $\rho'$ obtained from $\rho$ by the following operations:
\begin{itemize}
\item select a subset $\vec{y}_0$ of $\vec{y}$
(call variables from $\vec{x} \cup \vec{y}_0$ the \emph{inside variables} and
variables from $\vec{y} \setminus \vec{y}_0$ the \emph{outside variables});
\item select a partition $\vec{y}_1,\dots,\vec{y}_s$ of the outside variables,
with the property that for every $\psi_j$, either $\psi_j$ has no
outside variables or all of its outside variables are
contained in the partition element $\vec{y}_j$;
\item let $\chi_0$ be the conjunction of the $\psi_j$ using only inside variables,
and let $\chi_i$ for $i \in \set{1,\dots,s}$
be the  conjunction of the $\psi_j$ using outside variables
and satisfying $\free{\psi_j} \subseteq \vec{x} \cup \vec{y}_0 \cup \vec{y}_i$;
\item set $\rho'(\vec{x},\vec{y}_0)$ to be
\[\chi_0(\vec{x},\vec{y}_0) \wedge
\bigwedge_{i \in \set{1,\dots,s}} \exists \vec{y}_i ~ \chi_i(\vec{x},\vec{y}_0,\vec{y}_i) .\]
\end{itemize}

Roughly speaking, each specialization of $\rho$ describes
a different way that a CQ-shaped formula could be satisfied
by elements $\vec x$ represented in
a  node of a tree code. The inside variables represent
witnesses for the existential quantifiers
 that are found in the node itself. The partition of the outside variables
represent the different  directions 
from the node where the additional non-local witnesses are
to be found: moving either to an ancestor or to one of the children.
Since each atom of the CQ shape formula must be realized in a single
node, the atoms must be ``homogeneous'' with respect to the partition,
as captured in the second item above.

It is easy to see that if a specialization is realized, then
so is the original formula, since the realization of the specialization
gives witnesses for all the existential quantifiers:

\begin{lemma}\label{lemma:specialization-implication}
Let $\rho(\vec{x}) \in \gnf$ be a CQ-shaped formula
$\exists \vec{y} ~ \bigwedge_j \psi_j(\vec{x},\vec{y})$.
For all structures $M$ and for all specializations $\rho'(\vec{x},\vec{y}_0)$ of $\rho$,
if $M \models \rho'(\vec{a},\vec{b})$,
then $M \models \rho(\vec{a})$.
\end{lemma}

Since a formula is vacuously a specialization of itself, the converse
 direction is vacuously true.
What is more useful is that whenever a formula is realized, it is
realized by a specialization that is ``simpler'' than the original
formula it specializes.
We say a specialization
is \emph{non-trivial} if either
$\chi_0$ is non-empty or
the partition of the outside
variables is non-trivial ($s > 1$).
The following result captures the idea that in realizing a formula
we need to realize some simpler specialization:

\begin{lemma}\label{lemma:non-trivial-specialization}
Let $\rho(\vec{x}) \in \gnf$ be a CQ-shaped formula
$\exists \vec{y} ~ \bigwedge_j \psi_j(\vec{x},\vec{y})$.
Given a structure $M$ and its tree code $\tree$,
if there exists a vertex $v \in \tree$ that includes names $\vec{a}$ and
$M \models \rho([v,\vec{a}])$, 
then there is a non-trivial specialization $\rho'(\vec{x},\vec{y}_0)$ of $\rho$
and a vertex $w \in \tree$ with $\vec{a}$ and additional names $\vec{b}_0$ in its domain
such that $[w,\vec{a}] = [v,\vec{a}]$ and $M \models \rho'([w,\vec{a}],[w,\vec{b}_0])$.
\end{lemma}

The idea behind the lemma is that if the formula
holds at a node with certain witnesses for the free variables, 
we can traverse the nodes of the tree codes preserving
all those witnesses, 
until we arrive at a node $w$ where either  some of the witnesses
are found locally in $w$
or the witnesses are found in different directions from $w$.
In the   first case we have realized a specialization  in which $\chi_0$
is non-empty, and in the second case we have realized a specialization
in which the partition of the outside variables is non-trivial.

Let $\exists \vec{y} ~ \eta(\vec{x}, \vec{y})$ be any CQ-shaped $\gnnf$ formula
and $\exists \vec{y} ~ \eta(\vec{a},\vec{y})$ be formed by substituting
elements $a_i$ from $\namesk$ for each free variable $x_i$
 in $\exists \vec{y} ~ \eta(\vec{x}, \vec{y})$.
We will write $\specializations{\exists \vec{y} ~ \eta(\vec{a},\vec{y}),\namesk}$ 
for the set of all specializations of $\exists \vec{y} ~ \eta(\vec{a},\vec{y})$
with elements from $\namesk$ substituted for any new inside variables.
For convenience in the construction below,
each specialization $S \in \specializations{\exists \vec{y} ~ \bigwedge_j \psi_j(\vec{a},\vec{y}),\namesk}$ will be represented as a set:
that is, the specialization $\chi_0(\vec{a},\vec{b}_0) \wedge
\bigwedge_{i \in \set{1,\dots,s}} \exists \vec{y}_i ~ \chi_i(\vec{a},\vec{b}_0,\vec{y}_i)$
of
$\exists\vec{y} ~ \bigwedge_{j \in \set{1,\dots,r}} \psi_j(\vec{a},\vec{y})$
is represented as the set:
\[
\sset{ \psi_j(\vec{a},\vec{b}_0) : j \in \set{1,\dots,r}, \psi_j(\vec{a},\vec{b}_0) \text{ in } \chi_0 }
\cup
\sset{ \exists \vec{y}_i ~ \chi_i(\vec{a},\vec{b}_0,\vec{y}_i) : i \in \set{1,\dots,s} } .
\]
Again,
each formula in the set describes how a piece of the CQ-shaped formula is satisfied.

We are now ready to define the notion of subformula we are interested
in. 
Fix some $\gnf$ sentence $\phi$ in $\nf$.
The closure $\clgnf{\phi,\namesk}$ that is relevant for the automaton construction to decide satisfiability of $\phi$
consists of the subformulas of $\phi$ along with formulas that are part of the specializations of the CQ-shaped formulas.
Formally, elements of $\clgnf{\phi,\namesk}$ will be
written in the form $\state{\psi,p}$ where $\psi$ is a formula and $p \in \set{+,-}$ is a polarity
to indicate whether $\psi$ comes from a positive or negative part of $\phi$
(that is, a part under an even or odd number of negations).

Let $\clgnf{\phi,\namesk}$ be the smallest set  of formulas containing
$\state{\phi,+}$, $\state{\true,+}$, $\state{\true,-}$, $\state{\false,+}$, $\state{\false,-}$
and satisfying the following closure conditions:
\begin{itemize}
\item if $\state{\alpha \wedge \neg \psi,+} \in \clgnf{\phi,\namesk}$, then
$\state{\alpha,+}, \state{\psi,-} \in \clgnf{\phi,\namesk}$;
\item if $\state{\alpha \wedge \neg \psi,-} \in \clgnf{\phi,\namesk}$, then
$\state{\alpha,-}, \state{\psi,+} \in \clgnf{\phi,\namesk}$;
\item if $\state{\neg \psi,+} \in \clgnf{\phi,\namesk}$, then
$\state{\psi,-} \in \clgnf{\phi,\namesk}$;
\item if $\state{\neg \psi,-} \in \clgnf{\phi,\namesk}$, then
$\state{\psi,+} \in \clgnf{\phi,\namesk}$;
\item if $\state{\alpha \wedge \psi,+} \in \clgnf{\phi,\namesk}$, then
$\state{\alpha,+}, \state{\psi,+} \in \clgnf{\phi,\namesk}$;
\item if $\state{\alpha \wedge \psi,-} \in \clgnf{\phi,\namesk}$, then
$\state{\alpha,-}, \state{\psi,-} \in \clgnf{\phi,\namesk}$;
\item if $\state{\bigvee_i \psi_i, +} \in \clgnf{\phi,\namesk}$, then
$\state{\psi_i,+} \in \clgnf{\phi,\namesk}$ for all $i$;
\item if $\state{\bigvee_i \psi_i, -} \in \clgnf{\phi,\namesk}$, then
$\state{\psi_i,-} \in \clgnf{\phi,\namesk}$ for all $i$;
\item if $\state{\exists \vec{y} ~ \eta(\vec{a},\vec{y}),+} \in \clgnf{\phi,\namesk}$, then
$\state{\psi',+} \in \clgnf{\phi,\namesk}$ for all $S \in \specializations{\exists \vec{y} ~ \eta(\vec{a},\vec{y}),\namesk}$ and $\psi' \in S$;
\item if $\state{\exists \vec{y} ~ \eta(\vec{a},\vec{y}),-} \in \clgnf{\phi,\namesk}$, then
$\state{\psi',-} \in \clgnf{\phi,\namesk}$ for all $S \in \specializations{\exists \vec{y} ~ \eta(\vec{a},\vec{y}),\namesk}$ and $\psi' \in S$.
\end{itemize}

We are now ready to give a translation of $\gnf$ sentences
into automata, and show that size is controlled by the size of the subformula closure.

\begin{proposition}\label{prop:auttrans}
For every $\gnf$ sentence $\phi$ in \gnnf,
signature $\sigma$ containing relations of $\phi$,
and $k \in \N$,
there is a $\btwowayaltinf$ $\cA_{\phi}$ on $\sigk$-trees
such that $\cA_{\phi}$ accepts a consistent $\sigk$-tree $\tree$ iff
the decoding $\decode{\tree}$ satisfies $\phi$.
Moreover, the number of states of the automaton
is bounded by the size of  $\clgnf{\phi,\namesk}$, while
the overall size and the time needed to construct the automaton is
at most $f(\size{\phi} \cdot \size{\powerset{\sigk}}) \cdot \size{\namesk}^{f(\width{\phi} \rankcq{\phi})}$
for some polynomial $f$ independent of $\phi$ and $k$.
\end{proposition}

The $\btwowayaltinf$ automaton $\autA_\phi$ for $\phi$ is defined as follows:
\begin{itemize}
\item The state set is $\clgnf{\phi,\namesk}$.
\item The initial state is $\state{\phi,+}$.
\item The transition function $\delta$ is defined below.
\item The set of accepting states
consists of all states of the form
$\state{\True,+}$,
$\state{\False,-}$,
$\state{R(\vec{a}),-}$, or
$\state{\exists \vec{y} ~ \eta(\vec{a},\vec{y}),-}$.
\end{itemize}

We now describe the transition function.
For $\tau$ a collection of symbols in $\sigk$ and $\vec a$ a collection
of names in $\paramsk$, we say that $\vec a$ is \emph{represented in $\tau$}
if $\tau$ includes $D_{a_i}$ for each $a_i$ in $\vec a$;
thus a vertex $v$ labelled with $\tau$ that represents $\vec a$ has each $a_i$ in $\vec a$ as one of
its local names.

\begin{align*}
\delta(\state{R(\vec{a}),+},\tau) &:=
	\begin{cases}
		(\stay,\state{\false,+}) &\text{if $\vec{a}$ not represented in $\tau$} \\		
		(\stay,\state{\true,+}) &\text{if $R_{\vec{a}} \in \tau$} \\
		\bigvee_{d \in \act_2} (d,\state{R(\vec{a}),+}) &\text{otherwise}
	\end{cases}
\\
\delta(\state{R(\vec{a}),-},\tau) &:=
	\begin{cases}
		(\stay,\state{\true,+}) &\text{if $\vec{a}$ not represented in $\tau$} \\		
		(\stay,\state{\false,+}) &\text{if $R_{\vec{a}} \in \tau$} \\
		\bigwedge_{d \in \act_2} (d,\state{R(\vec{a}),-}) &\text{otherwise}
	\end{cases}
\\
\myeat{
\delta(\state{a=b,+},\tau) &:=
	\begin{cases}
		(\stay,\state{\true,+}) &\text{if $a = b \in \tau$} \\		
		(\stay,\state{\false,+}) &\text{if $a = b \notin \tau$}
	\end{cases}
\\
\delta(\state{a=b,-},\tau) &:=
	\begin{cases}
		(\stay,\state{\false,+}) &\text{if $a = b \in \tau$} \\		
		(\stay,\state{\true,+}) &\text{if $a = b \notin \tau$}
	\end{cases}
\\
}
\delta(\state{\True,+},\tau) &:=  (\stay,\state{\True,+}) \\
\delta(\state{\False,-},\tau) &:=  (\stay,\state{\False,-}) \\
\delta(\state{\True,-},\tau) &:=  (\stay,\state{\True,-}) \\
\delta(\state{\False,+},\tau) &:=  (\stay,\state{\False,+}) \\
\delta(\state{\bigvee_{i} \psi_i,+},\tau) &:= \textstyle \bigvee_i (\stay, \state{\psi_i,+}) \\
\delta(\state{\bigvee_{i} \psi_i,-},\tau) &:= \textstyle \bigwedge_i (\stay, \state{\psi_i,-}) \\
\delta(\state{\alpha \wedge \neg \psi,+},\tau) &:= (\stay,\state{\alpha,+}) \wedge (\stay,\state{\psi,-}) \\
\delta(\state{\alpha \wedge \neg \psi,-},\tau) &:= (\stay,\state{\alpha,-}) \vee (\stay,\state{\psi,+})
\\
\delta(\state{\neg \psi,+},\tau) &:= (\stay,\state{\psi,-}) \\
\delta(\state{\neg \psi,-},\tau) &:=  (\stay,\state{\psi,+})\\
\delta(\state{\alpha \wedge  \psi,+},\tau) &:= (\stay,\state{\alpha,+}) \wedge (\stay,\state{\psi,+}) \\
\delta(\state{\alpha \wedge \psi,-},\tau) &:= (\stay,\state{\alpha,-}) \vee (\stay,\state{\psi,-})
\\
\delta(\state{\exists\vec{y} ~ \eta(\vec{a},\vec{y}), +},\tau) &:=
	\begin{cases}
		(\stay,\state{\false,+}) \quad \text{if $\vec{a}$ not represented in $\tau$ } \\
		\bigvee_{S \in \specializations{\exists \vec{y} ~ \eta(\vec{a},\vec{y}),\bagnames{\tau}}} \bigwedge_{\psi \in S} (\stay, \state{\psi,+}) \ \vee \\
		\quad \bigvee_{d \in \act_2} (d,\state{\exists\vec{y} ~ \eta(\vec{a},\vec{y}), +}) \quad \text{otherwise}
	\end{cases}
\\
\delta(\state{\exists\vec{y} ~ \eta(\vec{a},\vec{y}), -},\tau) &:=
	\begin{cases}
		(\stay,\state{\true,+}) \quad \text{if $\vec{a}$ not represented in $\tau$} \\
		\bigwedge_{S \in \specializations{\exists \vec{y} ~ \eta(\vec{a},\vec{y}),\bagnames{\tau}}} \bigvee_{\psi \in S} (\stay, \state{\psi,-}) \ \wedge \\
		\quad \bigwedge_{d \in \act_2} (d,\state{\exists\vec{y} ~ \eta(\vec{a},\vec{y}), -}) \quad \text{otherwise}
	\end{cases}
\end{align*}

The correctness of the automaton construction is captured in the following
result:
\begin{lemma} \label{lem:autcorrect}
For each $\state{\psi(\vec{a}),+} \in \clgnf{\phi,\namesk}$,
$\psi(\vec x)$ holds in $\decode{\tree}$ with valuation
$[v,\vec{a}]$ for $\vec{x}$
if and only if the automaton above accepts 
when launched in $\tree$ from vertex $v$ with initial state $\state{\psi(\vec a), +}$.

Likewise, for each $\state{\psi(\vec{a}),-} \in \clgnf{\phi,\namesk}$,
$\psi(\vec x)$ does not hold in $\decode{\tree}$ with valuation
$[v,\vec{a}]$ for $\vec{x}$
if and only if the automaton above accepts 
when launched in $\tree$ from vertex $v$ with initial state $\state{\psi(\vec a), -}$.
\end{lemma}

\begin{proof}
The lemma is proven by structural induction.
The base cases are simple to observe by construction.
Lemmas~\ref{lemma:specialization-implication}~and~\ref{lemma:non-trivial-specialization} are utilized
in the inductive case for CQ-shaped formulas.
\end{proof}

We now calculate the size of $\clgnf{\phi,\namesk}$.

\begin{lemma}\label{lemma:closuresize}
Let $\phi \in \gnf$ in normal form,
and let $k \in \N$.
Then $\size{\clgnf{\phi,\namesk}} \leq f(\size{\phi}) \cdot \size{\namesk}^{f(\width{\phi}\rankcq{\phi})}$
for some polynomial function $f$ independent of $\phi$ and $k$.
\end{lemma}

\begin{proof}
Let $w = \width{\phi}$
and $r = \rankcq{\phi}$.

Note that in the definition of the closure set,
the only formulas that appear are either actual subformulas of $\phi$ (with names from $\namesk$ substituted for free variables),
or are formulas that come from specializations of CQ-shaped formulas (again, with names from $\namesk$).

Specializations of CQ-shaped subformulas that do not begin with existential quantification
(i.e.~a CQ-shaped formula without projection) only contribute actual subformulas of $\phi$ to the closure set.
However, the specializations of a CQ-shaped subformula $\eta$ with existential quantification
contribute up to $2^r$ additional CQ-shaped formulas that are based on taking some subset of the (at most) $r$ conjuncts of~$\eta$.

Since each of these formulas has at most $w$ free variables taking names from~$\namesk$,
this means that the overall size of the closure set is at most $\size{\phi} \cdot 2^r \cdot \size{\namesk}^w$.
\end{proof}

Let $w = \width{\phi}$ and $r = \rankcq{\phi}$.
Since the width and CQ-rank are bound by the size of the formula,
this means that the size of the closure set $\clgnf{\phi,\namesk}$,
and hence the number of states of the automaton $\cA_\phi$,
is at most exponential in the size of the formula.
But it is polynomial when
the maximal arity, width, and CQ-rank are fixed.

The size of $\powerset{\sigk}$ is at most $2^{\size{\sigma} \cdot \size{\namesk}^{\arity{\sigma}}}$,
which is doubly exponential
in general, but singly exponential when the maximal arity is fixed.

The size of each transition function formula is at most linear in
$2^w \cdot \size{\namesk}^w \cdot w^w \cdot \size{\clgnf{\phi,\namesk}}$.
In particular, note that the transition function formula for a CQ-shaped formula $\psi$
respects this bound since $\size{\specializations{\psi,\namesk}}$ is at most $2^w \cdot \size{\namesk}^w \cdot w^w$
(the maximum number of ways to choose the inside variables,
names for these inside variables, and the partition of the outside variables),
and each $S \in \specializations{\psi,\namesk}$ is of size at most $\size{\clgnf{\phi,\namesk}}$.

This means that the size of the transition function is linear in
$\size{Q} \cdot \size{\powerset{\sigk}} \cdot 2^w \cdot \size{\namesk}^w \cdot w^w \cdot \size{Q}$.
This is doubly exponential in general, but singly exponential when the maximal arity,
width, and CQ-rank are fixed.

Therefore, the overall size of $\cA_\phi$ and the time taken to construct it
is of size at most doubly exponential in the size of $\phi$,
but singly exponential when the maximal arity, width, and CQ-rank is fixed.

\myparagraph{From an automaton to decidability}
We are now almost done with our satisfiability
procedure. Combining Proposition \ref{prop:auttrans}
with  Proposition \ref{prop:treemod}, we see that $\phi$ is satisfiable
if and only if there is a consistent $k$-tree code that
satisfies $\cA_\phi$, where $k = \width{\phi}$.

Recall that a consistent $\sigk$-tree is just an arbitrary
 $\sigk$-tree such that every node $v$ satisfies
$\size{\names{v}} \leq k$ and
 for all $R_{\vec{a}} \in \sigk$, if $R_{\vec{a}}(v)$ then $\vec{a} \subseteq \names{v}$.
It is straightforward to see that there is a $\btwowayaltinf$ automaton
$\cA_{\consistent}$ that accepts exactly the trees that are
consistent in the above sense. The size
of  $\cA_{\consistent}$ is doubly-exponential (due
to the size of the alphabet) and singly-exponential
if the maximal arity of each relation is fixed. The
running time needed to form the automaton is likewise
doubly-exponential in general and singly-exponential when
the arity of relations is fixed.
Further the number of states is just two --- an initial  state  and
a ``rejection'' state representing a violation of consistency.

By the closure properties of 
$\btwowayaltinf$, we know that we can form an automaton $\cA_{\phi, \consistent}$
that accepts the intersection $L(\cA_\phi) \cap L(\cA_{\consistent})$ 
in time proportional to the sum
of the sizes of
$\cA_\phi$ and $\cA_{\consistent}$. The number of states of this automata
is just the sum of the number states of $\cA_\phi$ and $\cA_{\consistent}$.
Hence, by applying Theorem~\ref{thm:vardi},
we can conclude:

\begin{theorem} \label{thm:normalnoeq} There is a $\twoexp$ satisfiability testing algorithm for $\gnf$
sentences in normal form without equality and constants. When the width, CQ-rank and maximal arity
of the relations
are fixed, it shrinks to $\exptime$.
\end{theorem}

%% file: gnftoaut/equality.tex
The extension to handle equalities, in the absence of constants, is not difficult.
We consider the same tree codes as before.

We claim again that if a sentence $\phi$ in $\gnf$ with width bounded by $k$ is satisfiable,  then it is satisfied in
a structure with a $k$-code.

The conversion to normal form is the same, treating equality like
any other relation.

In the automaton construction, we need additional cases for equality.
\begin{align*}
\delta(\state{a=b,+},\tau) &:=
        \begin{cases}
                (\stay,\state{\true,+}) &\text{if $a$ is the same as $b$} \\
                (\stay,\state{\false,+}) &\text{if $a$ is not the same as $b$}
        \end{cases}
\\
\delta(\state{a=b,-},\tau) &:=
        \begin{cases}
                (\stay,\state{\false,+}) &\text{if $a$ is the same as $b$} \\
                (\stay,\state{\true,+}) &\text{if $a$ is not the same as $b$}
        \end{cases}
\end{align*}

The size bounds and running time of the construction remain the same.

\myparagraph{Constants} To handle constants requires more effort, since constants may have non-trivial
equalities.  One route to decidability, taken in \cite{gnfj}, is to
reduce satisfiability of $\gnf$ with equality and constants to satisfiability without
constants. The idea of the reduction in \cite{gnfj} is to extend the signature
with additional predicates that hold the constants.
However, using such a  reduction as a black-box does
 not give us the fine-grained bounds we desire in terms
of parameters like CQ-rank. We thus provide a more direct argument.

We consider $k$-tree codes in which the constants $\const{\sigma}$ are represented in each node,
along with at most $k$ local names.
The codes will also now include some equality facts, but with the following restrictions:
\begin{itemize}
\item There are no equality facts relating non-constants to each
other, and no equality facts relating constants to non-constants.
\item The equality facts on constants are identical  across vertices of the tree.
They  satisfy transitivity and reflexivity, 
as well as \emph{congruence}:
if we have a fact $R(\ldots c \ldots)$ holding in a vertex, where $c$ is a constant,
and we also have an equality fact $c=d$ then we have the fact
 $R(\ldots d \ldots)$.
\end{itemize}

We can extend $\cA_{\consistent}$ to check whether a tree is a code satisfying these additional restrictions.

We must change the notion of decoding of a tree to account
for equalities. For a consistent tree $\tree$ using
local names and constants $\const{\sigma}$, we let $\const{\sigma}_{=,\tree}$ be the equivalence
classes of constants under the equality  relation in $\tree$.
The decoding $\decode{\tree}$ is now the
$\sigma$-structure 
with universe
\[\set{ [v,a] : \text{$v \in \dom(\tree)$ and $a \in \names{v}$}} \cup \const{\sigma}_{=\tree}\]
such that
for each relation $R$, we have $R^{\decode{\tree}}([v_1,a_1],\dots,[v_j,a_j], e_1 \ldots e_l)$, where $a_i$ are local names and $e_i$ are
equivalence classes of constants, iff
there is $w \in \dom(\tree)$ such that
$R_{\vec{a}, c_1 \ldots c_l}(w)$ holds,
 $[w,a_i] = [v_i,a_i]$ for all~$i \leq j$ and
$c_i$ is in class $e_i$ for each $i \leq l$.

We further claim the following extension of Proposition \ref{prop:treemod}
\begin{proposition} \label{prop:treemodequality}
If a $\gnf$ sentence $\phi$ of width $k$, possibly using equality and constants, is satisfiable, then
$\phi$ is satisfiable in a structure that is the decoding of
some  $k$-code.
\end{proposition}
\begin{proof}
Consider an expanded signature where for each relation $R$ of arity $n$
and partial function $h$ from the positions of $R$ into constants, we have  have a relation $R_h$ of arity $n-|dom(h)|$. We can rewrite $\phi$ to a $\phi'$ in this signature that does not
contain constants, replacing atoms $R(x_1 \ldots x_n)$ by a disjunction of atoms over $R_h$ where $h$ varies 
over every partial function, and replacing subformulas with  negation
guarded by an  $R$-atom with a disjunction of subformulas guarded by an $R_h$-atom.
Note that $\phi'$ will be larger than $\phi$, but its width will still be
$k$.
Thus applying Proposition \ref{prop:treemod} we see that $\phi$ has a model $M'$
with a $k$-code in the expanded signature. But then we can reverse this
process on $M$, replacing atoms $R_h$ in $M$ with an atom $R$ but using
the additional constants as arguments. We can similarly add the equality
facts to the codes. Since equality in $M'$ must satisfy congruence, reflexivity,
and transitivity, we will obtain a structure satisfying the additional properties.
%then it is satisfiable in a structure of width $k+|C|$, where $C$ are the constants
%mentioned in $\phi$. 
\end{proof}

\myeat{
 as before,  that if a sentence is satisfied, then it is satisfied
by a structure with a tree code of this sort. The fact that For the first item,
note that if we have a fact involving a constant, then there is no need
to have this represented in the tree by a non-constant plus an 
equality fact.
}

The closure is now defined as before, but based on $\namesk \cup \const{\sigma}$ rather than~$\namesk$.

In the automaton construction, we need a few modifications:

We need a base case for equality atoms.
\begin{itemize}
\item 
For a non-negated equality of a local name with a constant, the automaton 
should ensure  rejection:
it does this by
switching to state $\state{\false,+}$, since there are no accepting runs
from such states.
Similarly for a negated equality of a local name with a constant,
the automaton should ensure acceptance by switching to state $\state{\true,+}$.
\item for an equality between constants, the automaton simply checks whether the equality is present in the vertex; if this is true
the automaton should ensure acceptance.
It does this by switching to state $\state{\true,+}$. Otherwise it 
ensures rejection by switching to state $\state{\false,+}$.
\end{itemize}
That is, for a name $a \in \namesk$ and for constants $c,d \in \const{\sigma}$, we have transitions:
\begin{align*}
\delta(\state{a=c,+},\tau) &:= (\stay,\state{\false,+}) 
\\
\delta(\state{a=c,-},\tau) &:= (\stay,\state{\true,+}) 
\\
\delta(\state{c=d,+},\tau) &:=
        \begin{cases}
                (\stay,\state{\true,+}) &\text{if $c = d \in \tau$} \\
                (\stay,\state{\false,+}) &\text{if $c = d \notin \tau$}
        \end{cases}
\\
\delta(\state{c=d,-},\tau) &:=
		\begin{cases}
                (\stay,\state{\false,+}) &\text{if $c = d \in \tau$} \\
                (\stay,\state{\true,+}) &\text{if $c = d \notin \tau$}
        \end{cases}
\end{align*}

We also modify the CQ-shaped formula case, to allow the automaton to draw witnesses from the constants:

\begin{align*}
\delta(\state{\exists\vec{y} ~ \eta(\vec{a},\vec{y}), +},\tau) &:=
	\begin{cases}
		(\stay,\state{\false,+}) \quad \text{if $\vec{a}$ not represented in $\tau$ } \\
		\bigvee_{S \in \specializations{\exists \vec{y} ~ \eta(\vec{a},\vec{y}),\bagnames{\tau} \cup \const{\sigma}}} \bigwedge_{\psi \in S} (\stay, \state{\psi,+}) \ \vee \\
		\quad \bigvee_{d \in \act_2} (d,\state{\exists\vec{y} ~ \eta(\vec{a},\vec{y}), +}) \quad \text{otherwise}
	\end{cases}
\\
\delta(\state{\exists\vec{y} ~ \eta(\vec{a},\vec{y}), -},\tau) &:=
	\begin{cases}
		(\stay,\state{\true,+}) \quad \text{if $\vec{a}$ not represented in $\tau$} \\
		\bigwedge_{S \in \specializations{\exists \vec{y} ~ \eta(\vec{a},\vec{y}),\bagnames{\tau} \cup \const{\sigma}}} \bigvee_{\psi \in S} (\stay, \state{\psi,-}) \ \wedge \\
		\quad \bigwedge_{d \in \act_2} (d,\state{\exists\vec{y} ~ \eta(\vec{a},\vec{y}), -}) \quad \text{otherwise}
	\end{cases}
\end{align*}

Using these modifications, we can now extend
Lemma \ref{lem:autcorrect}:

\begin{lemma} \label{aut:correctequality}
For each $\state{\psi(\vec{a}, \vec{c}),+} \in \clgnf{\phi,\namesk}$,
$\psi(\vec x, \vec y)$ holds in $\decode{\tree}$ at vertex $v$ with valuation
$[v,\vec{a}]$ for $\vec{x}$ and constants $c_1 \ldots c_l$ for
$\vec y$
if and only if the automaton above accepts
when launched in $\tree$ from vertex $v$ with initial state 
$\state{\psi(\vec a, \vec c), +}$.

Likewise, for each $\state{\psi(\vec{a}, \vec{c}),-} \in \clgnf{\phi,\namesk}$,
$\psi(\vec x,\vec y)$ does not hold in $\decode{\tree}$ at vertex
$v$ with the valuation above
if and only if the automaton above accepts
when launched in $\tree$ from vertex $v$ with initial state $\state{\psi(\vec a, \vec c), -}$.
\end{lemma}

Recall that the proof of Lemma \ref{lem:autcorrect} worked
by induction on $\psi$.
In the proof we first need to consider  base cases for equality.
For example, suppose  $x_1=x_2$ holds in $\decode{\tree}$ with valuation
$x_1=[v,a_1]$ $x_2=[v,a_2]$ for local names $a_1,a_2$. The only way
the equality can hold is if $a_1$ is actually the same
name as $a_2$. Thus the automaton run from $\state{a=b,+}$
will transition to $\state{\true,+}$, and will accept.
The converse direction is similar.

On the other hand,   suppose  $x_1=x_2$ holds in $\decode{\tree}$ 
with valuation
$x_1=[c]_{=,\tree}$, $x_2=[d]_{=, \tree}$ for constants $c,d$.
This holds exactly when the equality fact $c=d$
is present in the label of $v$.
But then looking at the transition function for
$\state{c=d,+}$ we see that the automaton accepts.

We must also reconsider the base cases for atomic relations.
Suppose  $R(x_1, \ldots x_j, y_1 \ldots y_l )$ holds in $\decode{\tree}$
with valuation $x_i=[v,a_i], y_i=[c_i]_{=,\tree}\}]$ for 
local names $\vec a$ and constants $\vec c$.
By definition of our decoding, along with the congruence closure
of the codes, this means that
we must have a fact $R([v,\vec a], \vec c)$ holding
in some node $v'$ in the tree. We now argue
as in the case without constants that  iterating the transition
function for an atom, the automaton will accept from $v$.

\begin{proposition} $\phi$ is satisfiable if and only if
the modified automaton $\cA_\phi$ accepts   a consistent tree.
This in turn can be checked by taking the automaton $\cA_{\consistent}$ for checking
consistency, forming an automaton $\cA'_\phi$ accepting the  intersection
of $\cA_{\consistent}$ with $\cA_\phi$, and checking non-emptiness of $\cA'_\phi$.
\end{proposition}

Thus we obtain the main result, which immediately implies
Theorem \ref{thm:gnfsatrefined} 
 n the body of the paper:

\begin{theorem} \label{thm:finale} There is a $\twoexp$ algorithm
for deciding satisfiability of  sentences  in $\gnf$, even allowing
equality and constants. For a sentence in normal form with
fixed width, CQ-rank and fixed arity of relations, we get
an $\exptime$ algorithm for satisfiability.
\end{theorem}

%% file: gnftoaut/convertgnnorm.tex
%\section{Conversion to normal form}
\newcommand{\convertnf}[1]{\kw{convert}(#1)}
\newcommand{\namesm}{N_m}

We note that the previous result to allows us to re-prove
the bounds for  satisfiability of $\gnf$ sentences
that are not in normal form from \cite{gnficalp,gnfj}. We include this only
because it might be useful to have a self-contained presentation of the $\gnf$ to automate
translation.
The idea is that general $\gnf$ sentences can be converted to
normal form in such a way that we blow up the size of the formula, but the size of the closure set remains
at most exponential in the size of the original formula.

\begin{proposition}\label{prop:normalformsize}
Let $\psi$ be a $\gnf$ formula
with $m = \size{\psi}$.
We can construct  a sentence $\convertnf{\psi}$ in \gnnf
equivalent to $\psi$ such that 
\begin{compactitem}
\item $\size{\convertnf{\psi}} \leq 2^{f(m)}$ ,
\item $\width{\convertnf{\psi}} \leq m$,
\item $\rankcq{\convertnf{\psi}} \leq m$,
\item $\size{\clgnf{\convertnf{\psi},\namesm}} \leq 2^{f(m)}$.
\end{compactitem}
where $f$ is a polynomial function independent of $\psi$.
\end{proposition}

\begin{proof}
We proceed by induction on $\psi$.
The output $\convertnf{\psi}$ is a UCQ-shaped formula in $\nf$, with the same free variables as $\psi$.
\begin{itemize}
\item If $\psi$ is atomic, then $\convertnf{\psi} := \psi$.
\item Suppose $\psi = \alpha \wedge \neg \psi'$ where $\alpha$ is a guard for $\free{\psi'}$.
Then $\convertnf{\psi} := \alpha \wedge \neg \convertnf{\psi'}$.

Similarly for the case of $\neg \psi$ where $\psi$ has at most
one free variable.

\item Suppose $\psi = \exists y ~ \psi'$.
If $\convertnf{\psi'}$ is a UCQ-shaped formula $\bigvee_{i} \exists \vec{z}_i ~ \bigwedge_j \psi_{ij}$,
then $\convertnf{\psi} := \bigvee_i \exists y \vec{z}_i ~ \bigwedge_j \psi_{ij}$.

\item Suppose $\psi = \psi_1 \vee \psi_2$.
Then $\convertnf{\psi}$ is the UCQ-shaped formula $\convertnf{\psi_1} \vee \convertnf{\psi_2}$.

\item Suppose $\psi = \psi_1 \wedge \psi_2$.
If $\psi_1$ and $\psi_2$ are answer-guarded, e.g., $\psi_1 = \alpha_1 \wedge \psi'_1$ and $\psi_2 = \alpha_2 \wedge \psi'_2$
with $\free{\alpha_1} \supseteq \free{\psi'_1}$ and $\free{\alpha_2} \supseteq \free{\psi'_2}$,
then $\convertnf{\psi} = (\alpha_1 \wedge \convertnf{\psi'_1}) \wedge (\alpha_2 \wedge \convertnf{\psi'_2})$.
The other cases where $\psi_1$ or $\psi_2$ have at most one free variable are handled
similarly.

Otherwise,
let
$\convertnf{\psi} := \bigvee_{i,i'} \exists \vec{y}_i \vec{y}'_{i'} ~ (\chi_i[\vec{y}_i / \vec{x}_i] \wedge \chi'_{i'}[\vec{y}'_{i'} / \vec{x}'_{i'}])$
where
$\convertnf{\psi_1} = \bigvee_{i} \exists \vec{x}_i ~ \chi_i$,
$\convertnf{\psi_2} = \bigvee_{i'} \exists \vec{x}'_{i'} ~ \chi'_{i'}$,
and the variables in every $\vec{y}_i$ and $\vec{y}'_{i'}$ are fresh.
\end{itemize}

By Lemma~\ref{lemma:closuresize}, the size of $\clgnf{\convertnf{\psi},N_m}$ is exponential in the size $m$ of $\psi$.
\end{proof}

Now when we apply the automaton construction of Proposition \ref{prop:auttrans}
to the output, we will get an automaton with state
set $\clgnf{\convertnf{\psi},N_{\size{\psi}}}$. By the above, the size of
this is bounded by an  exponential in the size of the original formula~$\psi$.
The size of the automaton alphabet is unaffected by this transformation.
Thus again we can apply Theorem \ref{thm:vardi}
to get a doubly-exponential algorithm for testing satisfiability:

\begin{corollary} \label{thm:gnf} \cite{gnfj}
There is a $\twoexp$ satisfiability testing algorithm for $\gnf$
sentences without equality.
\end{corollary}